\PassOptionsToPackage{table,dvipsnames}{xcolor}
\documentclass[%
 reprint,
superscriptaddress,
 amsmath,amssymb,
 aps,
pra,
]{revtex4-2}

\usepackage[normalem]{ulem}
\usepackage{svg}
\usepackage{xr}
\usepackage{graphicx}%
\usepackage{makecell}
\usepackage{dcolumn}%
\usepackage{tabularx}
\usepackage{multirow}
\usepackage{bm}%
\usepackage{hyperref}
\hypersetup{breaklinks=true, colorlinks=true, linkcolor={blue!90!black}, citecolor={blue!90!black}, urlcolor={blue!90!black}}
\usepackage{braket}
\usepackage{cleveref}
\usepackage{stmaryrd}
\usepackage{subfiles}
\usepackage{bibunits}

\let\oldaddcontentsline\addcontentsline
\newcommand{\stoptocentries}{\renewcommand{\addcontentsline}[3]{}}
\newcommand{\starttocentries}{\let\addcontentsline\oldaddcontentsline}

\newcommand{\QUK}{Quantinuum, Partnership House, Carlisle Place, London SW1P 1BX, UK}

\usepackage{amsthm, amsmath, amsfonts, amssymb}
\usepackage{braket}
\usepackage{graphicx}%
\usepackage{dcolumn}%
\usepackage{bm}%
\usepackage{hyperref}
\usepackage[dvipsnames]{xcolor}

\usepackage{times}
\usepackage{titlesec}
\usepackage{hyperref}
\hypersetup{breaklinks=true, colorlinks=true, linkcolor={blue!90!black}, citecolor={blue!90!black}, urlcolor={blue!90!black}}
\usepackage{cleveref}
\usepackage[utf8]{inputenc}
\usepackage{verbatim}
\usepackage{amsmath}
\usepackage{soul}
\usepackage{booktabs}
\usepackage{xr}
\usepackage{amssymb}
\usepackage{amsthm}
\usepackage{blindtext}
\usepackage{titlesec}
\usepackage{enumitem}
\usepackage{mathtools}
\usepackage{slashed}
\usepackage{comment}
\usepackage{newtxmath}
\usepackage{amsfonts}
\usepackage{makecell}
\usepackage{tabularx}
\usepackage{comment}
\usepackage{epstopdf}
\usepackage{float}
\usepackage{color}
\usepackage{multirow}
\usepackage{physics}
\usepackage{empheq}
\usepackage{bbm}
\usepackage{mdframed}
\usepackage{enumitem}
\usepackage{bbm}

\usepackage[figure,boxed,linesnumbered]{algorithm2e}
\DontPrintSemicolon
\SetKwInOut{Input}{Input}
\SetKwInOut{Output}{Output}
\SetKwComment{Comment}{/* }{ */}
\SetKw{Parameters}{Parameters : }
\SetKwFunction{SingleInvert}{SingleInvert}
\SetKwFunction{ForwardLearn}{ForwardLearn}
\IncMargin{1em}
\let\oldnl\nl
\newcommand{\nlnonumber}{\renewcommand{\nl}{\let\nl\oldnl}}

\makeatletter
\newcommand{\removelatexerror}{\let\@latex@error\@gobble}
\makeatother

\newcommand{\C}{\mathbb{C}}

\newcommand{\backvec}[1]{\reflectbox{$\vec{\reflectbox{$#1$}}$}}
\newcommand{\bb}[1]{\mathbb{#1}}
\renewcommand{\cal}[1]{\mathcal{#1}}

\renewcommand{\rm}[1]{\mathrm{#1}}

\DeclareMathOperator{\poly}{poly}
\DeclareMathOperator{\dist}{dist}

\newtheorem{theorem}{Theorem}

\newtheorem{conjecture}{Conjecture}
\newtheorem{definition}{Definition}
\newtheorem{corollary}{Corollary}
\newtheorem{lemma}{Lemma}

\newmdtheoremenv{protocol_framed}{Protocol}

\titleformat{\paragraph}
{\normalfont\normalsize\bfseries}{\theparagraph}{1em}{}
\titlespacing*{\paragraph}
{0pt}{3.25ex plus 1ex minus .2ex}{1.5ex plus .2ex}

\begin{document}

\title{Digital signatures with classical shadows on near-term quantum computers}%

\author{Pradeep Niroula}
\email{pradeep.niroula@jpmchase.com}
\thanks{Equal contribution}
\affiliation{Global Technology Applied Research, JPMorganChase, New York, NY 10017, USA}
\author{Minzhao~Liu}
\thanks{Equal contribution}
\affiliation{Global Technology Applied Research, JPMorganChase, New York, NY 10017, USA}
\author{Sivaprasad Omanakuttan}
\affiliation{Global Technology Applied Research, JPMorganChase, New York, NY 10017, USA}
\author{David Amaro}
\affiliation{\QUK}
\author{Shouvanik Chakrabarti}
\affiliation{Global Technology Applied Research, JPMorganChase, New York, NY 10017, USA}
\author{Soumik Ghosh}
\affiliation{Department of Computer Science, University of Chicago, Chicago, IL 60637, USA}
\author{Zichang He}
\affiliation{Global Technology Applied Research, JPMorganChase, New York, NY 10017, USA}
\author{Yuwei Jin}
\affiliation{Global Technology Applied Research, JPMorganChase, New York, NY 10017, USA}
\author{Fatih Kaleoglu}
\affiliation{Global Technology Applied Research, JPMorganChase, New York, NY 10017, USA}
\author{Steven Kordonowy}
\affiliation{Global Technology Applied Research, JPMorganChase, New York, NY 10017, USA}
\author{Rohan Kumar}
\affiliation{Global Technology Applied Research, JPMorganChase, New York, NY 10017, USA}
\author{Michael A. Perlin}
\affiliation{Global Technology Applied Research, JPMorganChase, New York, NY 10017, USA}
\author{Akshay Seshadri}
\affiliation{Global Technology Applied Research, JPMorganChase, New York, NY 10017, USA}
\author{Matthew Steinberg}
\affiliation{Global Technology Applied Research, JPMorganChase, New York, NY 10017, USA}
\author{Joseph Sullivan}
\affiliation{Global Technology Applied Research, JPMorganChase, New York, NY 10017, USA}
\author{Jacob~Watkins}
\affiliation{Global Technology Applied Research, JPMorganChase, New York, NY 10017, USA}
\author{Henry Yuen}
\affiliation{Columbia University, New York, NY 10027, USA}
\author{Ruslan~Shaydulin}
\email{ruslan.shaydulin@jpmchase.com}
\affiliation{Global Technology Applied Research, JPMorganChase, New York, NY 10017, USA}

\date{\today}%

\begin{abstract}
Quantum mechanics provides cryptographic primitives whose security is grounded in hardness assumptions independent of those underlying classical cryptography. However, existing proposals require low-noise quantum communication and long-lived quantum memory, capabilities which remain challenging to realize in practice. In this work, we introduce a quantum digital signature scheme that operates with only classical communication, using the classical shadows of states produced by random circuits as public keys. We provide theoretical and numerical evidence supporting the conjectured hardness of learning the private key (the circuit) from the public key (the shadow). A key technical ingredient enabling our scheme is an improved state-certification primitive that achieves higher noise tolerance and lower sample complexity than prior methods. We realize this certification by designing a high-rate error-detecting code tailored to our random-circuit ensemble and experimentally generating shadows for 32-qubit states using circuits with $\geq 80$ logical ($\geq 582$ physical) two-qubit gates, attaining $0.90 \pm 0.01$ fidelity. 
With increased number of measurement samples, our hardware‑demonstrated primitives realize a proof‑of‑principle quantum digital signature, demonstrating the near-term feasibility of our scheme.
\end{abstract}

\maketitle
\stoptocentries
\begin{bibunit}[apsrev4-2]

The study of quantum computing has traditionally focused on reducing computational resources such as time, memory, and energy, resulting in landmark quantum algorithms such as Shor’s algorithm for factorization \cite{shor1999polynomial}, Grover’s algorithm for unstructured search \cite{grover1996fast}, and methods for quantum Hamiltonian simulation \cite{lloyd1996universal, childs2010relationship, berry2015simulating, low2019hamiltonian}. Recently, a new dimension of quantum computational advantage has emerged, based on the development of cryptographic primitives whose security does not depend on the mathematical assumptions underpinning classical cryptography.
Classical digital signatures and related primitives typically rely on one-way functions (OWFs), families of functions that are efficiently computable yet hard to invert on average. RSA instantiates such hardness through integer factorization, which fails against a fault-tolerant quantum computer. Learning-with-Errors forms the foundation of many post-quantum cryptographic (PQC) protocols and is intended to be secure against quantum adversaries, but is a subject of ongoing research. These pressures motivate security assumptions that do not depend on the existence of OWFs.

Quantum mechanics offers independent routes to secure cryptographic schemes. While it has been long known that quantum \emph{communication} can enable unconditionally-secure key distribution \cite{bennett2014quantum}, recently there have been significant advances in using quantum \emph{computation} to propose cryptographic primitives that either rest on alternative assumptions than classical schemes or have no classical counterpart. Notably, some of these primitives are compatible with noisy devices~\cite{2511.03686,Liu2025}, unlike many traditional algorithms that typically require fault tolerance, opening opportunities to strengthen communication, commerce, and national security by grounding security in more fundamental physical assumptions.

\begin{figure*}
    \centering
    \includegraphics[width=\linewidth]{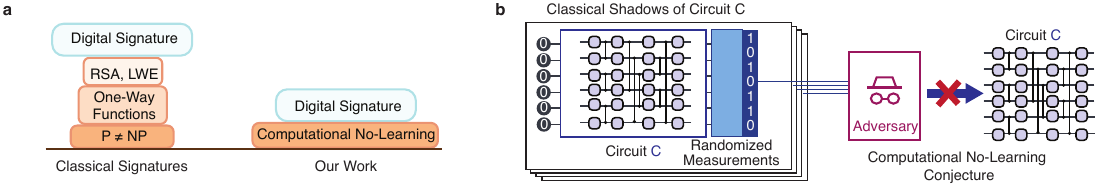}
    \caption{\textbf{Computational no-learning conjecture. a}, Our shadow-based signature schemes uses assumptions independent of those traditionally used in classical cryptography. \textbf{b}, Our conjecture states that it is not feasible for an adversary to learn circuits from randomized measurements (shadows).
    } 
    \label{fig:complexity-hierarchy}
\end{figure*}

One well-explored paradigm leverages the fact that quantum systems cannot be reliably observed without altering them~\cite{portmann2022security}. A more recent approach relies on the difficulty of learning quantum processes or states. In particular, it has been demonstrated that certain cryptographic primitives can be achieved through quantum communication, and could remain secure even in scenarios where OWFs do not exist \cite{kretschmer2021quantum,morimae2024one, fefferman2025hardness,morimae2022quantum,morimae2024one,khurana2024commitments}. However, existing schemes rely on the ability to transmit quantum states between communicating parties, a task that remains challenging with current and near-term technologies.

In this work, we present a digital signatures scheme that does not require OWFs or quantum communication. Digital signatures enable a recipient to verify the message’s true author using a signature generated by the sender. 
Our signature scheme modifies earlier proposals from \cite{fefferman2025hardness} to use classical shadows, instead of quantum states, as public keys, making the scheme friendly to near-term experiments. The security of our scheme is based on the conjectured computational hardness of learning quantum states from their classical shadows (``computational no-learning'' conjecture). This conjecture is independent of the existence of one-way functions as illustrated in Fig.~\ref{fig:complexity-hierarchy}.

More precisely, this protocol implements the quantum cryptographic primitive known as ``one-way puzzles", distinct from OWFs. One-way puzzles have been used \cite{khurana2024commitments} to construct secure quantum multiparty computation \cite{BCQ23} and non-interactive quantum bit commitments \cite{BCQ23,Yan22,HMY23}, providing OWF-free building blocks for quantum cryptography.

At present, no algorithm exists to efficiently learn a quantum circuit given its classical shadows. To provide evidence for the ``computational no-learning" conjecture, we instead consider algorithms requiring a more powerful access model. First, we show that existing algorithms for learning shallow circuits fail in our setting~\cite{fefferman2025hardness, huang2024learning, kim2024learning, landau2025learning, fefferman2024anti}. Second, we derive a learning algorithm specialized to the class of all-to-all circuits considered in this work, extending the approach from Ref.~\cite{fefferman2024anti}, and show security against it. In addition to support for CNL, we also provide evidence that ``spoofing'' the protocol is exponentially harder than verifying the signature. Specifically, we prove the separation for two general classes of potential attacks: low-degree polynomial learning algorithms~\cite{chen2025information} and variational circuit learning~\cite{anschuetz2026critical, anschuetz2025unified, anschuetz2022quantum}.

We demonstrate the near-term feasibility of the proposed scheme by generating shadows on a trapped-ion processor, successfully certifying the fidelity of the prepared quantum state. To enable the experiment, we introduce improved classical shadow schemes with better tolerance to experimental imperfections and a variant of Iceberg error-detecting code~\cite{Steane_iceberg,Gottesman_1998,self2024protecting} tailored to our random circuits. 
Using multi‑block Iceberg encoding, we achieve beyond-break-even performance on circuits with 40 logical qubits and 680 logical gates (560 single‑qubit and 120 two‑qubit), where the encoded fidelity exceeds the fidelity without error detection, a result of independent interest. For hard-to-learn 32-qubit states, we achieve a fidelity of $0.90\pm0.01$ with 17,316 samples.
We show that increasing the number of samples by 19.5 times is sufficient for a proof-of-principle demonstration of a signature scheme, demonstrating that our proposal can be realized in the near term.

\section*{Digital signatures}

\begin{figure*}[!ht]
    \centering
    \includegraphics[width=1\textwidth]{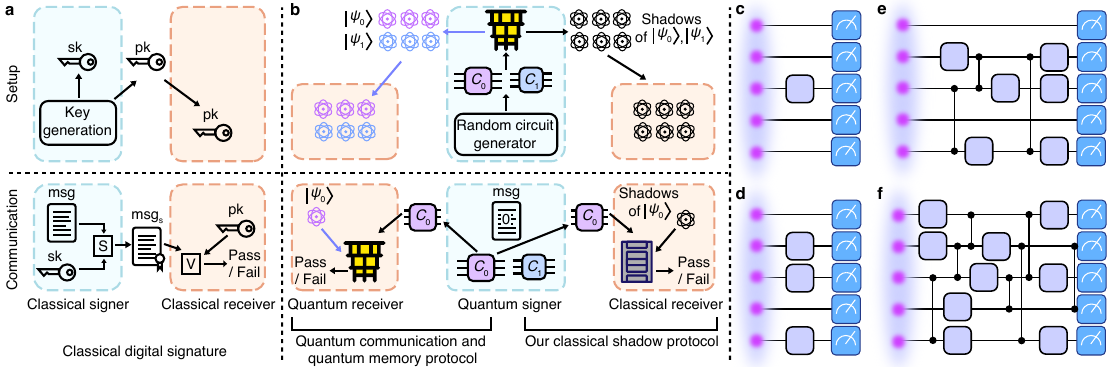}
    \caption{\textbf{Digital signature schemes.} \textbf{a}, The setup and the communication phases of classical digital signature. \textbf{b}, The setup and the communication phases of quantum digital signature. We show an example where the message communicated is a single bit $\mathrm{msg}=0$. Blue arrows denote quantum communication, and black arrows represent classical communication. The left half of the panel shows a protocol requiring quantum communication and quantum memory. The right half of the panel shows our classical shadow protocol. Here, shadows of a quantum state correspond to classical descriptions of measurement outcomes on the state. Verification by the receiver only requires a classical computer but takes time exponential in the number of qubits. \textbf{c,d,e,f,} Classical shadow schemes. \textbf{c}, Single-qubit random Pauli basis protocol of Ref.~\cite{huang2024certifying}, where a randomly chosen qubit is measured in a random Pauli basis and the rest are measured in the computational basis. \textbf{d}, An $m$-level shadow overlap protocol of Ref.~\cite{huang2024certifying}, where $m=3$ qubits are randomly chosen and measured in random Pauli bases, and the rest are measured in the computational basis. We improve the multi-qubit protocol by modifying the post-processing step. \textbf{e}, An $m$-level protocol, where $m=3$ qubits are randomly chosen and an $m$-qubit random Clifford circuit is applied to them before all qubits are measured in the computational basis. This protocol is proposed by this work. \textbf{f}, A special case of the protocol in \textbf{e} where $m$ is the same as the number of qubits. A global random Clifford circuit is applied to all qubits before measurement. This is the same as the classical shadow protocol based on random Clifford measurements for predicting quantum fidelities \cite{huang2020predicting}.%
    }
    \label{fig:protocols}
\end{figure*}

In a classical digital signature scheme, a sender cryptographically ``signs" a message such that a receiver may verify that the message truly originated from the sender and could not have been sent by an impersonator. To do this, the signer produces a key pair consisting of a public and a secret key. The public key $\mathrm{pk}$ is broadcast while the secret key remains private to the sender. When the sender wishes to sign a message $\rm{msg}$, they use the secret key $\mathrm{sk}$ to generate the signed message $\rm{msg}_{\rm{s}}$~$=S(\mathrm{sk}, \rm{msg})$, where $S$ denotes the signing function. The signed message $\rm{msg}_{\rm s}$ is then sent to the receiver, who uses the sender's public key and a verification function $V(\mathrm{pk}, \rm{msg}_{\rm s})$ to check if the signed message $\rm{msg}_{\rm s}$ is valid, i.e., corresponds to the sender. This protocol is illustrated in Fig. \ref{fig:protocols}\textbf{a}.

The security of such a protocol is based on the fact that, without the knowledge of the secret key $\mathrm{sk}$, it is hard for an adversary to alter the contents of message $\rm{msg}_{\rm s}$ and still pass the verification by the receiver. To enable this, it must be difficult to invert the signing function $S$ and reconstruct the secret key $\mathrm{sk}$, even if an adversary were to intercept both $\rm{msg}$ and $\rm{msg}_{\rm s}$. The non-invertibility of the signing function $S$ and therefore the existence of OWFs are central to a secure classical signature scheme. 

Many quantum cryptographic protocols are based on quantum analogues of OWFs whose existence is, in a formal sense, independent of classical OWFs. As a notable example friendly to near-term implementation, Ref.~\cite{fefferman2025hardness} constructs one-way state generators \cite{morimae2022quantum} using random quantum circuits, a class of circuits that have been extensively studied both theoretically and numerically. Specifically, it is easy to generate (on a quantum computer) a quantum state given the classical description of a randomly sampled quantum circuit as input, whereas it is hard in general to infer the circuit description given copies of the quantum state alone. This hinges on the so-called ``computational no-learning" conjecture that forbids efficient learning of random quantum circuits from copies of the quantum state. Although proving the conjecture outright would require significant breakthroughs in computational complexity theory, there is extensive evidence in favor of the conjecture \cite{fefferman2025hardness}.

Such a one-way state generator may also be leveraged to devise a quantum analog of a digital signature \cite{fefferman2025hardness}. In the quantum digital signature protocol proposed by Ref.~\cite{fefferman2025hardness}, to sign a single bit, the sender begins with classical descriptions of two quantum circuits $C_0$ and $C_1$, which serve as the secret key. For the public key, the sender distributes copies of the corresponding quantum states, $\ket{C_0} = C_0 \ket{0}$ and $\ket{C_1} = C_1 \ket{0}$, for example via a trusted cloud repository. To transmit the bit $b$, the sender reveals the signed message $(b, C_b)$ to the receiver. The receiver then downloads copies of the quantum state $\ket{C_b}$ from the cloud and verifies that the state matches the circuit $C_b$—for instance, by checking that $C_b^\dagger \ket{C_b}$ yields the all-zero state. This protocol is illustrated in the left half of Fig. \ref{fig:protocols}\textbf{b}.

For an adversary to alter the signed message, i.e., flip the signed bit from $b$ to its negation $\neg b$, it must output a message $(\neg b,D)$, without prior knowledge of $C_{\neg b}$ (which is private to the sender), such that $\ket{D} \approx \ket{C_{\neg b}}$. The security of this protocol is based on the assumption that learning an equivalent circuit $D \sim C_{\neg b}$ is hard, even if the adversary possesses copies of the public key $\ket{C_{\neg b}}$.

However, while both the generation and verification steps are efficient with a quantum computer, the protocol relies on the coherent storage and distribution of quantum public keys—a requirement that far exceeds current technological capabilities. Establishing a cloud repository for quantum signatures also introduces new challenges. Although digital signatures protect against man-in-the-middle attacks on messages, they do not prevent an adversary from impersonating the cloud server itself. In classical digital signature schemes, this risk is mitigated through the use of digital certificates: users perform a one-time download of a trusted list that associates public keys with their legitimate owners. For quantum signatures, additional mechanisms must be developed to safeguard the integrity of public keys, or new trust assumptions must be introduced. Likewise, unclonability of quantum states limits the broadcasting and publishing of quantum public keys. It is therefore desirable for the public keys to be classical.

We propose a protocol where the public keys are not quantum states but classical shadows. In particular, for secret keys $(C_0, C_1)$, the sender publishes polynomial-length classical shadows $\mathcal{S}(\ket{C_0}), \mathcal{S}(\ket{C_1})$ as public keys. Here, $\mathcal{S}$ denotes the quantum procedure that, across a polynomial number of samples, selects random local measurement bases on the qubits, measures the state accordingly, and for each sample records the chosen bases together with the corresponding measurement outcomes; the collection of these recorded basis choices and outcomes is the classical shadow.
To communicate a bit $b$, the signer releases $(b, C_b)$ as before. The receiver verifies the message by certifying that $\mathcal{S}(C\vert0\rangle)\approx \mathcal{S}(C_b\vert0\rangle)$ using the shadow overlap protocol defined below, from which the receiver concludes that $C\vert0\rangle \approx C_b \vert0\rangle$ \cite{huang2024certifying, gupta2025few}. Assuming the hardness of learning circuits from shadows, the verifier concludes that $C$ indeed came from the sender's private key. The protocol is illustrated in the right half of Fig. \ref{fig:protocols}\textbf{b}. To facilitate the analysis of our scheme, we introduce a general compiler that converts any shadow overlap protocol into a digital signature protocol (Supplementary Information (SI) Sec.~\ref{sec:protocol}).

\section*{Evidence for Security}

For our shadow-based digital signatures protocol to be secure, it is crucial that learning the circuit $C_i$, given public classical shadows, is infeasible. In the protocol where the public keys are instead quantum states, the assumption of hardness is formalized by the computational no-learning (CNL) from states conjecture. This conjecture precludes the existence of efficient algorithms to learn certain ensembles of quantum circuits given copies of the output state \cite{fefferman2025hardness}. In a similar spirit, we introduce the ``computational no-learning from shadows" conjecture, which asserts the hardness of learning quantum circuits given polynomial-length shadows. This new conjecture is a strictly weaker 
assumption (which is desirable) than the original CNL because, given polynomially many copies of an output state $C\ket{0}$, polynomial-length shadows can be obtained by simple measurements.
As a result, the hardness of learning from shadows follows directly from the conjectured hardness of learning from states.

Protocol security requires the existence of a family of random quantum circuits for which CNL holds. We provide evidence for this by showing in SI Sec.~\ref{sec:learning_alg} that no existing algorithms~\cite{huang2024learning,landau2025learning, fefferman2024anti, kim2024learning} satisfy even two out of the three major requirements necessary to falsify CNL for the all-to-all circuits used in our experiment, namely (i) ``proper learning'' (i.e., recovering a circuit with the exact prescribed architecture), (ii) support for arbitrary circuit geometries, and (iii) reliance only on shadows and not on stronger forms of access to the circuit that are not reasonably available to an adversary. 

To provide further evidence for the security of our protocol, we design a learning algorithm that is proper yet effective for all-to-all shallow circuits (the type of circuits in our experiment), extending the protocol from Ref.~\cite{fefferman2024anti} (SI Sec.~\ref{subsec:new_learning_alg}). However, this algorithm cannot be directly extended to satisfy all three requirements since it relies on a stronger access model. Furthermore, we show that even under this stronger access model, we can design a family of circuits secure against our new attack. The central idea is that known approaches for learning circuits (including ours) proceed by inverting the gates one-by-one. Such local-inversion strategies are thwarted if the causal ``lightcone" saturates the entire computational register, which we show can be achieved with shallow all-to-all circuits used in our experiment. Combined, this evidence indicates that major algorithmic breakthroughs are required to falsify CNL and compromise the security of our protocol.

The arguments above provide theoretical and empirical evidence against the efficient learning of private circuits from public shadows. Since, for shadow-based signatures, the verification step -- wherein a verifier validates a signature -- is also expected to take exponential time, it would be desirable to choose circuit families such that verification can be done in a feasible wall-time while spoofing remains prohibitively expensive. We expect such a ``gap" between verification time and learning or spoofing time to arise since circuit learning involves a search from the space of all possible circuits in the ensemble, which is as large as $\exp(O(nd))$ where $d$ is the depth of the ensemble. In this work, we call a protocol secure if the time of verification $T_V$ and the time of learning the circuits $T_L$ satisfy $T_L/T_V \geq \Omega(\exp(n))$, where $n$, the circuit size, serves as the security parameter. In other words, there is an \textit{exponential} gap between the cost to cheat and the cost to verify. In SI Sec.~\ref{sec:further_gap}, we formalize this conjecture and provide evidence towards it. 

The first part of our evidence derives from the study of so-called ``low-degree" learning algorithms~\cite{chen2025information}, whose outputs can be expressed as low-degree polynomials over the input. In our case, the inputs to the adversary are the classical shadows that form the public key. Since the shadows are obtained by local measurements of quantum states, this setting is an instance of the low-degree algorithms for quantum state learning discussed in Ref.~\cite{chen2025information}, where it was shown that this family encompasses all known state-of-the-art learning algorithms. 
In particular, it has been posited~\cite[Hypothesis 2.1.5]{hopkins2018statistical} that for most statistical learning problems, the class of algorithms of degree $\le d$ is no less powerful than the class of algorithms that take time $O(\exp(d))$. Under this heuristic, if a learning algorithm requires a degree $n^{1+\delta}$ for $\delta > 0$, it would be expected to take time $T_L = O(\exp(n^{1+\delta}))$, which would be enough for an exponential gap since verification takes time $T_V = \Theta(2^n)$. We show that we can indeed construct circuit ensembles that form ``approximate designs'' with respect to the Haar measure, such that the degree of a learning algorithm that takes as input non-adaptive local measurements on quantum states (such as the classical shadows that form the public key) has to be at least $n^{1+\delta}$ for $\delta > 0$. Such a choice of circuit ensemble would give us the desired exponential gap between spoofing and learning.

The runtime-degree correspondence discussed above, while applicable to a wide variety of algorithms, does not, however, capture \textit{all} potential attacks. In particular, it may fail for settings where the best known algorithm is a greedy procedure such as gradient descent. Hence, the argument above cannot be applied to adaptive learning strategies that optimize a parameterized quantum circuit~\cite{anschuetz2026critical, anschuetz2025unified, anschuetz2022quantum}.
As the second part of our evidence, we show that our circuits exhibit exponentially-many poor local minima, making any variational training difficult (See SI Sec.~\ref{sec:variational}). In particular, such variational algorithms would have to compute the overlap with the target circuit exponentially many times, where each overlap computation incurs an exponential cost. Therefore, spoofing a circuit via variational training is also exponentially harder than verification. This exponential gap between learning and verification establishes the security of our protocol.

\section*{Shadow overlap}

The digital signature scheme sketched out on the right side of Fig.~\ref{fig:protocols}b requires certifying that a revealed circuit $C$ prepares a state close to the honest state $\ket{\psi_{\mathrm{hon}}} = C_b\ket{0}$, that is, $C\ket{0} \approx \ket{\psi_{\mathrm{hon}}}$. Given classical shadows of a noisy laboratory state $\sigma_{\mathrm{lab}}$, we estimate the overlap with any target pure state $\ket{\psi}$ specified by a circuit $C_{\psi}$.

We build on classical shadow overlap certification while introducing modifications tailored to our setting. To fix notation, we briefly recall the Pauli–based baseline of Ref.~\cite{huang2024certifying} (illustrated in Fig.~\ref{fig:protocols}\textbf{c,d}). That scheme performs $T$ measurements. On each shot, it selects a subset $\bm{k}$ of $m$ qubits, measures those $m$ qubits in random single–qubit Pauli bases $\bm{P} \in \{X,Y,Z\}^m$, measures the remaining $n-m$ qubits in the $Z$ basis, and records the bitstring $\bm{z} \in \{0,1\}^n$. This produces the shadow dataset
\begin{equation}
\bigl\{\, \bm{k}^{(i)},\; \bm{P}^{(i)} \in \{X,Y,Z\}^m,\; \bm{z}^{(i)} \,\bigr\}_{i \in [T]} .
\label{eq:pauli-shadows}
\end{equation}
Using these shadows together with a polynomial number of target amplitudes $\psi(x) \equiv \langle x \vert C_{\psi} \vert 0 \rangle$, one computes a score $\omega_i$ for each $i$. The expectation $\mathbb{E}[\omega]$ is the shadow overlap, which controls the fidelity between $\ket{\psi}$ and $\sigma_{\mathrm{lab}}$~\cite{huang2024certifying}.

Before introducing our modifications, we summarize the verification complexity and the role of the ``relaxation time''. In this framework, two standard implications connect shadow overlap and fidelity for a fixed hypothesis state $\ket{\psi}$: if $\mathbb{E}[\omega] \ge 1 - \varepsilon$ then $\bra{\psi}\sigma_{\mathrm{lab}}\ket{\psi} \ge 1 - \tau \varepsilon$, and if $\bra{\psi}\sigma_{\mathrm{lab}}\ket{\psi} \ge 1 - \varepsilon$ then $\mathbb{E}[\omega] \ge 1 - \varepsilon$. Here $\tau \ge 1$ is the relaxation time of a Markov chain induced by the chosen measurement rule and operator used in post–processing, and depends on the protocol and the circuit ensemble. Computing the overlap statistic requires a polynomial number of amplitudes of $C_{\psi}\ket{0}$, so exact verification scales as $O(2^n)$ up to polynomial factors in $T$ and $n$. Large $\tau$ and large estimator variance increase the number of measurements, $T$, needed to avoid false accepts and false rejects; see Fig.~\ref{fig:discrimination}.

Motivated by these considerations, our protocol adopts the same overall structure but modifies both the measurement rule and the estimator to reduce $\tau$ and the variance. For shadow generation, on each shot we apply a random $m$-qubit Clifford circuit to a uniformly random subset of $m$ qubits, then measure all qubits in the computational basis (Fig.~\ref{fig:protocols}\textbf{e}). This yields shadows
\begin{equation}
\bigl\{\, \bm{k}^{(i)},\; U^{(i)} \in \mathrm{Clifford}[m],\; \bm{z}^{(i)} \,\bigr\}_{i \in [T]} .
\label{eq:clifford-shadows}
\end{equation}
In post–processing, we use a refined overlap estimator that achieves the minimal possible relaxation time $\tau=1$ in the limit $m=n$. Together, these choices define a tunable family of shadow–based state–certification protocols that interpolate between two well–studied extremes: global random Clifford measurements, which are most sample efficient but require global control (Fig.~\ref{fig:protocols}\textbf{f}; Ref.~\cite{huang2020predicting}), and single–qubit random Pauli measurements, which are less sample efficient but only require local control (Fig.~\ref{fig:protocols}\textbf{c,d}; Ref.~\cite{huang2024certifying}). Details appear in SI Sec.~\ref{sec:shadow}.

We now state the honest and adversarial overlap bounds used in verification, focusing on the depolarizing model and deferring general noise to SI Sec.~\ref{subsec:security_of_single_bit_signatures}. Let $\phi_{\mathrm{hon}} \equiv \langle\psi_{\mathrm{hon}}|\sigma_{\mathrm{lab}}|\psi_{\mathrm{hon}}\rangle = 1 - \varepsilon_{\mathrm{hon}}$ denote the honest signer’s fidelity. By the overlap–fidelity implication, this gives the honest lower bound
$$
\mathbb{E}[\omega_{\mathrm{hon}}] \ge 1 - \varepsilon_{\mathrm{hon}}.
$$
For an adversary, the computational no–learning assumption limits the state overlap with the honest target to $|\langle\psi_{\mathrm{adv}}|\psi_{\mathrm{hon}}\rangle|^2 \le 1 - \varepsilon_{\mathrm{CNL}}$. Under depolarizing noise, Lemma~\ref{lem:depolarizing_noise} (SI Sec.~\ref{subsec:security_of_single_bit_signatures}) implies an upper bound on the adversarial fidelity to the laboratory state,
$$
\phi_{\mathrm{adv}} \equiv \langle\psi_{\mathrm{adv}}|\sigma_{\mathrm{lab}}|\psi_{\mathrm{adv}}\rangle \le 1 - \bigl(\varepsilon_{\mathrm{CNL}} - \varepsilon_{\mathrm{hon}}\bigr) + \varepsilon_{\mathrm{hon}}\,2^{-n}.
$$
Applying the contrapositive of the overlap–fidelity implication then yields an upper bound on the adversarial overlap,
$$
\mathbb{E}[\omega_{\mathrm{adv}}] \le 1 - \frac{\varepsilon_{\mathrm{CNL}} - \varepsilon_{\mathrm{hon}}}{\tau} + \frac{\varepsilon_{\mathrm{hon}}}{2^{n}\tau} \approx 1 - \frac{\varepsilon_{\mathrm{CNL}} - \varepsilon_{\mathrm{hon}}}{\tau},
$$
where $\tau$ is the relaxation time associated with the chosen hypothesis and protocol. The $2^{-n}$ term is negligible at the system sizes considered here. Together, these bounds imply a gap between the honest and adversarial overlap distributions whenever $\varepsilon_{\mathrm{hon}}$ is small and $\varepsilon_{\mathrm{CNL}}$ is bounded away from zero. Reducing $\tau$ and the estimator variance, as achieved by our modified measurement rule and estimator, enlarges this gap at practical $T$, lowering the sample complexity and increasing robustness to experimental noise in signature verification.

\begin{figure}[!t]
    \centering
    \includegraphics[width=1\linewidth]{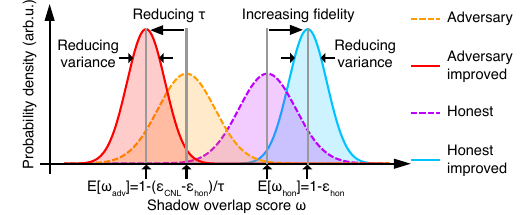}
    \caption{\textbf{Discriminating between the honest signer and the adversary}. To improve the discriminating power of the shadow overlap, we devise an improved protocol that reduces the relaxation time $\tau$ to shift the adversary distribution to the left and reduces the variance of the score. Additionally, we use quantum error detection to increase the fidelity to shift the honest distribution to the right.}
    \label{fig:discrimination}
\end{figure}

\section*{Experimental Demonstration}

To demonstrate near-term feasibility, we implement shadow-based state certification on a trapped-ion processor using the random Clifford shadow protocol at level $m=4$ (Protocol~\ref{prot:using_lower_bound_procedure} in SI Sec.~\ref{sec:shadow}). Certification must operate in a regime where the honest and adversarial shadow overlaps are separated, i.e. $\mathbb{E}[\omega_{\mathrm{adv}}] < \mathbb{E}[\omega_{\mathrm{hon}}]$, which implies an upper bound on tolerable experimental infidelity in public-key generation $\varepsilon_{\mathrm{hon}} < \varepsilon_{\mathrm{CNL}}/(\tau+1)$. We construct a circuit family with low depth that yields causal light cones obstructing learning algorithms based on gate-by-gate local inversion described in SI Sec.~\ref{sec:learning_alg}. For our circuits at $n=32$ and $m=4$, the relaxation time is $\tau \approx 8$ (SI Sec.~\ref{sec:supp_improved_shadow}). Assuming that learning a circuit with overlap at least $1-\varepsilon_{\mathrm{CNL}}=0.01$ is hard, one needs an experimental fidelity $\phi_{\mathrm{hon}}=1-\varepsilon_{\mathrm{hon}}\geq0.89$. Table~\ref{tab:qed_improvement} shows that this fidelity is out of reach without encoding. 

To overcome this hardware limitation, we introduce two constructions tailored to our random-circuit ensemble and hardware constraints. The first is a multi-block Iceberg encoding that partitions the $k$ logical qubits into several independent Iceberg blocks, each with its own stabilizer checks. Blocks are coupled only through transversal CZ gates, which limits the spread of errors between blocks (Methods and Fig.~\ref{fig:experimental_circuit}). The introduction of multiple blocks increases the number and reduces the weight of syndrome measurements, thereby reducing idling overhead during high-weight checks. The second is a gauge-fixing scheme for the Iceberg code that produces additional ancilla qubits, allowing parallel execution of single-qubit logical layers and introducing additional low-weight checks that are less sensitive to errors (SI Sec.~\ref{sec:qed}). As shown in \Cref{tab:qed_improvement}, the use of multiple code blocks and gauge fixing enables beyond-break-even performance with only a modest sampling overhead.

We experimentally certify a four-block-encoded 32-qubit hard-to-learn random circuit with logical two-qubit gate depth 5 (Fig.~\ref{fig:experimental_circuit}). Using Protocol~\ref{prot:using_lower_bound_procedure} at $m=4$, we apply random four-qubit Clifford circuits after decoding and before measurement. The total physical two-qubit gate count per shot is 582. A total of 44{,}560 shots are executed. After discarding trials with detected errors, 17{,}316 classical shadows are retained, yielding an shadow overlap of $0.91$ and $\phi_{\mathrm{hon}}=0.90\pm0.01$ (see SI Sec.~\ref{sec:qed_results}). This fidelity is lower than the fidelities estimated in Table \ref{tab:qed_improvement} due to the use of random Clifford circuits which require additional gates compared to the random Paulis in Table~\ref{tab:qed_improvement}.

Under $\varepsilon_{\mathrm{hon}}<0.11$ and $1-\varepsilon_{\mathrm{CNL}}\leq0.01$, the expected adversary shadow overlap is $\mathbb{E}[\omega_{\mathrm{adv}}]\leq0.89$. The protocol parameters are such that, in the malicious case, no adversarial circuit can pass our protocol with probability higher than $\delta=0.15$ (SI Theorem~\ref{thm:lower_bound_protocol_zhang}). On the other hand, our experimentally observed scores $\{\omega_{\mathrm{exp},i}\}_{i\in[T]}$ give shadow overlap $\mathbb{E}[\omega_{\mathrm{exp}}]=0.91$. The data passes the protocol and is consistent with the honest case.

\begin{table}
    \centering
    \begin{tabular}{|c|c|c|c|c|c|c|}
        \hline
        Type & \makecell{Logical\\qubits} & Blocks & Shots & $p$ & \makecell{2-Q\\depth} & \makecell{Fidelity\\$\phi_{\mathrm{hon}}$} \\
        \hline
        Physical & 32 & N/A & 500 & 1 & 5 & $0.82(7)$ \\
        \hline
        Iceberg & 32 & 2 & 2,000
        & 0.37(1)  %
        & 128 & $0.79(6)$ \\
        \hline
        Iceberg & 32 & 4 & 2,000
        & 0.39(1)  %
        & 75 & $0.98(5)$ \\
        \hline
        Iceberg+ & 32 & 2 & 2,000
        & 0.38(1)  %
        & 88 & $0.94(5)$ \\
        \hline
        Physical & 40 & N/A & 700 & 1 & 6 & $0.82(6)$ \\ 
        \hline
        Iceberg & 40 & 2 & 3,000
        & 0.204(7)  %
        & 177 & $0.63(7)$ \\
        \hline
        Iceberg & 40 & 4 & 3,000
        & 0.226(8)  %
        & 98 & $0.94(6)$ \\
        \hline
        Iceberg+ & 40 & 2 & 3,000
        & 0.162(7)  %
        & 108 & $0.75(7)$ \\
        \hline
    \end{tabular}
   \caption{\textbf{Estimated fidelity for random circuits.} 
   The number of logical qubits is 32 and 40, and the logical two-qubit gate depth is 5 and 6 respectively. All data is taken from Quantinuum H2 trapped-ion processor with 56 physical qubits~\cite{Moses2023}.
    ``Physical" refers to running the circuit without encoding into the Iceberg code.
    ``Iceberg" means that the circuit is encoded into a specified number of Iceberg code blocks.
    ``Iceberg+" indicates that the circuit is encoded into augmented Iceberg code blocks, in which $l=2$ logical qubits (per code block) are reserved for parallelizing logical operations.
    $p$ is the probability of acceptance after post-selection based on error detection.
    2-Q depth is the physical two-qubit gate depth required to implement the circuit.
    Fidelity $\phi_{\mathrm{hon}}$ is estimated from shadow overlap as described in SI Sec.~\ref{sec:qed_results}.
    The reduction in physical two-qubit depth for Iceberg+ over Iceberg (with two code blocks) improves the estimated fidelity, as memory and idling errors are significant in current hardware. Uncertainties show standard error.}
    \label{tab:qed_improvement}
\end{table}

\begin{figure*}[!ht]
    \centering
    \includegraphics[width=1\textwidth]{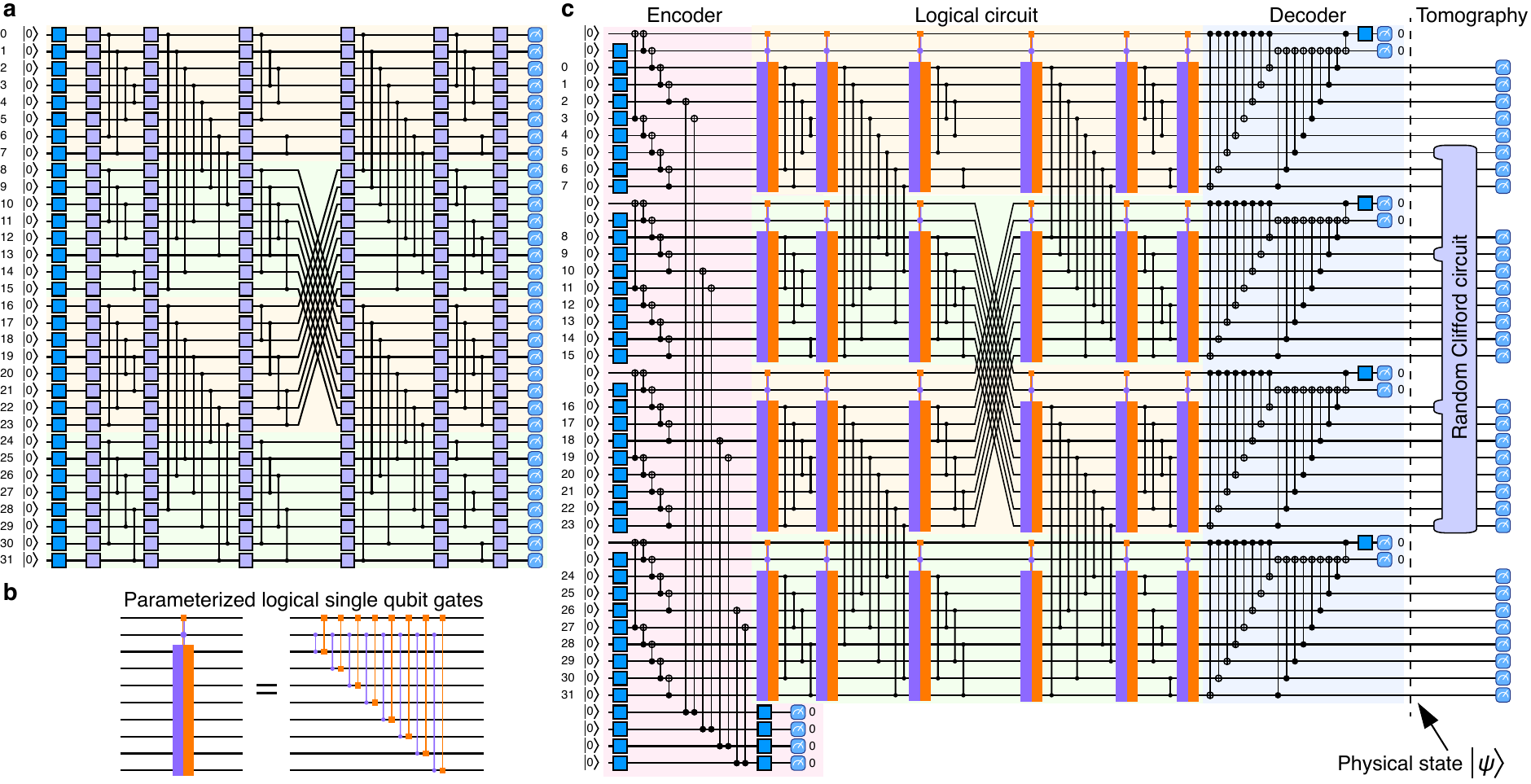}
    \caption{\textbf{Quantum circuits used in the experiment.} \textbf{a}, The circuit without encoding. The circuit has four blocks.
    Blue gates are Hadamard gates, and purple gates are parameterized single-qubit gates $Z^p X^{1/2}$ for $p\in\{-1,-\frac{3}{4},\dots,\frac{3}{4}\}$. We interleave layers of random pairs of inner-block CZ gates and random pairs of transversal intra-block CZ gates.
    \textbf{b}, One layer of logical single-qubit gates. Each logical single-qubit $R_{Z}$ gate is implemented using a physical parameterized $R_{ZZ}$ gate (purple gates with circular ends), and each logical $R_{X}$ gate is implemented using a physical $R_{XX}$ gate (orange gates with square ends). \textbf{c}, The circuit with Iceberg code encoding. The encoder prepares the logical $\vert{+}\rangle$ state with post selection.
    All logical two-qubit gates are implemented by physical two-qubit gates. 
    The decoder transforms the encoded state into the physical state with post selection. 
    Then, tomography on the decoded physical state is performed by collecting random Clifford classical shadows.
    The random Clifford circuit is applied to a random set of physical four qubits after decoding. 
    }
    \label{fig:experimental_circuit}
\end{figure*}

Simply running the certification experiment for longer would allow our scheme to achieve nontrivial security for the digital signature. Corollary \ref{cor:sample_complexity_soundness} (SI Sec.~\ref{sec:supp_cert_experiment}) shows that the sample complexity is proportional to $\ln(1/\delta)$. Therefore, decreasing $\delta$ from $\textcolor{black}{0.15}$ to $10^{-15}$ requires $\ln(1/10^{-15})/\ln(1/\textcolor{black}{0.15})=19.5$ times more samples, resulting in a total of $340$k samples. Consider a computationally bounded adversary that tries 10 billion adversary circuits and checks if each one passes the protocol. The probability that it spoofs the signature protocol is only $10^{-5}$. To do this, it computes $10^{10}\times340{,}000=3.4\times10^{15}$ values of $\omega$. The computation of each $\omega$ takes about $0.3$ seconds on an AMD EPYC 7R13 processor using our implementation. Using the same implementation, an adversary that is one million times more powerful still requires 32 years.
While future improvements to learning algorithms may make be possible to learn our 32-qubit circuit, a modest increase in circuit size would restore the security of the scheme.

\section*{Outlook}

The protocol described above only signs a single bit. In the future, it would be interesting to demonstrate multi-bit quantum digital signatures. In SI Sec.~\ref{sec:protocol}, we present a $k$-bit message protocol with a favorable verification cost. Instead of verifying every bit and paying a total verification cost of $k\cdot O(2^{n})$, we can pay a cost independent of $k$ by verifying a constant number of bits chosen randomly. If all verified bits pass, then it is guaranteed that most bits are not altered by the adversary. If the message is encoded in an error-correcting code, the message can be recovered even if some bits are altered. Using an error-correcting code for multi-bit message digital signatures gives a signer complexity of $k\cdot\mathrm{poly}(n)$, a receiver complexity of $O(1)\cdot O(2^{n})$, and a spoofing complexity of $O(k)\cdot \Omega(2^{n})$. Even for long messages with large $k$, signing remains feasible since each signature can be produced efficiently. In turn, verification is realistic via spot checks on a subset of signatures, while spoofing remains computationally prohibitive.

The equivalence between our protocol, which realizes one-time signatures, and one-way puzzles allows one to construct many other quantum cryptographic primitives without assuming the existence of OWFs (see SI Sec.~\ref{sec:microcrypt}). Our work is a step towards realizing this broad class of primitives, opening up new avenues of quantum advantage realizable with near-term quantum computers. %

\putbib[citations]

\clearpage

\section*{Methods}
\label{sec:methods}

\begin{figure*}[!ht]
    \centering
    \includegraphics[width=\linewidth]{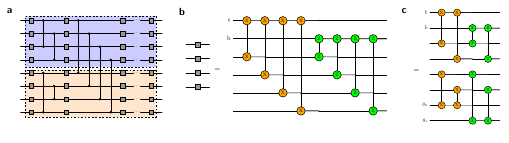}
    \caption{\textbf{Multi-block random circuit and encoding in Iceberg code.}
    \textbf{a}, Logical random circuit is composed of blocks of qubits, with each block encoded in a separate Iceberg code block. 
    \textbf{b}, Single-qubit logical gates in Iceberg error-detection code: logical $R_{X}$ and $R_{Z}$ rotations are implemented with a physical two-qubit interaction involving the “top” qubit and the “bottom” qubit respectively. 
    \textbf{c}, Promoting select logical operators to stabilizers (``gauge fixing'') enables multiple logical single-qubit rotations to be performed in parallel. 
    The qubits gauge-fixed to eigenstates of logical $x$ and logical $z$ are denoted $g_x$ and $g_z$ respectively.
}
    \label{fig:random_circuit_in_iceberg}
\end{figure*}

\subsection{Multi-Block Iceberg Encoding of Random Circuits}

The success of quantum state certification and the overall digital signature protocol relies crucially on achieving high fidelity. To improve fidelity, we use a multi-block random circuit whose structure is chosen to improve the performance of the Iceberg error-detecting code.  

The Iceberg code~\cite{Steane1996, Gottesman1998, self2024protecting} is a $\llbracket k+2,k,2\rrbracket$ error-detecting code that encodes $k$ logical qubits into $k+2$ physical qubits. It is a stabilizer code with two weight-$(k+2)$ stabilizers ($Z$ and $X$) and is capable of detecting any single-qubit error. Trapped-ion devices like the Quantinuum H2 race-track processor~\cite{Moses2023} used in this work are particularly suitable for this code due to the ability to implement high-weight stabilizers using atom movement.

The structure of random circuit used in this work is summarized in Fig.~\ref{fig:random_circuit_in_iceberg}\textbf{a}. The circuit consists of repeating sequence of four layers: (i) single-qubit Clifford gates sampled independently and uniformly for each qubit, (ii) two-qubit controlled-Z (CZ) gates on random pairs of qubits within block, (iii) another layer of single-qubit gates, and (iv) two-qubit inter-block CZ gates. We remark that similar structure has been used in Ref.~\cite{bluvstein2024logical} for IQP circuits encoded in the 3D color code. When encoding the random circuit in Iceberg code, each ``block'' of qubits is encoded within a separate Iceberg code block, with inter-block CZ gates applied transversally. We measure all stabilizers at the end of the circuit and post-select on the outcome.

\subsection{Parallelization in Iceberg Code by Gauge Fixing}

The random circuit we use has more single-qubit than two-qubit gates. This creates a challenge for the Iceberg code since each logical single-qubit gate is implemented with a physical two-qubit gate involving either ``top'' (for $R_{X}$ rotation) or ``bottom'' ($R_{Z}$) ancilla qubit (see Fig.~\ref{fig:random_circuit_in_iceberg}\textbf{b}). Additionally, the high weight of the Iceberg code stabilizers leads to memory error accumulation during atom movement necessary to read out the stabilizer~\cite{he2025performance}.

We overcome both of these challenges using a stabilizer extension procedure that we call \emph{gauge fixing}, in which a selected subset of logical operators (such as $\overline{Z}_j$ or $\overline{X}_j$) is added to the stabilizer group. This projects the code space onto the joint $+1$ eigenspace of the chosen operators, thereby restricting the logical subspace and increasing the number of physical qubits to encode a fixed number of logical qubits.
The resulting code has parameters $\llbracket k + \ell + 2, k, 2 \rrbracket$, where $k$ is the number of computational logical qubits and $\ell$ is the number of logical operators promoted to stabilizers.
See the SI Sec.~\ref{sec:qed} for a full mathematical treatment and further discussion.

Promoting additional logical operators to stabilizers introduces ``proxy'' top or bottom qubits, allowing multiple logical rotations to be performed in parallel as shown in Fig.~\ref{fig:random_circuit_in_iceberg}\textbf{c}. Furthermore, the stabilizer checks corresponding to the introduced small stabilizers can be executed with much smaller atom movement overhead as compared to the stabilizers of the standard Iceberg code.
For $l=2$, we have four stabilizers per block and each of them are measured during each readout. 
Only runs in which all the stabilizer outcomes are $+1$ are retained for post processing.

\section*{Acknowledgments}
We thank Kaushik Chakraborty for helpful discussions of digital signature schemes.
We thank Bill Fefferman for discussions on the security of a shadow-based digital signature protocol.
We thank Anthony Armenakas, Soorya Rethinasamy, and Cameron Foreman for helpful feedback on the manuscript.
We thank Brian Neyenhuis and Zachary Massa for their assistance with the experimental implementation of the Clifford Shadow protocol on the Quantinuum devices.
We thank Rob Otter for the executive support of the work and invaluable feedback on this project. We thank Vas Rajan for helpful feedback on the project. We are grateful to Jamie Dimon, Lori Beer, and Scot Baldry for their executive support of JPMorganChase’s
Global Technology Applied Research Center and our work in quantum computing.
We thank the technical staff at JPMorganChase’s Global Technology Applied Research for their support. H.Y. is supported by AFOSR award FA9550-23-1-0363, NSF awards CCF2530159, CCF-2144219, and CCF-2329939, and the Sloan Foundation.

\section*{Author Contributions}
P.N., S.C., R.S. conceptualized the project. P.N., S.C. devised the single-bit digital signature protocol from shadow overlap protocols and conducted exploratory analysis. M.L. devised the improved shadow overlap protocol and the multi-bit signature protocol, proved the security of single-bit and multi-bit protocols, and provided low-degree polynomial hardness evidence. J.S., J.W. investigated classical error detection codes for multi-bit signatures.
S.O., Y.J., J.S., M.A.P., M.S., Z.H., R.K. devised the quantum error detection protocol and compatible circuit ansatz. 
D.A. devised the parallelization technique shown in Figure~\ref{fig:random_circuit_in_iceberg}~\textbf{c}.
S.O., Y.J., P.N. conducted the experiment. P.N. and A.S. analyzed the data. P.N., S.C., S.K. analyzed the variational learning adversary. J.W. and S.K. developed and analyzed the gate-by-gate proper learning adversary. F.K. established the connection to one-way puzzles and other quantum cryptographic primitives. All authors contributed to the writing of the manuscript and shaping of the project.

\section*{Competing Interests}

The authors declare no competing interests.

\section*{Additional Information}
Supplementary Information is available for this paper.

\section*{Disclaimer}
This paper was prepared for informational purposes with contributions from the Global Technology Applied Research center of JPMorgan Chase \& Co. This paper is not a product of the Research Department of JPMorgan Chase \& Co. or its affiliates. Neither JPMorgan Chase \& Co. nor any of its affiliates makes any explicit or implied representation or warranty and none of them accept any liability in connection with this paper, including, without limitation, with respect to the completeness, accuracy, or reliability of the information contained herein and the potential legal, compliance, tax, or accounting effects thereof. This document is not intended as investment research or investment advice, or as a recommendation, offer, or solicitation for the purchase or sale of any security, financial instrument, financial product or service, or to be used in any way for evaluating the merits of participating in any transaction.

\end{bibunit}
\clearpage
\newpage
\starttocentries

\begin{bibunit}[apsrev4-2]

\onecolumngrid

\begin{center}
    \textbf{\large Supplementary Material for: Digital signatures with classical shadows on near-term quantum computers}
\end{center}

\tableofcontents
\newpage

\section{Digital signature from shadow overlap protocols}\label{sec:protocol}
A theoretical contribution of our work is a general compiler that converts any shadow overlap protocol into a digital signature protocol. A key ingredient of our digital signature protocol is the ability to verify if a given classical shadow corresponds to a quantum circuit.  Shadow overlap protocols are used to certify that classical shadows of a ``laboratory" quantum state (i.e., published shadows) is indeed close to a given hypothesis state (i.e., state generated by the quantum circuit), where closeness is defined with respect to state fidelity. To start, we define general shadow overlap protocols.

\begin{definition}[Shadow Overlap Protocol]\label{def:shadow_overlap_protocol}
    Let $\ket{\psi}$ be a fixed hypothesis state. A non-adaptive, shadow overlap quantum state fidelity certification protocol (``shadow overlap protocol") consists of a data collection phase and a certification phase. In the data collection phase, $T$ copies of a laboratory state $\rho$ are manipulated and measured to produce $T$ independent shadows $\{o_i\}_{i\in[T]}$. In the certification phase, using the shadows $\{o_i\}_{i\in[T]}$ and any known information about the hypothesis state, with probability at least $1-\delta$, the protocol will correctly output \textsc{Failed} if the fidelity between the hypothesis state and the laboratory state is low, i.e. $\langle\psi\vert\rho\vert\psi\rangle<1-\varepsilon$, and will correctly output \textsc{Certified} if the fidelity is high, i.e. $\langle\psi\vert\rho\vert\psi\rangle\geq1-{\varepsilon}/{2\tau}$, where $\tau$ is a function that depends on $\vert\psi\rangle$. We call this a $(T, \delta, \varepsilon,\tau(\cdot))$ protocol, emphasizing that $\tau(\cdot)$ is a function.
\end{definition}
\noindent The parameter $\tau$ is the relaxation time for a certain Markov chain generated by the distribution of local measurements of the state $\ket{\psi}$. Its functional relationship with the hypothesis state $\ket{\psi}$ is a key factor in the performance of the shadow-based certification, and hence we highlight this dependence in our parametrization. Protocol \ref{prot:random_pauli_improved} and Protocol \ref{prot:random_clifford_improved} in this work, as well as Protocol 2 in \cite{huang2024certifying}, are examples of shadow overlap protocols. In this work, each shadow $o_i$ is a tuple of a) a uniformly random subset of $l \subset [n]$, b) a random Clifford operation acting on the qubits in $l$, and c) the computational basis measurement outcome of all $n$ qubits. 

Now, we formally define the digital signature protocol.
\begin{protocol_framed}[Digital Signature of a single bit via Shadow Overlap]\label{prot:single_bit_digital_signature}\quad 
\normalfont
\newline
\begin{itemize}
    \item $\mathsf{SKGen}$: Sample descriptions of quantum circuits $C^{b}$ uniformly at random from some ensemble $\mathcal{C}$, for both $b = 0,1$. Set $\mathsf{sk}=\left( C^b\right)_{b\in \{0, 1\}}$
    \item $\mathsf{PKGen}$: Perform the data collection phase of some $(T, \delta, \varepsilon,\tau(\cdot))$ non-adaptive shadow overlap quantum state fidelity certification protocol on the state $\rho$ obtained by applying $C^b$ to the all zero state for $b\in\{0,1\}$. The public key is $\mathsf{pk} = \{o_i^b\}_{i \in [T], b \in \{0, 1\}}$.
    \item  $\mathsf{Sign}$: To sign bit $b$, output the bit and the classical description of the circuit $m_{\rm s}=(b, C^b)$ as the signed message.
    \item $\mathsf{Ver}$: Use the certification phase of the non-adaptive shadow overlap-type quantum state fidelity certification protocol. Specifically, use $\{o_i\}_{i \in [T]}$ that corresponds to the signed bit $b$ in $\mathsf{pk}$ and certify the state $C^b\vert 0\rangle$. Output the same output as the fidelity certification protocol.
\end{itemize}
\end{protocol_framed}
\noindent The security of the protocol relies on the assumption that it is hard for an adversary to learn a circuit $C'$ in the circuit ensemble $\cal{C}$ given classical shadows $\{o^b_i\}_{i\in[T]}$ generated using a secret key $C^b$. The shadow overlap protocol certifies that the unknown quantum state producing the shadows $\{o^b_i\}_{i\in[T]}$ is close to the state generated by the circuit revealed by the signer in the $\mathsf{Sign}$ phase of the protocol. Therefore, for an adversary to pass the protocol, it must learn the circuit $C^b$ (or more precisely, must learn a circuit $\tilde{C}$ such that $\tilde{C}\ket{0} \equiv C^b \ket{0}$. If the circuit ensemble is such that learning $C^b$ (or equivalent circuits) is difficult from the shadows, the protocol is hard to spoof if the adversary does not know $C^b$ a priori.

The above digital signature protocol certifies a single bit. Longer messages can be certified, in principle, by iterating the protocol over the entire message. This, however, is highly inefficient. Single-bit certification with shadows is already exponential time (as it involves calculation of quantum state amplitudes), and thus one would like to minimize the number of required calls to this protocol.

Thankfully, we can substantially reduce the number of certifications using classical error detection. %
Take some block code with parameters $[M, k, d]$, where a logical message of $k$ bits is encoded into an $M$-bit codespace, and the minimum Hamming distance between valid code words is at least $d$. Certification will now only take place on some subset of the bits. Below we give the updated protocol in Protocol \ref{prot:multi_bit_digital_signature} for a $k$-bit logical message.

\parbox{0.96\textwidth}{\begin{protocol_framed}[Signature for a $k$-bit message]\label{prot:multi_bit_digital_signature}\quad 
\normalfont
\newline
\begin{flushleft}\underline{Input:}\end{flushleft}
\begin{itemize}
    \item An $[M, k, d]$ error detection code.
    \item $\mathsf{Enc}:\{0,1\}^k\rightarrow\{0,1\}^M$: a function that encodes a $k$-bit message into $M$ bits.
    \item $\mathsf{Chk}:\{0,1\}^M\rightarrow\{0,1\}$: a boolean function that outputs $1$ if the input is not in the code space (i.e. an indicator function). 
    \item $\mathsf{Dec}:\{0,1\}^M\rightarrow\{0,1\}^k$: a function that decodes the $M$-bit encoded message into the original $k$-bit message.
    \item $V$: an integer parameter that determines how many circuits and classical shadows to check.
\end{itemize}
\begin{flushleft}\underline{Procedure:}\end{flushleft}
 \vspace{-1.5em}
\begin{itemize}
    \item $\mathsf{SKGen}$: Sample descriptions of quantum circuits $C^{b,j}$ uniformly at random from some ensemble $\mathcal{C}$, for $b = 0,1$ and $j\in[M]$. Set $\mathsf{sk}=\left( C^{b,j}\right)_{{b\in \{0, 1\}},j\in[M]}$
    
    \item $\mathsf{PKGen}$: Perform the data collection phase of some $(T, \delta, \varepsilon,\tau(\cdot))$ non-adaptive shadow overlap-type quantum state fidelity certification protocol on the state $\rho$ obtained by applying $C^{b,j}$ to the all zero state subject to device infidelity for $b\in\{0,1\}$ and $j\in[M]$. The public key is $\mathsf{pk} = \{o_i^{b,j}\}_{i \in [T], b \in \{0, 1\}, j\in[M]}$.
    
    \item  $\mathsf{Sign}$: To sign the message $m=(m^1,\dots,m^k)\in\{0, 1\}^k$, first encode it into encoded message $m_{\rm e}=(m_{\rm e}^1,\dots, m_{\rm e}^M)=\mathsf{Enc}(m)$. Then, output $m_{\rm s}=\left(m_{\rm e},\{C^{m_{\rm e}^j,j}\}_{j\in[M]}\right)$ as the signed message.
    
    \item $\mathsf{Ver}$:
    \begin{itemize}
        \item To verify $m_{\rm s}$, first check $\mathsf{Chk}(m_{\rm e})=0$, otherwise abort.
        \item Randomly select a subset $\mathcal{V}\subset [M]$ satisfying $\vert\mathcal{V}\vert=V$.
        \item For all $j\in\mathcal{V}$, use the certification phase of the shadow overlap protocol. Specifically, use $\{o_i^{m_{\rm e}^j, j}\}_{i \in [T]}$ as the classical shadows and the target state $C^{m_{\rm e}^j,j}\vert 0\rangle$.
        \item Abort if shadow overlap protocol outputs \textsc{Failed} for any $j\in\mathcal{V}$.
        \item Decode the message by outputting $\mathsf{Dec}(m_{\rm e})$.
    \end{itemize}
\end{itemize}
\end{protocol_framed}}

\noindent To see why error detection helps, assume learning from classical shadows is intractable, and thus an adversary $E$ is prevented from forging a signature by providing the circuit secret key. Then $E$ must resort to changing a subset of bits and hoping they are not certified. To change the underlying message, the error correcting code forces $E$ to modify at least $d$ bits (the receiver will throw away any non-codeword). If $d = d(M)$ scales sufficiently with code size, then these changes can be detected with high probability by, say, random sampling of bits on the codeword. The details of this argument are carried out in~\Cref{subsec:security_of_multi_bit_signatures}.

\subsection{Security of single-bit signatures} \label{subsec:security_of_single_bit_signatures}

We begin our security analysis with single bit messages. For our digital signatures protocol to be secure, it is necessary that a computationally-bounded adversary cannot obtain quantum circuits (within the circuit class of the protocol) that pass the shadow-based certification, in a time faster than certification itself. Thus, the critical assumption is the inability to efficiently learn circuits from classical shadows.
\begin{conjecture}[Computational no-learning from shadows]\label{conj:nlfcs}
    Let $\varepsilon, \eta \in [0,1]$ be fixed constants, and let $\cal{C}$ be a family of unitary quantum circuits. The family $\cal{C}$ is said to satisfy $(\varepsilon,\eta)$-Computational No-Learning from Shadows conjecture if the following holds with respect to any polynomial-time quantum adversary $\cal{A}$: given polynomial-sized classical shadows of some state $\ket{C} = C \ket{0}$ for $C \in \cal{C}$, $\cal{A}$ cannot output a classical description of some $D \in \cal{C}$ such that $\abs{\braket{C}{D}}  = \abs{\bra{0}C^\dagger D\ket{0}} \geq 1 - \varepsilon$ with probability at least $1 - \eta$.
\end{conjecture}
\noindent The task of learning quantum circuits from shadows has received relatively little attention compared to the related task of learning circuits from \textit{quantum states} and from black-box queries to the circuit unitary~\cite{landau2025learning,huang2024learning,fefferman2024anti,kim2024learning}. Nonetheless, hardness results of the latter serve as evidence towards the former. For example, the no-learning from shadows assumption is implied (and inspired) by ``no-learning from states" conjecture as introduced in~\cite{fefferman2025hardness}.
\begin{conjecture}[Computational no-learning from states, Conjecture 1.1 of~\cite{fefferman2025hardness}] \label{conj:nlfqs}
    Let $\varepsilon, \eta \in [0,1]$ be fixed constants, and let $\cal{C}$ be a family of unitary quantum circuits. The family $\cal{C}$ is said to satisfy $(\varepsilon,\eta)$-Computational No-Learning from State conjecture if, given $T=\poly(n)$ copies of an $n$-qubit quantum state $\ket{C} = C\ket{0}$, a polynomial-time quantum adversary $\cal{A}$ cannot produce a $\poly(n)$-length classical representation of a circuit $D$ such that  $\abs{\braket{C}{D}} = \abs{\bra{0}C^\dagger D \ket{0}} \geq 1-\varepsilon$ with probability at least $1-\eta$.
\end{conjecture}
\noindent This conjecture is strictly stronger than Conjecture ~\ref{conj:nlfcs}. To see that, note that if one can learn circuits from shadows using an algorithm $A$, one can also learn the same circuit from copies of a state by first measuring the shadows, and then using $A$ thereafter. 
Ref. \cite{fefferman2025hardness} provides several arguments in support of Conjecture~\ref{conj:nlfqs}. Direct proof of the conjecture is obstructed: proving it for depth $d = \log^2 n$ brickwork Haar-random circuits would imply $\mathrm{P}\neq\mathrm{PSPACE}$, amounting to a mathematical breakthrough. However, examining the best current circuit learning algorithms suggest the hardness of learning sufficiently beyond $\log$-depth, such as with depth $\log^2 n$. In fact, such hardness has been proven against the family of low-degree polynomial algorithms, which encompasses a large class of algorithms. Furthermore, the hardness of RingLWE used by post-quantum cryptography also implies the hardness of learning circuits, although assuming such hardness would defeat the purpose of avoiding assumptions used in classical cryptography. Finally, it is also shown that the no-learning conjecture is equivalent to assuming the existence of random circuit one-way state generators.

Below, we formally instantiate the computational no-learning from states ($\mathrm{CNL}$) conjecture for a set of concrete circuit ensemble, adversary, and parameters.
\begin{definition}
We say the conjecture $\mathrm{CNL}_{\mathcal{C},T,\mathcal{A},\varepsilon,\eta}$ holds if the following is true. Given $T$ copies of the state $\vert\psi\rangle$ such that $\vert\psi\rangle=C\vert0\rangle$ for some $C\in\mathcal{C}$, with probability at least $1-\eta$ the adversary $\mathcal{A}$ cannot learn a circuit $C'\in\mathcal{C}$ such that $\vert\langle0\vert C'^\dagger C\vert0\rangle\vert^2\geq1-\varepsilon$.
\end{definition}

\begin{theorem}\label{thm:noiseless_single_bit_security}
If $\mathrm{CNL}_{\mathcal{C},T,\mathcal{A},\varepsilon_{\mathrm{CNL}},\eta}$ holds and $\tau(C\vert 0\rangle)\leq\tau^*$ for all $C\in\mathcal{C}$, then, with probability at least $1-\vert\mathcal{C}\vert\delta-\eta$, an adversary that does not have access to $C^b$ fails Protocol \ref{prot:single_bit_digital_signature} instantiated with a $(T,\delta,\varepsilon_{\mathrm{CNL}},\tau(\cdot))$ non-adaptive shadow overlap protocol. Additionally, an honest signer with fidelity at least $1-\varepsilon_{\mathrm{CNL}}/(2\tau^*)$ passes it with probability at least $1-\delta$.
\end{theorem}
\begin{proof}
If $\mathrm{CNL}_{\mathcal{C},T,\mathcal{A},\varepsilon_{\mathrm{CNL}},\eta}$ holds, then with probability $1-\eta$, any circuit $C'$ that $\mathcal{A}$ returns must satisfy $\vert\langle0\vert C'^\dagger C\vert0\rangle\vert^2<1-\varepsilon_{\mathrm{CNL}}$. To show this, we use proof by contradiction. The only information $\cal{A}$ has access to are the shadows for the circuit $C^b$. If $\cal{A}$ can learn $C'$ such that $\vert\langle0\vert C'^\dagger C\vert0\rangle\vert^2\geq1-\varepsilon_{\mathrm{CNL}}$ just from the shadows, then it must also be able to learn $C'$ from copies of the state since it can measure the states and produce the shadows.

Denote as $\mathcal{W}$ the set of circuits $C'$ that satisfy $\vert\langle0\vert C'^\dagger C\vert0\rangle\vert^2<1-\varepsilon_{\mathrm{CNL}}$, and we have $\vert\mathcal{W}\vert<\vert\mathcal{C}\vert$. Denote the output of $\mathsf{Ver}$ when verifying on $\mathsf{pk}$, $b$, and $C^b$ as $\mathsf{Ver}(\mathsf{pk}, b, C^b)$. For any $C'$ such that $\vert\langle0\vert C'^\dagger C\vert0\rangle\vert^2<1-\varepsilon_{\mathrm{CNL}}$, we have 
\begin{align}
\Pr_{\mathsf{pk\sim\mathsf{Sign}}}[\mathsf{Ver}(\mathsf{pk}, b, C')=\textsc{Failed}]\geq1-\delta,
\end{align}
which is guaranteed by the shadow overlap protocol since $\tau(C'\vert0\rangle)\;\leq\tau^*$. Here, the hypothesis state is $\vert\psi\rangle=C'\vert0\rangle$, and the state to certify is the lab state prepared by noisy $C^b$.

The probability that no state the adversary can output leads to false certification is
\begin{align}
&\Pr_{\mathsf{pk\sim\mathsf{Sign}}}[(\mathsf{Ver}(\mathsf{pk}, b, C')=\textsc{Failed})\text{ for all }C'\in\mathcal{W}]\\
\geq&1-\sum_{C'\in\mathcal{W}}\Pr_{\mathsf{pk\sim\mathsf{Sign}}}[(\mathsf{Ver}(\mathsf{pk}, b, C')=\textsc{Certified})] \geq 1-\vert\mathcal{C}\vert\delta.
\end{align}
Overall, this bounds the malicious-case false positive probability to $\vert\mathcal{C}\vert\delta+\eta$. The honest-case false negative probability (the probability of a honest signer signer failing the test) is simply $\delta$.
\end{proof}
Note that since $\vert\mathcal{C}\vert$ is exponential in $n$, we need to make $\delta$ exponentially small in $n$. However, since the sample complexity (size of classical shadows) scales with $\ln(1/\delta)$, this only increases sample complexity by a $\poly(n)$ factor. 

The above derivation assumes that the classical shadows are collected for the ideal pure state $C\vert 0\rangle$. In reality, a mixed laboratory state $\sigma_{\mathrm{lab}}$ is prepared to generate classical shadows with fidelity $\phi_{\mathrm{hon}}\equiv1-\varepsilon_{\mathrm{hon}}\equiv\langle 0\vert C^\dagger\sigma_{\mathrm{lab}}C\vert 0\rangle$. In this case, we need to place a minimum assumed honest case fidelity $\phi_{\mathrm{hon}}$ and find the corresponding upper bound on $\vert\langle0\vert C'^{\dagger}\sigma_{\mathrm{lab}}C'\vert0\rangle\vert^2$ which is no longer bounded by simply $\vert\langle0\vert C'^{\dagger}C'\vert0\rangle\vert^2\leq1-\varepsilon_{\mathrm{CNL}}$. Instead, the bound is given by the following two theorems.

\begin{lemma}\label{lem:depolarizing_noise}
Under depolarizing noise,
\begin{align}
\vert\langle0\vert C'^{\dagger}\sigma_{\mathrm{lab}}C'\vert0\rangle\vert^2\leq1-(\varepsilon_{\mathrm{CNL}}-\varepsilon_{\mathrm{hon}})+\varepsilon_{\mathrm{hon}}\frac{1}{2^n}.
\end{align}
\end{lemma}
\begin{proof}
Under the depolarizing noise,
\begin{align}
\sigma_{\mathrm{lab}}=(1-\varepsilon_{\mathrm{hon}})C\vert0\rangle\langle0\vert C^\dagger+\varepsilon_{\mathrm{hon}}\frac{I}{2^n}.
\end{align}
Therefore,
\begin{align}
\vert\langle0\vert C'^{\dagger}\sigma_{\mathrm{lab}}C'\vert0\rangle\vert^2&=(1-\varepsilon_{\mathrm{hon}})\vert\langle0\vert C'^{\dagger}C'\vert0\rangle\vert^2+\varepsilon_{\mathrm{hon}}\frac{1}{2^n}\\
&\leq(1-\varepsilon_{\mathrm{hon}})(1-\varepsilon_{\mathrm{CNL}})+\varepsilon_{\mathrm{hon}}\frac{1}{2^n}\\
&\leq1-(\varepsilon_{\mathrm{CNL}}-\varepsilon_{\mathrm{hon}})+\varepsilon_{\mathrm{hon}}\frac{1}{2^n}.
\end{align}
\end{proof}

\begin{lemma}\label{lem:general_noise}
Under general noise,
\begin{align}
\vert\langle0\vert C'^{\dagger}\sigma_{\mathrm{lab}}C'\vert0\rangle\vert^2=\sup_{0\leq\varepsilon'_{\mathrm{CNL}}\leq\varepsilon_{\mathrm{CNL}}}\left(\sqrt{(1-\varepsilon_{\mathrm{hon}})(1-\varepsilon_{\mathrm{CNL}}')}+\sqrt{\varepsilon_{\mathrm{hon}}\varepsilon_{\mathrm{CNL}}'}\right)^2.
\end{align}
\end{lemma}
\begin{proof}
Denote $\vert\psi_{\mathrm{hon}}\rangle=C\vert0\rangle$ and $\vert\psi_{\mathrm{adv}}\rangle=C'\vert0\rangle$. For any mixed state $\sigma_{\mathrm{lab}}$ such that
\begin{align}
1-\varepsilon_{\mathrm{hon}}\equiv\langle\psi_{\mathrm{hon}}\vert\sigma_{\mathrm{lab}}\vert\psi_{\mathrm{hon}}\rangle,
\end{align}
we can write it as a sum of pure states
\begin{align}
\sigma_{\mathrm{lab}}=\sum_iw_i\vert\psi_i\rangle\langle\psi_i\vert
\end{align}
such that
\begin{align}
1-\varepsilon_i&\equiv\vert\langle\psi_{\mathrm{hon}}\vert\psi_i\rangle\vert^2\\
\sum_iw_i(1-\varepsilon_i)&=1-\varepsilon_{\mathrm{hon}},
\end{align}
can we can write the pure states as
\begin{align}
\vert\psi_i\rangle=\sqrt{1-\varepsilon_i}\vert\psi_{\mathrm{hon}}\rangle+\sqrt{\varepsilon_i}\vert\psi_{\mathrm{hon},i}^\bot\rangle
\end{align}
for some normalized $\vert\psi_{\mathrm{hon},i}^\bot\rangle$ orthogonal to $\vert\psi_{\mathrm{hon}}\rangle$.

Similarly, for any adversary state $\vert\psi_{\mathrm{adv}}\rangle$ such that
\begin{align}
\vert\langle\psi_{\mathrm{adv}}\vert\psi_{\mathrm{hon}}\rangle\vert^2=1-\varepsilon'_{\mathrm{CNL}}\leq1-\varepsilon_{\mathrm{CNL}},
\end{align}
we can write it as
\begin{align}
\vert\psi_{\mathrm{adv}}\rangle=\sqrt{1-\varepsilon'_{\mathrm{CNL}}}\vert\psi_{\mathrm{hon}}\rangle+\sqrt{\varepsilon'_{\mathrm{CNL}}}\vert\psi_{\mathrm{hon},\mathrm{adv}}^\bot\rangle
\end{align}
for some normalized $\vert\psi_{\mathrm{hon},\mathrm{adv}}^\bot\rangle$ orthogonal to $\vert\psi_{\mathrm{hon}}\rangle$.

Therefore,
\begin{align}
\vert\langle\psi_{\mathrm{adv}}\vert\psi_i\rangle\vert^2&=\left(\sqrt{(1-\varepsilon_{i})(1-\varepsilon'_{\mathrm{CNL}})}+\sqrt{\varepsilon_{i}\varepsilon'_{\mathrm{CNL}}}\langle\psi_{\mathrm{hon},\mathrm{adv}}^\bot\vert\psi_{\mathrm{hon},i}^\bot\rangle\right)^2\\
&\leq\left(\sqrt{(1-\varepsilon_{i})(1-\varepsilon'_{\mathrm{CNL}})}+\sqrt{\varepsilon_{i}\varepsilon'_{\mathrm{CNL}}}\right)^2\\
&\equiv f(1-\varepsilon_i).
\end{align}
Further, one can easily check that for $0\leq\varepsilon_i,\varepsilon_{\mathrm{CNL}}\leq 1$, we have $f''(1-\varepsilon_i)<0$ and $f$ is concave.

With this, we can write
\begin{align}
\langle\psi_{\mathrm{adv}}\vert\sigma_{\mathrm{lab}}\vert\psi_{\mathrm{adv}}\rangle&=\sum_i w_i\vert\langle\psi_{\mathrm{adv}}\vert\psi_i\rangle\vert^2\\
&=\sum_i w_i f(1-\varepsilon_i)\\
&\leq f\left(\sum_i w_i (1-\varepsilon_i)\right)\\
&=f(1-\varepsilon_{\mathrm{hon}})\\
&=\left(\sqrt{(1-\varepsilon_{\mathrm{hon}})(1-\varepsilon'_{\mathrm{CNL}})}+\sqrt{\varepsilon_{\mathrm{hon}}\varepsilon'_{\mathrm{CNL}}}\right)^2.
\end{align}

However, since we do not know the fidelity $\vert\langle\psi_{\mathrm{adv}}\vert\psi_{\mathrm{hon}}\rangle\vert^2=1-\varepsilon'_{\mathrm{CNL}}$ other than the fact that $\varepsilon'_{\mathrm{CNL}}\geq\varepsilon_{\mathrm{CNL}}$ and that $f$ is not necessarily monotonic in $\varepsilon'_{\mathrm{CNL}}$, we have to optimize $\varepsilon'_{\mathrm{CNL}}$ over the allowed interval to find the upper bound.
\end{proof}

With this, we can have an error aware version of the security guarantee.
\begin{corollary}
Set $\varepsilon_{\mathrm{noisy}}=\varepsilon_{\mathrm{CNL}}-\varepsilon_{\mathrm{hon}}$ under depolarizing noise, or
\begin{align}
\varepsilon_{\mathrm{noisy}}=1-\sup_{0\leq\varepsilon'_{\mathrm{CNL}}\leq\varepsilon_{\mathrm{CNL}}}\left(\sqrt{(1-\varepsilon_{\mathrm{hon}})(1-\varepsilon_{\mathrm{CNL}}')}+\sqrt{\varepsilon_{\mathrm{hon}}\varepsilon_{\mathrm{CNL}}'}\right)^2
\end{align}
under general noise. If $\mathrm{CNL}_{\mathcal{C},T,\mathcal{A},\varepsilon_{\mathrm{CNL}},\eta}$ holds and $\tau(C\vert 0\rangle)\leq\tau^*$ for all $C\in\mathcal{C}$, then with probability at least $1-\vert\mathcal{C}\vert\delta-\eta$, an adversary that does not have access to $C^b$ fails Protocol \ref{prot:single_bit_digital_signature} instantiated with a $(T,\delta,\varepsilon_{\mathrm{noisy}},\tau(\cdot))$ non-adaptive shadow overlap protocol. Additionally, an honest signer with fidelity at least $1-\frac{\varepsilon_{\mathrm{noisy}}}{2\tau^*}$ passes it with probability at least $1-\delta$.
\end{corollary}
\begin{proof}
The proof is similar to Theorem \ref{thm:noiseless_single_bit_security}. For all circuits $C'$ such that $\vert\langle0\vert C'^\dagger C\vert0\rangle\vert^2\leq1-\varepsilon_{\mathrm{CNL}}$, we have
\begin{align}
\vert\langle0\vert C'^{\dagger}\sigma_{\mathrm{lab}}C'\vert0\rangle\vert^2\leq 1-\varepsilon_{\mathrm{noisy}}
\end{align}
by Lemma \ref{lem:depolarizing_noise} and Lemma \ref{lem:general_noise}. Then, the shadow overlap protocol is going to output failed with probability at least $1-\delta$ for such circuits. A similar union bound for the probability is applied for all circuits in the ensemble. The honest case is similarly trivial.
\end{proof}

\subsection{Security of multi-bit signatures} \label{subsec:security_of_multi_bit_signatures}
In our proposal for a multi-bit signature, the sender encodes a $k$-bit message into a $M$-bit codeword using a classical error detection code $[M, k, d]$ of distance $d$. Since the generation of signature is efficient, the sender signs all $M$ bits of the encoded message. The receiver, on the other hand, only verifies the signature of a small fraction $\alpha M$ of the message received. Specifically, the receiver accepts the signature if and only if a) the received message is a valid encoding per the error detection code and b) a random subset of $\alpha M$ bits pass the signature verification. For any adversary to alter the signed message, it would have to ``spoof" or alter at least $d$ bits among the $M$-bit long message (otherwise, the received message is no longer a valid encoding). If $d$ bits of the message were forged by an adversary, by uniformly spot-checking a small fraction $V = \alpha M$ of the bits, the receiver would be able to detect any tampering of the original message.

We begin by showing that the probability the verifier fails to detect at least one tampered bit decreases exponentially with $V$, the number of bits whose signature is verified.
\begin{lemma}\label{lem:FP_FN_prob}
Consider an $M$-bit encoded message and suppose the verifier checks the signature of $V$ bits uniformly at random. If an altered bit is checked, the verifier outputs \textsc{Failed} with probability at least $1-\delta_1$. If an honest bit is checked, the verifier outputs \textsc{Certified} with probability at least $1-\delta_2$. Denote as $\Omega_{\mathrm{FP}}$ the false-positive event where an adversary alters at least $\alpha M$ bits but none of the $V$ signature certifications output \textsc{Failed}. Additionally, denote as $\Omega_{\mathrm{FN}}$ the false-negative event where at least one of the $V$ signature certification outputs \textsc{Failed} but there is no adversarial tampering (no bits were truly altered) (False Negative). Then
\begin{align}
\Pr[\Omega_{\mathrm{FP}}]\leq e^{-\alpha V}+\delta_1, \qquad \text{and} \qquad 
\Pr[\Omega_{\mathrm{FN}}]\leq \delta_2 V.
\end{align}
\end{lemma}
\begin{proof}
Let $\Omega_{\mathrm{adv}}$ denote the event that an adversary alters at most $\alpha M$ bits and let $n_{\text{altered}}$ denote the number of altered bits in the randomly chosen validation set. The probability that the validation set will contain none of the ``altered" bits is:
\begin{align}
\Pr[n_{\mathrm{altered}}=0\vert\Omega_{\mathrm{adv}}]&\leq\frac{(1-\alpha)M}{M}\times\frac{(1-\alpha)M-1}{M-1}\times\cdots\times\frac{(1-\alpha) M-V+1}{M-V+1} \leq (1-\alpha)^V\leq e^{-\alpha V}.
\end{align}
On the other hand, if $n_{\mathrm{altered}}\geq1$, then the probability that at least one of the $V$ signature certification outputs \textsc{Failed} is at least $1-\delta_1$. Denote the number of \textsc{Failed} signatures as $n_{\textsc{Failed}}$. Overall,
\begin{align}
\Pr[\Omega_{\mathrm{FP}}]=& \Pr[n_{\mathrm{altered}}=0\vert\Omega_{\mathrm{adv}}]+\Pr[n_{\textsc{Failed}}=0\vert n_{\mathrm{altered}}\geq 0\wedge\Omega_{\mathrm{adv}}]\cdot\Pr[n_{\mathrm{altered}}\geq0\vert\Omega_{\mathrm{adv}}] \leq e^{-\alpha V}+\delta_1.
\end{align}

In the event $\Omega_{\mathrm{honest}}$ where there is no adversary (no bits altered), the probability that at least one signature outputs \textsc{Failed} is
\begin{align}
\Pr[\Omega_{\mathrm{FN}}]=\Pr[n_{\textsc{Failed}}\geq1\vert\Omega_{\mathrm{honest}}]=1-(1-\delta_2)^V \leq \delta_2 V.
\end{align}
\end{proof}

\begin{corollary}\label{cor:FP_FN_requirements}
For a message with $M$ encoded bits, to achieve $\Pr[\Omega_{\mathrm{FP}}]\leq\varepsilon$ and $\Pr[\Omega_{\mathrm{FN}}]\leq\varepsilon$ under an adversary that alters at least $\alpha M$ bits,
having a single-bit digital signature protocol with $\delta_1\leq\frac{\varepsilon}{2},\delta_2\leq \frac{\varepsilon}{V}$ and a verification fraction $V\geq \frac{\ln(2/\varepsilon)}{\alpha}$ is sufficient.
\end{corollary}
\begin{proof}
By Lemma \ref{lem:FP_FN_prob}, we require $\delta_2\leq \frac{\varepsilon}{V}$ by the false negative condition. Similarly, by the false positive condition, we require $\varepsilon\leq e^{-\alpha V}+\delta_1$. If $\delta_1\leq\varepsilon/2$, then
$\varepsilon/2\leq e^{-\alpha V}\Rightarrow V\geq {\ln(2/\varepsilon)}/{\alpha}$.
\end{proof}

Finally, we show Protocol \ref{prot:multi_bit_digital_signature} is secure with constant verification under reasonable assumptions.

\begin{theorem} \label{thm:signatures_ecc}
    If there exists a linear $[M, k, d]$ error detection code with relative distance $d_r = d/M$, the multi-bit signature Protocol 2 has false positive and false negative probabilities of at most $\varepsilon$. In particular, it requires $2M$ shadow-circuit pairs to use as public/secret keys, $V = \max\{2, d_r^{-1}\ln(2/\varepsilon)\}$ of which are checked by the verifier. Each signature check must output \textsc{Failed} on an altered bit with probability $1-\delta_1 \geq 1-\varepsilon/2$, and it must output $\textsc{Certified}$ on an honest bit with probability $1-\delta_2 \geq 1-\varepsilon/2$.
\end{theorem}
\begin{proof}
 If the adversary alters fewer than $d$ bits, the tampering can be detected using the error detection code. If an adversary tampers at least $d$ bits, by Corollary \ref{cor:FP_FN_requirements}, such an adversary may be detected with  $\delta_1\leq\frac{\varepsilon}{2},\delta_2\leq\frac{\varepsilon}{V}$ provided that $V \geq \frac{\ln(2/\varepsilon)}{d_rM}$. Taking $V$ at least 2 ensures that the confidence bound holds even for $\epsilon$ near 1.
\end{proof}

We saw how good error correcting codes could be used in the digital signatures protocol to avoid certifying a large fraction of the message. The key properties we would like in such a code are constant rate and relative distance as the block size increases. Furthermore, we would naturally like efficient encoding and decoding. 
Thankfully, there exists a restricted ensemble which retains good performance and has structure which allows for efficient decoding. Introduced by Robert Gallager in a 1963 PhD thesis~\cite{gallager1962low}, low density parity check (LDPC) codes are classical linear codes defined by a parity check matrix $H$ which is sparse, in the sense that (a) few variables are assigned to each check and (b) each variable participates in a small number of checks. See~\cite{mezard2009information} for a pedagogical introduction and~\cite{eczoo_ldpc} for specific code properties.

    As stated above, LDPC codes form a very large and imprecise class of sparse codes. For our purposes, it will suffice to consider a specific class known as \emph{Gallager's ensemble}, where the parity check matrix $H$ is sampled uniformly at random among ``regular" LDPC codes. Here regular means each check has fixed number of variables (degree) $k$, and similarly each variable has fixed degree $\ell$. Note that $m k = n \ell$, where $m$ is the number of checks and $n$ the number of variables.

    Let $M$ be the block size and $N \leq M$ the message size, with rate $r = N/M$, for a classical code $H \sim \cal{G}(k,\ell)$ sampled from Gallager's ensemble of degree $k$ checks and degree $\ell$ variables. It is conjectured~\cite{zyablov1975estimation} that the maximum efficiently decodable error weight $w$ is given by
    \begin{equation}
        \frac{w}{M} = (1 + o_k(1)) \rm{H}_2^{-1}\left(\frac{c^*(1-r)}{k}\right)
    \end{equation}
    where $o_k$ denotes asymptotics with respect to $k$, $c^*$ is a universal constant, and $\rm{H}_2^{-1}:[0,1]\rightarrow [0,1/2]$ is the inverse of the binary Shannon entropy. For fixed $k$ and $r$, we thus see that, under this conjecture, the fraction of efficiently correctable errors in the message $M$ remains fixed, even for constant rate. Moreover, it was shown in~\cite{gallager1962low} that the typical minimum distance of $H \sim \cal{G}(k,\ell)$ scales linearly in the block size $M$. Thus, when LDPC codes are used in the digital signatures protocol, one only need check $\gamma M = d_r^{-1} \ln(2/\epsilon) = O(\log(1/\epsilon))$ bits per~\Cref{thm:signatures_ecc}. Thus, for such codes the number of verifications scales only with respect to the desired confidence, and logarithmically.

\newpage
\section{Evidence for computational no-learning} \label{sec:learning_alg}

Section \ref{sec:protocol} above formally introduces the proposed protocol for quantum digital signatures, and analyzes its security under the assumption of Computational No-Learning (CNL). 
This section discusses existing, and provides new, evidence for the CNL assumption within the setting of shallow quantum circuits. 

In order to falsify CNL for a given family of random quantum circuits and thereby compromise our proposed quantum digital signatures protocol, an adversary must have access to a circuit learning algorithm $L$ which satisfies the following requirements.
\begin{enumerate} [label=(\alph*)]
    \item The learning must be \emph{proper}: $L$ must return a circuit $C$ within the circuit family $\cal{F}$ agreed upon by all parties. \\[-1em]
    \item $L$ must work for arbitrary circuit geometries, including circuits designed for devices with all-to-all connectivity. \\[-1em]
    \item $L$ must learn given only access to the  classical shadow measurement results published by the communicating parties.
\end{enumerate}
As we discuss in Sec.~\ref{subsec:failure_of_known_algos_for_learning}, there is no known circuit learning algorithm that satisfies even two of the three requirements. In Sec.~\ref{subsec:new_learning_alg}, we introduce an algorithm satisfying the first two requirements. However, there is no clear path to extending it to satisfy all three requirements as it requires a much stronger access model. Finally, Sec.~\ref{subsec:security_against_new_learning_algo} we show that the circuits used in our experiment are secure against this new algorithm \emph{even if adversary is allowed stronger access}. In combination, this suggests that major algorithmic breakthroughs are required to compromise the security of the our protocol.

\subsection{Failure of known quantum algorithms for learning shallow quantum circuits}\label{subsec:failure_of_known_algos_for_learning}

At present, there is no known algorithm for learning shallow quantum circuits which satisfies property (c) above: learning from access to classical shadow measurements and a fixed input state. This rules out any immediate challenges to the protocol. However, there is also no obvious barrier to the development of such algorithms, and the study of learning shallow circuit has only been initiated relatively recently. To assess the potential for effective learning algorithms to break the quantum signatures protocol, here we consider algorithms for closely related circuit learning scenarios. Even allowing this flexibility, we will see that these algorithms fail our stringent learning requirements, and overcoming any of them would require major algorithmic developments or conceptual breakthroughs. 

Most existing algorithms for learning shallow circuits take advantage of the fact that each qubit at input has causal influence over a small subset of qubits \cite{huang2024learning,landau2025learning, fefferman2024anti}. Such causal structure, for circuits with interaction locality (each gate has bounded number of input qubits), tends to break down at logarithmic depth or greater. This causes this class of algorithms to fail on circuits of depth $O(\log^2 n)$ and higher. However, there are several reasons to consider circuit families which are as shallow as possible, yet hard to learn. Besides reducing computational footprint, using shallower circuits as the private key brings quantum digital signatures closer to the realm achievable by near-term quantum hardware. Additionally, it is desirable to delineate the precise transition into CNL with respect to depth, both for its own sake and for assessing how secure a family of circuits is for the protocol. 

The work of~\cite{huang2024learning} initiated the recent studies of learning shallow circuits, showing how circuits $C$ of constant depth, with arbitrary geometry, can be learned in diamond distance in polynomial time, producing a circuit representation $C'$ that is also depth $O(1)$. However, the algorithm requires either unitary query access to $C$ or to randomized measurements over randomized (tomographically complete) inputs. This fails requirement (c) in several ways, in particular, requiring access to variable initial states. Moreover, the algorithm is an improper learner and fails point (a): the representation circuit $C'$ has additional qubits and longer depth compared to $C$ itself. Thus, this algorithm only satisfies point (b) above.

A more suitable algorithm, which learns state preparation circuits is given by~\cite{landau2025learning,kim2024learning}. Here the task is more aligned with digital signatures, since one is required only to produce a circuit $C'$ such that $C'\ket{0}\approx C\ket{0}$. However, these algorithms also fail to be proper (requirement (a)), requiring ancillary qubits and greater circuit depth than the original. Moreover, unlike the work of~\cite{huang2024learning}, the algorithm is limited to geometrically local circuits, and thus does not satisfy property (b). However, (c) is closer to being satisfied compared with the aforementioned algorithm, since the learner only requires a fixed input state.

To summarize, all but one existing results on learning shallow circuits give \emph{improper} algorithms, failing to learn circuit representations within the class $\cal{C}_{n,d}$ of $n$-qubit $d$-depth circuits. While circuit compression schemes could conceivably reduce improperly-learned circuits to proper ones, such computations ought to be generally expensive (though there appears few concrete statements theoretically towards this). To our knowledge, the only competitive algorithm for proper learning of shallow unitary quantum circuits comes from~\citeauthor{fefferman2024anti}~\cite{fefferman2024anti}. There, it was shown that brickwork random circuits with discretized Haar random gates can be efficiently and properly learned up to logarithmic depth. This impressive performance comes at the price of loss of generality: a gate set is specified and the architecture (brickwork) is known to the learner a priori. While the authors recognized that the method can be generalized to circuits with ``good lightcone properties," the full dimensions of this statement have not been explored. Moreover, the algorithm fails to satisfy point (c) above, requiring query access to $C$ so that the learner can test cause-effect relationships with changes to the input. 

Taken with the condition of proper learning, the need for strong access to $C$ seems to eliminate all existing candidates for learning shallow circuits in our protocol. While this already provides evidence in favor of CNL, we further strengthen our claims of security by assuming a hypothetical learner that has the appropriately strong access and devising an extension of the proper learning algorithm of Ref.~\cite{fefferman2024anti}.

\subsection{Proper learning algorithm for all-to-all shallow circuits} \label{subsec:new_learning_alg}

In pursuit of better learning algorithms, and a better understanding of thee hardness of learning under relevant settings, we extend \citeauthor{fefferman2024anti} to accommodate more flexible situations and architectures, such as all-to-all circuits. We choose the algorithm of~\cite{fefferman2024anti} as a starting point because, despite the overly strong circuit access not available in our setting, this algorithm is closest to the type needed to properly learn shallow circuits. It is plausible that an adversary may have knowledge of the circuit architectures used in the signature protocol, without knowing specific gates. Additionally, we note that hardness results derived under stronger access conditions necessarily hold in the weaker case of practical relevance. These results on learning may be of interest beyond cryptographic applications, and additional details of the method will be the subject of a forthcoming paper. 

Because the core idea of the method is simple, we start with a description before moving into technical conditions for success. Most existing learning algorithms for shallow circuits attempt to exploit what is perhaps the key feature induced by shallowness: limited causal influence. Let $Q$ be the collection of $n$ qubits in a circuit. A unitary quantum circuit $C$ naturally defines a directed, acyclic graph, with vertices representing gates and edges representing qubit wires, with specific qubit label. This graph defines causal relationships between qubits according to path connectivity. A forward causal set $\vec{q}_\ell \subseteq Q$ for qubit $q \in Q$ is the collection of qubits which are contained in some path originating at $q$ at layer $\ell \in [d]$. Likewise, a backward causal set $\backvec{q}_\ell$ can be defined in the reverse sense. Forward and backward causal sets generally offer independent sources of information about causal connectivity in the circuit graph. In principle, any $q' \in \vec{q}$ can be influenced by changes in the state of $q$ at input. 

The original causal, or ``light cone," approaches to learning shallow circuits are due to~\cite{huang2024learning} and~\cite{landau2025learning}. These approaches utilize the low operator growth of local observables to learn their time evolution and, from a ``stitching" procedure, produce a single circuit which implements each Heisenberg evolution. Although~\cite{fefferman2024anti} also uses a light cone argument, their method instead utilizes architecture knowledge to detect whether connections have been ``broken" by local gate inversions at the front (or back) of the circuit. To illustrate, consider the setup of Fig. \ref{fig:light_cone}. There is a gate $G$ at the beginning of the circuit such that, when $G$ is present, $q$ and $q'$ are path connected, but when removed, $q, q'$ are disconnected. We call such a gate a \emph{pivot gate}. While~\cite{fefferman2024anti} considered brickwork circuits in particular, this method can be generalized to various architectures, even geometrically nonlocal ones. In an analogous manner, pivot gates can be identified at the back of the circuit.
\begin{figure}
\includegraphics[width=\linewidth]{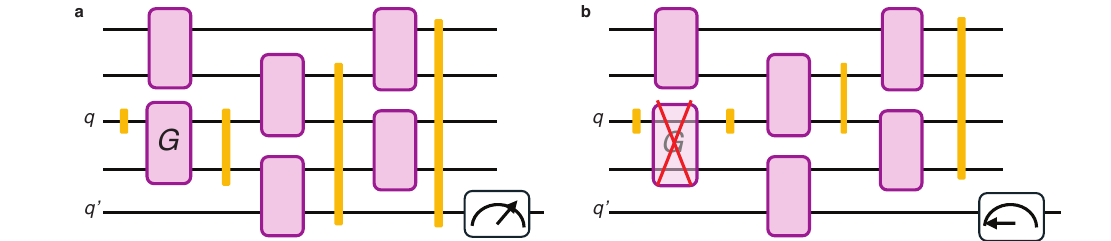}
    \caption{\textbf{Basic principle behind local inversion learning algorithm, using pivot gate $G$. a}, Qubit $q$ has a forward light cone, shown with yellow lines, determined by the gate layout of the circuit. Because qubit $q'$ is contained in this light cone, effects from $q$ can be measured, in principle, at output. \textbf{b}, When gate $G$ is removed, however, $q$ at input does not connect to $q'$ at output, as shown by the new light cones. Measurements observe no changes when the initial state of $q$ is modified. In some cases, this is a computationally efficient observable that can allow for effective learning at shallow depths.}
    \label{fig:light_cone}
\end{figure}

Given the circuit graph (``gate layout") of $C$, the learner's strategy is to identify the location of a pivot gate $G$, prepend (or append, for $G$ in back) to $C$ some ``trial inverse" $G'^\dagger$, and test if $q'$ is no longer causally influenced by $q$. More precisely, one seeks to induce an approximate factorization $G'^\dagger G \approx U_1 \otimes U_2$ for single qubit $U_1, U_2$. This event is detected through single-qubit tomography between an input qubit $q$ and output $q'$. More precisely, for distinct choices $\rho, \rho'$ of state $q$ at input, one checks whether the corresponding output states $\pi, \pi'$ are either (a) identical and $\dist(\pi, \pi') = 0$, or (b) distinct by some threshold $g>0$ so that $\dist(\pi, \pi') > g$. Fig.~\ref{fig:single_inversion} describes this procedure \texttt{SingleInvert}, which is a generalization of Algorithm 2 from~\cite{fefferman2024anti}. 
\begin{figure}
    \centering
    \begin{algorithm}[H]
        \nlnonumber\SingleInvert \;
        \Input{Black-box query access to $U_C$. Labels for input/output qubits $\partial G = \{q, q_2\}$ of pivot gate $G \in \bb{U}(4)$, and corresponding pivot qubits $(q,q')$.}
        \Output{Description of gate $G'$ such that $G'^\dagger G \approx U_1\otimes U_2$.}
        \nlnonumber\Parameters{Factorization precision $\epsilon_f$, tomography precision $\epsilon_t$.}
        \BlankLine
        \For{$G'$ in trial gate set on $\{q,q_2\}$}{
            \eIf{$G$ is front gate}{
                $U \gets  U_C G'^\dagger$\;
            }{
                $U \gets G'^\dagger U_C$\;
            }
            \For{$(\sigma, \sigma_2)$ in $\{X,Y,Z\}\times \{I,X,Y,Z\}$}{
    
                \BlankLine
                initialize $Q \gets \ket{0}^{\otimes n}$\;
                set $q \gets (I+\sigma)/2$ and $q_2 \gets(I+\sigma_2)/2$\;
                $\pi \gets$ SingleQubitTomography$(q', U, Q, \epsilon_t)$
                \BlankLine
                initialize $Q \gets \ket{0}^{\otimes n}$\;
                set $q \gets (I-\sigma)/2$ and $q_2 \gets(I+\sigma_2)/2$  \Comment*[f]{Same but with changed state of $q$}\;
                $\pi' \gets$ SingleQubitTomography$(q', U, Q, \epsilon_t)$\;
                \BlankLine
                \If{$\dist(\pi, \pi') < \epsilon_f$}{
                    return $G'$\;
                }
            }
        }
        \textbf{Protocol fails if this point reached}\;
    \end{algorithm}
    \caption{Single gate inversion heuristic. The meta-algorithm iterates over a set of trial local inversions until one of them breaks the connection between $q$ and $q'$ induced by $G$. Detecting this breaking is done through the SingleQubitTomography protocol, which tests whether $q'$ is affected by changes in the state of $q$ at input.}
    \label{fig:single_inversion}
\end{figure}

A full circuit learning protocol can be built from recursive calls to \texttt{SingleInvert}. For each successful inverse $G'^\dagger$ of $G$, we add $G'$ to our circuit $C'$ under construction and prepend $G^\dagger$ to our calls to $U_C$. Some technical obstacles must be dealt with, such as the possibility of consecutive gates with identical support. This thwarts a naive approach to learning because removing only one of these gates does not change the causal structure. The natural solution is to treat all such sequences as a single composite gate. After this compressed circuit has been learned, arbitrary gate insertions can be used to match the original architecture, if desired. Fig.~\ref{fig:forward_learn} provides pseudocode for the full heuristic algorithm.
\begin{figure} 
    \centering
    \begin{algorithm}[H]
        \nlnonumber\ForwardLearn \;
        \Input{Black-box query access to $U_C$. Circuit graph for $C$.}
        \Output{Description of circuit $C'$ (circuit graph with labelled gates) with same circuit graph as $C$.}
        \nlnonumber\Parameters{Factorization precision $\epsilon_f$, tomography precision $\epsilon_t$.}
        \BlankLine
        $\texttt{Compress}(C)$\Comment*[f]{Group together consecutive gates, if exist}\;
        Compute all forward causal sets for $C$.\;
        partial\textunderscore only $\gets$ false\;
        \For{$G$ in $C$(front:back)} {
            \If{$G$ is not pivot in $C$} {
                partial\textunderscore only $\gets$ true\;
                break\Comment*[f]{Continue to invert what is possible}\; 
            }
            Find pivot pair $(q,q')$ for $G$, and $\{q,q_2\} = \partial G$\;
            $G' \gets \SingleInvert(\partial G, C, q')$\;
            $C$.prepend($G'^\dagger$) and $C'$.append($G'$)\;
        }
        Invert trivial parts of remaining circuit (if applicable, e.g. single-qubit rotations).\;
        \BlankLine
        \If{partial\textunderscore only} {
            print: Warning, only partial inversion possible.\;
        }
        \BlankLine
        \texttt{Decompress}($C'$) \Comment*[f]{Arbitrarily decompose any grouped gates from \texttt{Compress}}\; 
        return $C'$\;
    \end{algorithm}
    \caption{Forward-only learning protocol for a unitary circuit $C$, given oracle access and the circuit graph but not the gates themselves. After computing the causal sets of each qubit, and pre-conditioning the circuit structure, one iterates the \texttt{SingleInvert} protocol with parameters $\epsilon_t, \epsilon_f$ depending on specific settings. With each successful inversion, $G$ is removed from $C$'s circuit graph, and the unitary $G'$ is added to the result $C'$. Final decompression consist of choosing arbitrary sequence for $G'$ to match the layout of uncompressed $C$. We assume each \texttt{SingleInvert} is successful.}
\label{fig:forward_learn}
\end{figure}

\subsubsection{Requirements}
Our description of the learning approach omits important details for implementation, such as the set of trial gates $G'^\dagger$ or factorization precision $\epsilon_f$. These will depend on the class of circuits, and learning will not always be tractable. However, two characteristics are necessarily required for learning:

\begin{enumerate}%
    \item \textbf{Good light cones:} At each stage of the algorithm, there must exist a pivot gate. 
    \item \textbf{Strong signal propagation:} For pivot qubits $(q, q')$, it must be possible to determine if $q'$ belongs to $q$'s light cone using local measurements.
\end{enumerate}
The first property more explicitly means that, after each successful local inversion, there must be some outer gate $G$ and qubits $q, q' \in Q$ such that $q$ connects to $q'$, but only through $G$. This holds for many interaction-local unitary circuits, especially below depth $d = O(\log n)$, since below this depth the causal sets are proper subsets of $Q$ and additional casual connections are possible. If at any point during the simplification process, no outer gates are pivot gates (at either end of the circuit), then the learning algorithm fails. Given any circuit graph, one can check this property by efficiently checking the light cone properties classically, prior to attempting a learning protocol. 

Strong signal propagation can be formulated in the following way: let $\Phi_C$ be the single-qubit channel between pivot qubits $q, q'$ induced by $U_C$. For any input states $\rho, \rho'$ of $q$, we would like $\dist(\Phi_C(\rho), \Phi_C(\rho')) > G \dist(\rho, \rho')$ for some $G$ depending on properties such as graph distance between $q$ and $q'$. This may be hard to show formally in many cases, but for the case of Haar-random circuits such a bound was shown by~\citeauthor{fefferman2024anti}.
\begin{theorem}[Theorem 1.1 of~\cite{fefferman2024anti}] \label{thm:mixing_lower_bounds}
    Let $C$ be a random quantum circuit with a fixed architecture, where each gate is a $k$-qubit independent Haar random unitary. Let $q, q'$ be qubits that are $d$ distance apart in the circuit. Arbitrarily fix the inputs to $C$ except $q$, and let $\Phi_C$ be the channel from $q$ to $q'$ induced by $C$. Then for every $\gamma > 0$, with probability at least $1-\gamma$ over $C$ the following holds: For every two single-qubit states $\rho$ and $\rho'$,
    \begin{equation*}
        \norm{\Phi_\cal{C}(\rho) - \Phi_\cal{C}(\rho')}_\rm{F} \geq \norm{\rho - \rho'}_\rm{F} \cdot (2^{-d} \gamma)^{c_k}
    \end{equation*}
    where $c_k > 0$ is a constant that depends only on $k$ and $\norm{\cdot}_\rm{F}$ is the Frobenius norm.
\end{theorem}
\noindent We observe that this non-mixing result is agnostic to the choice of architecture as long as it is fixed in advance. Because $G$ exponentially decreases in $d$ for such Haar random circuits, polynomial-time learning is limited to logarithmic depths. It is plausible that Haar random circuits have very strong mixing properties, and that other gate families would allow for even stronger degrees of influence. 

Finally, besides the above two ``essential" ingredients for learning, one must also properly consider the effects of small errors in the factorization $U_1 \otimes U_2$. If $G'^\dagger G$ is not exactly product structure, lingering causal influences will interfere with subsequent inversions. For continuous Haar random circuits, a simple error analysis conducted in~\cite{fefferman2024anti} shows that compensating for these errors leads to quasi-polynomial runtime. This is still sub-exponential and may plausibly be done faster than verification. Moreover, working with circuits that only have finite gate families can ameliorate this computational barrier to learning. Ultimately, we find this issue to be generally of less practical concern for effective learning than the previous two points. 

\subsection{Security against extended learning algorithm}\label{subsec:security_against_new_learning_algo} 
Of the conditions mentioned previously for successful learning by local inversion, perhaps the one most naturally in control of the communicating parties is choice of circuit graph. By choosing a circuit with poor light cones, this type of learning attack can be successfully thwarted. A natural choice is to ensure that the forward and backward light cones for each qubit $q \in Q$ contain every single qubit, even if several layers are removed. 

There are limitations imposed by interaction locality on how quickly such light cones can be generated. In particular, $k$-local circuits should require at least $\log_k(n)$ depth. With proper circuit design this can be shown to be sufficient. 
\begin{theorem} \label{thm:safe_circuit_log_depth}
    Fix an integer $k\geq 2$. For every $d \in \bb{Z}_+$, there exists a circuit $C_d$ with $n = k^d$ qubits, depth $d$, and $k$-local gates, such that for every qubit $q \in Q$, $\vec{q} = \backvec{q} = Q$. 
\end{theorem}
\begin{proof}
    We construct $C_d$ as follows. Label each qubit with a unique integer in $[n]$ represented in $k$-ary: $q = q_{d-1}\ldots q_1q_0$ ($q_j \in [k]$). Let $(L_\ell)_{\ell\in[d]}$ be a sequence of $k$-regular partitions of $[n]$, where a block $B \in L_\ell$ includes all $k^{d-1}$ qubits which match in every digit except the $\ell$th. Define the circuit $\cal{C}_d$ by assigning a $k$-qubit gate according to each block in the layer, and stacking the layers together in order, $L_1$ applied first. 

    Consider now a qubit $q=q_{d-1}\ldots q_1q_0$ at input, and consider the succession of causal sets $\vec{q}_{\ell}$ defined by paths from $q$ at input of length at most $\ell$. For $\ell=1$ $\vec{q}_{,1}$ contains all bitstrings which match all but the least significant digit $q_0$. After the second layer, $\vec{q}_{,2}$ contains all qubits which match $q$ except for the last \emph{two} digits. Continuing inductively, after $d$ layers, $\vec{q}_{,d} \equiv \vec{q}_0 = [n]$. Because $q$ was arbitrary, we see that $\vec{q} = Q$ for every $q$. 
    
    The backwards causal sets $\backvec{q}_{,\ell}$ out to depth $\ell$ can be determined in the same manner, but traversing the strings in the opposite direction starting from the most significant bit. We conclude all forward light cones and backward light cones, from input and output respectively, contain every other qubit. 
\end{proof}
\noindent The circuit constructed in the above proof is used to demonstrate the existence of $\log_k(n)$ depth circuits that thwarts the light cone learning algorithm, and is not actually used in the experiment. It obeys the connectivity of a hypercube, and has intermediate connectivity between geometrically local (for fixed dimension) and all-to-all architectures. However, there is no faster way to expand the light cones compared with the above example: at each stage the causal set multiplies by a factor of $k$. Finally, by appending to $C_d$ at least one layer, and maybe more for enhanced security, no causal boundary across layers is present. It is plausible that such a circuit is relatively secure against other conceivable learning attacks based on light cone arguments.

\subsubsection{Advanced considerations} 

For circuits with simple geometric structures, such as brickwork, there is no need to differentiate between removing gates at the beginning or the end due to the symmetry. For more general circuits, however, this symmetry is not present, and a combination of front and back learning can be strictly more powerful than one direction alone. In particular, it is possible that the removal of a gate at the end may cause a new gate at the beginning to become a pivot gate. This is illustrated in Fig. \ref{fig:iterative_learning}, demonstrating that alternating between forward and backward inversions is strictly more effective. However, such a protocol is obstructed from learning the circuit of~\Cref{thm:safe_circuit_log_depth} with appended layers, or any circuit with many layers of full light cone coverage. This is because no outer gates are pivots from the very beginning.

\begin{figure}
    \centering
    \includegraphics[width=\linewidth]{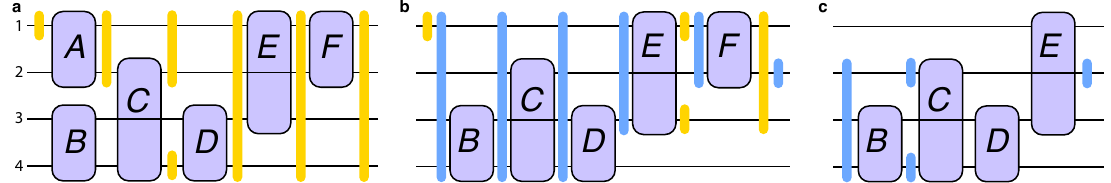}
    \caption{\textbf{Example of circuit learnable only through forward-backward alternation for learning. a,} the only outer gates are $A, B,$ and $F$. If removing gates is only allowed from the back, the algorithm will fail. This is because $F$ is not a pivot gate and the backward causal influence of any qubit is all qubits. Similarly, if the learning algorithm is only allowed to remove gates from the front, the algorithm will fail. Although one can remove pivot gate $A,$ (see \textbf{a,b} for changes in the forward causal influence of qubit $1$ before and upon removing $A$), $B$ becomes the remaining outer gate and is not a pivot. However, after removing $A$, the gate $F$ becomes a pivot gate, as is shown by the change in the backward causal influence of qubit $2$ before and after removing $F$ (compare \textbf{b,c}). Therefore, $F$ and can be removed from the back. After this, one can easily learn gate-by-gate.}
    \label{fig:iterative_learning}
\end{figure}

Even with full light cone coverage, more sophisticated and expensive learning attacks could be tried. For example, one could group together gates into a composite unitary on $k > 2$ qubits that does have good light cone properties. Then one could search for local inversions over the larger unitary space, in a brute force manner. Such a procedure is always guaranteed to succeed, since in the hardest case it reduces to full search over entire circuit. The possibility of gate grouping could require somewhat deeper circuits in order to thwart an adversary. However, we stress that, for our present digital signatures experiment, the adversary already has much weaker circuit access than considered in this section, and therefore these advanced learning techniques are not of primary concern.

\newpage 
\section{Gap between verification and learning}\label{sec:further_gap}
The evidence presented in Section \ref{sec:learning_alg} suggests that existing quantum algorithms for learning sufficiently deep random circuits require exponentially many queries of either the target quantum state or the unitary. As a result, access to polynomial-length classical shadows (without the ability to manipulate the input states) appears to render the task of learning circuit from shadows computationally expensive, if not intractable, for such circuits.

However, because signature verification itself requires exponential time, it is desirable that learning the circuit in question is not merely as difficult as verification, but rather exponentially harder. Formally, if $T_V$ denotes the time for verification, and $T_L$ the time for an adversary to learn the private circuit, one would like $T_L/T_V$ to be very large, e.g., $T_L/T_V = 2^{\Omega(n)}$. Below we provide two types of evidence for such a gap, first  by considering the limitations of so-called ``low-degree polynomial'' learning algorithms in Sec.~\ref{sec:supp_ldp_gap}, and second by considering oracle-based learning algorithms like those based on training variational circuits in Sec.~\ref{sec:supp_oracle_gap}.

\subsection{Gap due to learning hardness with low-degree polynomials}\label{sec:supp_ldp_gap}

One approach to establishing evidence for hardness of learning is to prove hardness for a class of algorithms that use \emph{low-degree polynomials}, where the output of the algorithm is a low-degree polynomial of the input. These algorithms are frequently analyzed in the context of high-dimensional statistical inference problems that seek conclusions about the distribution of high-dimensional data from some samples. In such problems, even when the data contains sufficient statistical information for inference, there may not exist an efficient algorithm to reach that decision. Such problems are said to exhibit an \emph{information-computation gap}. Low-degree polynomials provide a powerful framework of obtaining heuristic or rigorous results on the computational complexity of problems, and they capture a wide range of algorithms including spectral methods, approximate message passing, and local algorithms on graphs \cite{wein2025computational}.

In the quantum setting, we are frequently interested in learning properties of a quantum state by measuring copies. As before, even when polynomially many copies may provide enough information, there may not exist computationally efficient algorithms to learn such properties. A recent initial study of low-degree quantum learning algorithms is provided in Ref.~\cite{chen2025information}, where the authors show hardness of various learning tasks using such algorithms. In particular, they apply their framework to the digital signature of Ref.~\cite{fefferman2025hardness} and claim to prove that learning depth $\Tilde{\Theta}(k\log n)$ brickwork random circuits using degree $k$ polynomials only succeeds with an inverse superpolynomially small probability, which encompasses the learning algorithms discussed above from Refs.~\cite{huang2024learning,landau2025learning}. This holds in the usual low-degree model where $k=\omega(\log n)$ and the algorithm has polynomial time, hence polynomial time low-degree learning of $\mathrm{depth}\geq\log^2(n)$ brickwork random circuits is impossible. In fact, their argument is valid up to $k=\Theta(n)$.

\subsubsection{Relations between learning and hypothesis testing}\label{sec:learning_and_hypothesis_testing}

We first briefly recap the result in Ref.~\cite{chen2025information} showing that depth-($\log^2 n$) circuits cannot be learned using a polynomial-time low-degree algorithm. Consider algorithm $\mathcal{A}$ for learning quantum circuits $C\sim\mathcal{E}$ (with high probability) sampled according to some distribution $\mathcal{E}$, given $m$ copies of the state $C\vert0\rangle$. In particular, $\mathcal{A}$ performs single-copy projection-valued measurements (PVMs) on $C\vert0\rangle$, and the set of all PVMs captures all measurements that can be performed on the state without using ancillae. These measurements on the $m$ copies of $C\vert0\rangle$ could be non-adaptive, meaning they are independent from each other, or adaptive, meaning earlier measurement results may affect the choice of PVMs for later measurements. For the algorithm to be low-degree, the output of the algorithm must be a low-degree polynomial of the binary variables $b_{ij}$ corresponding to the $i$th bit of the $j$th measurement outcome. Fig.~1 and 2 in Section 6 and 7 of Ref.~\cite{chen2025information} provides an intuitive illustration of what we mean by a low-degree algorithm using single-copy PVMs.

To prove the hardness of learning, Section 8.2 of Ref.~\cite{chen2025information} compiles $\mathcal{A}$ into a hypothesis testing algorithm and shows that the related hypothesis testing problem cannot be solved using a low-degree polynomial algorithm. Specifically, they consider the following problem: given access to copies of an unknown $n$-qubit quantum state $\rho$ and given an ensemble $\mathcal{S}$ of structured states, distinguish between the following two cases:
\begin{itemize}
    \item Null case: $\rho=I/2^n$.
    \item Alternative case: $\rho\sim\mathcal{S}$.
\end{itemize}
To connect to quantum circuit learning, choose $\mathcal{S}=\{C\vert0\rangle\}_{C\sim\mathcal{E}}$.

If $\mathcal{A}$ learns a circuit $C_\mathrm{adv}$ such that $\Pr[\vert\langle0\vert C_\mathrm{adv}^\dagger C\vert0\rangle\vert^2\geq\varepsilon]\geq\delta$ over choices of $\rho \sim S$ and random outcomes $C_{\mathrm{adv}}$ of the algorithm, then the following algorithm succeeds at the above hypothesis testing problem. First, use $\mathcal{A}$ on some of the measurement samples to output $C_\mathrm{adv}$. Then, apply $C_\mathrm{adv}^\dagger$ to a copy of the state and measure if the resulting state $C_{\mathrm{adv}}^\dagger C \ket{0}$ is $\vert 0\rangle$. If $\ket{0^n}$ is measured, we output the ``Alternative Case''. From our construction of $\mathcal{A}$, if $\rho \sim \mathcal{S}$, with probability at least $\delta$, we have $\vert\langle0\vert C_\mathrm{adv}^\dagger C\vert0\rangle\vert^2\geq\varepsilon$ and the probability of measuring $\vert 0\rangle$ is at least $\varepsilon$. On the other hand, if $\rho = I/2^n$ (null case), the probability of measuring all zero is $1/2^n$. Overall, the probability that the alternative case is guessed correctly is at least $\delta\varepsilon$. The probability that the null case is accidentally guessed as the alternative case is $1/2^n$.

Therefore, the algorithm is said to achieve $\eta$-detection, where $\eta=\delta\varepsilon-1/2^n$. Here, $\eta$ is defined as the gap in the probability that the detection algorithm outputs $1$ in the null and alternative cases, and the algorithm is said to achieve weak detection if $\eta=\Omega(1)$. For constant $\varepsilon$, for low degree algorithms to achieve weak detection of the hypothesis testing problem, the low-degree learning algorithm $\mathcal{A}$ must satisfy $\Pr[\vert\langle0\vert C_\mathrm{adv}^\dagger C\vert0\rangle\vert^2\geq\varepsilon]\geq \Omega(1)$. 

In our setting, the signer non-adaptively measures $m$ copies of the laboratory state $\sigma_\mathrm{lab}$ honestly prepared on a finite-fidelity quantum computer to generate the classical shadows, and this information is used by a spoofer to learn the circuit with a learning algorithm $\mathcal{A}$. Our learning-from-shadows setting restricts $\mathcal{A}$'s access to $m$ copies of the quantum state using non-adaptive PVMs. Further, for the case of single-qubit random Pauli shadows, the set of PVMs is restricted to single-qubit local PVMs. In the case of level-$\ell$ random Clifford shadows where a random $\ell$-qubit Clifford unitary is applied before measurement of each copy, the learner $\mathcal{A}$'s access is restricted to $\ell$-qubit local PVMs.

\subsubsection{Prior results on the hardness of hypothesis testing}

We first introduce the machinery used for analyzing hypothesis testing. In this section, although certain discussions follow Ref.~\cite{hopkins2018statistical}, we appropriately translate the notation of Ref.~\cite{hopkins2018statistical} into that of Ref.~\cite{chen2025information}.

In Ref.~\cite{hopkins2018statistical}, one aims to distinguish between the null and alternative hypotheses by sampling a random variable $X\in\mathcal{X}$ from distribution $D$ in the null case and $D'$ in the alternative case. A test function $t:\mathcal{X}\rightarrow\{-1,1\}$ aims to output the incorrect binary label with at most $o(1)$ probability. With unbounded computational power, the statistical limit of distinguishing the hypotheses is given by the optimal test function using the likelihood ratio test:
\begin{align}
t(x)=
\begin{cases}
1\quad&\text{if}~\frac{D'(x)}{D(x)}>1\\
-1\quad& \text{otherwise.} 
\end{cases}
\end{align}
The test function is said to achieve $\eta$-detection if
\begin{align}
\left\vert\underset{x
\sim D}{\Pr}[t(x)=1]-\underset{x
\sim D'}{\Pr}[t(x)=1]\right\vert>\eta.
\end{align}

Moving to low-degree polynomials, binary output test functions can be constructed using the so-called \emph{simple statistic}. We say $p$ is a $k$-simple statistic if $p(x)=\sum_ip_i(x)$, each $p_i$ depends on at most $k$ input variables, and $\mathbb{E}_{D'}[p(x)^2]=1$. Ref.~\cite{chen2025information} does not use the language ``simple statistic'' and does not center $p$, whereas Ref.~\cite{hopkins2018statistical} shifts $p$ so that $\mathbb{E}_{D}[p(x)]=0$.

The simple statistic captures in spirit a scoring function that likely outputs a high value for the alternative case, meaning $\mathbb{E}_{D'}[p(x)]$ is large, and it is common to assume a bound on $\mathbb{E}_{D'}[p(x)^2]$. It can then be used within a binary output test function $t(x)=\mathbbm{1}[p(x)\geq c(n)]$ (where $\mathbbm{1}$ outputs $1$ if the condition is true and $-1$ otherwise) for some threshold $c(n)$ that grows with the problem size. This only works when the variances of $p(x)$ under both hypotheses are bounded and the gap in the expectation value of $p(x)$ does not approach 0.

Therefore, the asymptotic behavior of $\mathbb{E}_{D'}[p(x)]$ (which is the same as the gap for Ref.~\cite{hopkins2018statistical}) is crucial in determining whether detection of the alternative hypothesis is possible. In particular, if $\mathbb{E}_{D'}[p(x)]=o(1)$, we say that weak detection is not possible. On the other hand, having $\Omega(1)$ gap in the simple statistic does not guarantee weak detection, since the resulting test function still might under-perform. Nevertheless, upper bounding the gap in $p(x)$ is sufficient to rule out distinguishers.

Since Ref.~\cite{chen2025information} does not center $p$ for the null hypothesis, the gap is more explicitly expressed as
\begin{align}
\mathrm{adv}=\underset{\|p\|_{D}=1}{\sup}\bigg\vert\mathbb{E}_{D}[p(x)]-\mathbb{E}_{D'}[p(x)]\bigg\vert
\end{align}
and is called the distinguishing advantage.

If $p=p^*$ is the optimal polynomial, the advantage is: 
\begin{align}
\mathrm{adv}&=\bigg\vert\sum_{x}D(x)p^*(x)-\sum_{x}D'(x)p^*(x)\bigg\vert=\Bigg\vert\sum_{x}D(x)p^*(x)\left[\frac{D'(x)}{D(x)}-1\right]\Bigg\vert=\vert\langle p^*, \bar{D}-1\rangle_D\vert
\end{align}
for $\bar{D}=D'/D$, where $\bar{D}$ is the likelihood ratio. By the Cauchy-Schwarz inequality, $p^*=(\bar{D}-1)/\|\bar{D}-1\|_D$ maximizes this inner product, and $\mathrm{adv}=\|\bar{D}-1\|_D$. The rest of the analysis in \cite{chen2025information} follows assuming $D$ is the uniform distribution, which simplifies the analysis. 

For low-degree polynomials, they constrain the optimization to $p\in\mathcal{V}^{\leq k}$ (polynomials of degree $k$). Redefining the advantage to maximize only over $k$-degree polynomial results in a smaller advantage $\mathrm{adv}=\|\bar{D}^{\leq k}-1\|_D$, where $f^{\leq k}$ is the projection of $f$ onto $V^{\leq k}$. Finally, to rule out low-degree polynomials, we need $\|\bar{D}^{\leq k}-1\|_D=o(1)$, which is hypothesized to rule out weak detection in Conjecture 1 of Ref.~\cite{chen2025information} or Conjecture 2.2.4 of Ref.~\cite{hopkins2018statistical}.

To achieve this for degree-$k$ polynomials, it turns out that the so-called the $k$-local indistinguishability needs to be small
\begin{align}
kLI=d_{\mathrm{tr}}  \left(\mathbb{E}_{\rho \sim \mathcal{E}} \mathrm{tr}_{[mn]/T}\rho^{\otimes m} , \frac{I}{2^{|T|}}\right)\leq2^{-\Tilde{\Omega}(k\log{n})},
\end{align}
where $T$ is a set of at most $k$ qubits. The condition above checks if a $mn$-sized subsystem ($m$ samples, each of size $n$) of $\rho^{\otimes m}$ is measurably different from the maximally mixed state. However, this cannot hold for $k>n/\log{n}$. Even if $\rho$ is drawn from a random circuit distribution (or some ensemble that approaches Haar-random), this trace distance is expected to decay at most by $2^{-n}$ \cite{ghosh2025unconditional}, which means that one can only satisfy the requirement for $k\leq n/\log{n}$.

\subsubsection{Our extension}

As evidence for the security of our protocol, we'd like to show that the time-complexity of circuit learning is $2^{\omega(n)}$. Since verification takes time $\Theta(2^n)$, this would ensure an asymptotic gap between verifying and spoofing. The time-complexity of a learning algorithm is intimately related to the degree of the learning algorithm. But, how large does $k$ need to be to ensure a gap? Naively, for polynomials with $N$ variables, a $k$-degree algorithm can have at most $\binom{N+k}{k}$ many terms and therefore the time complexity is $O\left(\binom{N+k}{k}\right)$. However, certain polynomials of high degree might be much more efficient to evaluate (for example when the polynomial consists of a single term). Nevertheless, a heuristic degree-runtime correspondence has been hypothesized in \cite{hopkins2018statistical} that argues that, as discussed above Hypothesis 3.2 of Ref.~\cite{wein2025computational}, for natural problems, a degree of $n^\delta$ for a constant $\delta>0$ corresponds to runtime $\exp(n^{\delta\pm o(1)})$. Following this heuristic argument, a degree of  $k=\omega(n)$ would provide an evidence of a gap.

However, due to the fact that the $k$-local indistinguishability is expected to decay at most as $2^{-n}$ as discussed earlier, the argument of \cite{chen2025information} only applies to low-degree algorithms of degree $k =O(n/\log(n))$ and this alone is not sufficient to rule out a time complexity $\omega(2^n)$. To circumvent the aforementioned barrier on the $k$-local distinguishability, we choose the null case to correspond to measuring copies of a Haar random state (instead of the maximally mixed state). In this case, one can make the $k$-local indistinguishability arbitrarily small by making $\mathcal{E}$ close to the Haar measure. For this new null hypothesis, one can similarly compile a hypothesis testing algorithm. If the circuit is from $\mathcal{E}$, then the circuit learning algorithm can learn the circuit and invert the next copy of the state to obtain the zero bitstring. On the other hand, if the circuit is Haar random, it should not be able to learn the circuit and invert the state to obtain the zero bitstring.

However, computing $\|\bar{D}-1\|_D$ is no longer easy due to the presence of a non-uniform measure. To simplify the analysis, we instead consider an alternative definition of distinguishing advantage, which we justify in Section~\ref{sec:change_or_measure}:
\begin{align}
\mathrm{adv}'=\underset{\|p\|_{U}=1}{\sup}\bigg\vert\mathbb{E}_{D}[p(x)]-\mathbb{E}_{D'}[p(x)]\bigg\vert,
\end{align}
where $U$ is the uniform distribution. Then,
\begin{align}
\mathrm{adv}'&=\left|\sum_xp^*(x)[D'(x)-D(x)]\right| =\left|\langle p^*, 2^{mn}(D-D')\rangle_U\right|.
\end{align}
The $2^{mn}$ factor arises from the fact that the measure $U$ has value $1/2^{mn}$ for all $x$, since there are $mn$ bits for $m$ copies of states measured and $n$ bits per measured state. Then, $p^*=(D-D')/\|D-D'\|_U$ and $\mathrm{adv}'=2^{2mn}\|D-D'\|_U$.

Similar to \cite{chen2025information}, we limit $p\in\mathcal{V}^{\leq k}$ (polynomials of degree $k$). Redefining the advantage to maximize only over $k$-degree polynomial gives
\begin{align}
p^*=\frac{D^{\leq k}-D'^{\leq k}}{\|D^{\leq k}-D'^{\leq k}\|_U} \qquad 
\mathrm{adv}'=2^{2mn}\|D^{\leq k}-D'^{\leq k}\|_U.
\end{align}

Now, mirroring section 6.1 of \cite{chen2025information}, consider the Fourier expansion
\begin{align}
D'(s)=\sum_{T\subseteq[mn]}\alpha_T\chi_T(s) \qquad D(s)=\sum_{T\subseteq[mn]}\beta_T\chi_T(s).
\end{align}
Then, using the fact that $\chi_T$ forms an orthonormal basis,
\begin{align}
2^{2mn}\|D^{\leq k}-D'^{\leq k}\|_U&=2^{2mn}\sum_x\left(\sum_{\|T\|\leq k}\alpha_T\chi_T(s)-\sum_{\|T\|\leq k}\beta_T\chi_T(s)\right)^2 =2^{2mn}\sum_{\|T\|\leq k}(\alpha_T-\beta_T)^2.
\end{align}
Further, just as Eq. 5 of Ref.~\cite{chen2025information},
\begin{align}
\alpha_T&=\mathbb{E}_sD'(s)\chi_T(s)=2^{-mn}\mathbb{E}_{\rho,s[T]}2^{\vert T\vert}\mathrm{tr}\left(\underset{[mn]\setminus T}{\mathrm{tr}}\rho^{\otimes m}\bigotimes_{(i,r)\in T}\vert\phi_{s_{i,r}}\rangle\langle\phi_{s_{i,r}}\vert\right).
\end{align}
Similarly,
\begin{align}
\beta_T=2^{-mn}\mathbb{E}_{\rho\in\mathrm{Haar},s[T]}2^{\vert T\vert}\mathrm{tr}\left(\underset{[mn]\setminus T}{\mathrm{tr}}\rho^{\otimes m}\bigotimes_{(i,r)\in T}\vert\phi_{s_{i,r}}\rangle\langle\phi_{s_{i,r}}\vert\right).
\end{align}
Now, define the new $k$-local indistinguishability (NOT $kLLR$)
\begin{align}
kLI'=d_{\mathrm{tr}}\left(\underset{\rho\sim\mathcal{E}}{\mathbb{E}}\underset{[mn]\setminus T}{\mathrm{tr}}(\rho^{\otimes m}),\underset{\rho\sim\mathrm{Haar}}{\mathbb{E}}\underset{[mn]\setminus T}{\mathrm{tr}}(\rho^{\otimes m})\right)\leq\varepsilon,
\end{align}
we have (similar to the proof of Corollary 6.2 in B.1 of \cite{chen2025information})
\begin{align}
\vert \alpha_T-\beta_T\vert=2^{-mn}\big\vert\mathbb{E}_s2^{\vert T\vert}\zeta_s\big\vert\leq2^{\vert T\vert-mn}\varepsilon
\end{align}
for some $\zeta_s\in[-\varepsilon,\varepsilon]$. Therefore,
\begin{align}
\mathrm{adv}'=2^{2mn}\|D^{\leq k}-D'^{\leq k}\|_U\leq 2^{2\vert T\vert}\varepsilon^2k(mn)^k
\end{align}
similar to Proposition 6.3.

While this looks very similar to the original \cite{chen2025information}, we have a new $k$-local distinguishability that can be reduced arbitrarily by making the circuit close to Haar random. Let $\mathcal{E}$ form an $\varepsilon_D$ approximate $m$-design. The reduced density operators are degree-$m$ polynomials of $\rho$, i.e. $f(\rho)=\underset{[mn]\setminus T}{\mathrm{tr}}(\rho^{\otimes m})$ is a degree-$m$ polynomial, so
\begin{align}
&\left\|\underset{\rho\sim\mathcal{E}}{\mathbb{E}}\underset{[mn]\setminus T}{\mathrm{tr}}(\rho^{\otimes m})-\underset{\rho\sim\mathrm{Haar}}{\mathbb{E}}\underset{[mn]\setminus T}{\mathrm{tr}}(\rho^{\otimes m})\right\|_\infty = \left\|\underset{\rho\sim\mathcal{E}}{\mathbb{E}}f(\rho)-\underset{\rho\sim\mathrm{Haar}}{\mathbb{E}}(\rho)\right\|_\infty\leq\varepsilon_D.
\end{align}
Then, by the standard norm inequality between the trace norm and the operator norm that $d_{\mathrm{tr}}(\cdot)\leq d\|\cdot\|_\infty$ where $d$ is the dimensionality of the Hilbert space, we have $kLI\leq \varepsilon_D 2^k$ and $\mathrm{adv}\leq\varepsilon_D^2 2^{4k} k(mn)^k$. Setting $\varepsilon_D=2^{-\Tilde{\Omega}(k\log(n))}$ gives $\mathrm{adv}=o(1)$ and that weak detection is degree-$k$ hard.

As we discussed before, having $\mathrm{adv}=o(1)$ is sufficient for the purpose of digital signatures, and we require that $k=\omega(n)$ for the gap in the time complexity between verification and spoofing. Therefore, the circuit family has to form an $2^{-\Tilde{\omega}(n\log(n))}$-approximate $\mathrm{poly}(n)$-design. Since \cite{chen2025information} shows how to construct $\varepsilon_D$-approximate $t$-designs in depth $O(nt\log^7(t)+\log(1/\varepsilon_D))$ (Lemma 8.7), we can construct circuits of depth polynomial in $n$ to satisfy the requirement.

\subsubsection{Justification of reformulated advantage}\label{sec:change_or_measure}

We are concerned with the possibility that $\mathrm{adv}'\ll\mathrm{adv}$. Fortunately,
\begin{align}
\mathrm{adv}&=\|\bar{D}-1\|_D=\sum_x\frac{(D'(x)-D(x))^2}{D(x)}\\
\mathrm{adv}'&=2^{2mn}\|D-D'\|_U=2^{mn}\sum_x(D'(x)-D(x))^2.
\end{align}
Each term's contribution to the sum is changed by a factor of $2^{mn}D(x)$. To bound the factor by which the distinguishing advantage is decreased, we simply need to lower bound $D(x)$.

In our case, the distribution $D$ is generated by non-adaptive shadow measurements over copies of a Haar random state. Consider a set of shadow measurement bases fixed ahead of time. The first output has the uniform distribution since the state is Haar random. The first measurement outcome tells us something about the actual state. Suppose all measurements are in the computational basis, this would double the probability of measuring this bitstring in the second round and uniformly decrease the probability of the rest bitstrings. For all subsequent rounds, the probability of measuring any bitstring is proportional to the number of occurrence of that bitstring in prior measurements. In the $(i+1)$st round, the probability of sampling a particular bitstring that never occurred before is $2^{-n}\times\left(1-\frac{i}{2^{n}}\right)$. Overall, the lowest probability event is to sample a different bitstring each time with probability
\begin{align}
\prod_{i=0}^{m-1}2^{-n}\times\left(1-\frac{i}{2^{n}}\right),
\end{align}
which is negligibly close to $2^{-mn}$ for $m=\poly(n)$.

For the actual shadow protocol, the above bound on the probability of any particular event is not applicable, but it is plausible that it should be similarly close to $2^{-mn}$. If we were only interested in the distribution of classical variables which is the probabilities in the computational basis, measurement in the computational basis would be the optimal learning strategy. Further, the actual distribution we should examine is $D^{\leq k}(x)$, which should be even closer to the uniform distribution.

\subsection{Gap due to oracle-based learning hardness}\label{sec:supp_oracle_gap}

Another approach of making such an argument is by considering the family of learning algorithms that uses the same computational primitive as the verifier. The verifies uses the public shadows to perform a sort of fidelity estimation, or to perform an overlap against some projector. This primitive takes time $T_P$. We can now wonder how many calls to a similar primitive an adversarial learner has to make in order to learn the circuit. 

On one extreme, we can consider a brute force strategy where the adversary iterates through all possible circuits in the ensemble and performs fidelity estimation with respect to the public shadows. Here, the primitive is the same as verification, which means $T_P=T_V$. This strategy would correspond to $T_L \sim |\mathcal{C}| T_V$, where $|\mathcal{C}|$ is the size of the ensemble and is typically exponential in the number of gates $(nd)$.

However, an adversary can indeed use far more sophisticated algorithms than a brute-force search. An example of an algorithm that use oracle access to ``fidelity-estimation"-like primitive for learning is variational learning where parameters are iteratively optimized against a loss-function, where each query of the loss-function involves fidelity estimation with respect to the target state. Two arguments -- first from Statistical Query Learning and second from the difficulty of variational training -- suggest that one needs exponentially many queries to such an oracle in order to learn a circuit ``properly."

\subsubsection{Statistical query learning argument}
In the statistical query (SQ) learning framework, one aims to learn some observable $M$ on a distribution of states $\mathcal{D}$. One is provided with a correlational stastical query $qCSQ(O, \tau)$ which takes in a bounded observable $O$ with $\| O \| \leq 1$ and a tolerance $\tau$ and returns a value in the range
\begin{equation}
    \mathbb{E}_{\rho \sim \mathcal{D}} [\mathrm{Tr}(O\rho)\mathrm{Tr}(M\rho) - \tau] \leq qCSQ(O, \tau) \leq   \mathbb{E}_{\rho \sim \mathcal{D}} [\mathrm{Tr}(O\rho)\mathrm{Tr}(M\rho) + \tau]
\end{equation}
The hardness of SQ learning is characterized by the query dimension of the hypothesis class $\mathcal{H}$ which the target observable $M$ belongs to. In particular, if the query dimension is $d$, SQ learning requires at least $d(\tau^2-1)/2$ queries to the oracle. 

In the case of learning random circuits $M = \ket{C}\bra{C}$, the hypothesis class $\mathcal{H}$ spans the space of all quantum states, suggesting that $d \sim 4^n$. This can be seen by noting that a $n$-qubit quantum circuit of depth $\log(n)$ can express all $4^n$ Pauli observables. Therefore, the number of oracle calls one needs is exponential in $n$.

However, statistical query learning framework assumes adversarial noise on the query a consequence of which is that even simple circuits (i.e., consisting of single qubit gates) have exponential lower bounds. In the next section, we next explore a more realistic learning algorithm, that is not subject to the restriction of adversarial noise.%

\subsubsection{Hardness of learning variational circuits}
\label{sec:variational}
In variational training, a parameterized circuit is optimized iteratively to optimize a loss function. Given target state $\ket{C}$, one can train a parameterized circuit $\ket{D(\theta)}$ and tune the parameters, using some optimization method, to minimize the loss function, which in the case of state-learning is the infidelity with respect to the target state:
\begin{equation}\label{eq:variational_loss}
    \ell(\theta, C) = 1- |\braket{D(\theta)}{C}|^2
\end{equation}
We consider an adversarial learner that is granted access to an oracle which computes such loss functions and their gradients (or higher derivatives). Such an oracle is at least as powerful as the ``verification" oracle needed to verify signatures. Below, we argue that an adversary defined as such necessarily needs to make an exponential number of queries to such an oracle, suggesting that the ratio of ``time for spoofing with a variational algorithm" to ``time to verify signature" is an exponential function in $n$.

Such an argument is made using a series of works that establish the trainability of variational circuits~\cite{anschuetz2026critical, anschuetz2022quantum, anschuetz2025unified}. Omitting technical details, \cite{anschuetz2025unified} defines a variational circuit $D(\theta)$ to be \textit{trainable} when both the following properties are satisfied:
\begin{enumerate}
    \item There is an absence of barren plateaus, and
    \item There is an absence of poor local minima
\end{enumerate}
\noindent Expressive ansatzes, where the light-cone extends to $\Omega(n)$ qubits, are known to exhibit barren plateaus, meaning that the gradients are exponentially small, requiring exponentially many measurements of the state $\ket{D(\theta)}$ to even learn the direction of the gradient with respect to an ansatz parameter. However, in our conjecture, we allow for a powerful adversary that not only has oracle access to the loss function but also to its gradients. Therefore, we assume the adversary does not encounter barren plateau issues and the first part of the trainability definition is always satisfied. 

Instead, we focus only on the hardness that results from the non-convexity of the optimization landscape. Over the last decade, it has become apparent that variational quantum loss functions can be optimized efficiently only if they are \emph{heavily overparameterized}. Since we can construct our verification protocol to reject circuits that have too many gates, we can ensure that the adversary has to perform an exponentially difficult computation to optimize an acceptable ansatz even if given oracle access to the gradients of the loss function. In other words, any circuits accepted by our verification protocol are not deep enough to avoid poor local minima in their training landscape and thus are not efficiently trainable.

To be more precise, we restrict the adversary to only use ansatzes from our circuit ensemble, which is defined as follows.  We divide $n$ qubits to blocks, each of $b$ qubits. The circuit is constructed by interleaving ``intra-block" and ``inter-block" layers of two-qubit gates. For each ``intra-block", we apply $b/2$ gates between arbitrary pairs within the same block. For each ``inter-block" gates, we pair qubits between any two blocks and apply entangling gates in them. To make the learner circuit more powerful, we allow each of these gates to be SU(4) gates. In particular, we have $d_{\rm inner}$ intra-block layers for each of the $d_{\rm outer}$ inner-block layers. We denote this learner ensembles as $\mathcal{D}(b, d_{\rm inner}, d_{\rm outer})$. For convenience, we consider $n$ to be a multiple of $b$. Let $\ket{C}$ denote the target state.  Let $D_d(\theta)$ denote the learner circuit. The following lemma uses the definition of overparameterization given in \cite{anschuetz2025unified} applied to this circuit ansatz. 

\begin{lemma}\label{lem:first_overparameterization_lemma}
    The variational circuit $D_d(\theta) \sim \mathcal{D}(b, d_{\rm inner}, d_{\rm outer})$ optimizing loss $\ell(\theta,C)$ achieves overparameterization at depth $d = d_{\rm outer} \cdot d_{\rm inner} = O(2^n /n)$. Consequently, any circuit with $d \ll 2^n/n$ exhibits exponentially-many poor local minima and therefore variational training starting from random angles succeeds with an exponentially low probability.
\end{lemma}

This lemma is proven in three steps. (1) Show that the number of parameters $p$ in $D_d(\theta)$ is equal to $15 (n/2) \cdot d_{\rm inner} \cdot d_{\rm outer}$. (2) In order for a variational circuit to have good fidelity with $\ket{C}$, it should have at least $p^* \geq 2^{n+1}$ parameters. (3) Therefore, a circuit from our ansatz needs depth $d \geq \frac{2^{n+1}}{15n}$ in order for variational training with random starting angles to succeed with high probability.

The crux of this proof comes from calculating the overparameterization ratio as defined in \cite{anschuetz2025unified}. In their work, the algebraic structure induced by the circuits is a Jordan algebra $\mathcal{A}$ that admits decomposition $\bigoplus_i \mathcal{A}_i$. From here, $\ell(\theta,C)$ is then understood on each subalgebra and hardness arises from the worst-case subalgebra. For our purposes, the Jordan algebra generated by the circuit ansatz $\mathcal{C}_n$ is simply $\mathcal{H}^{2^n}(\C)$, the space of $2^n \times 2^n$ Hermitian matrices over $\C$ with no additional decomposition structure. By Theorem 13 in \cite{anschuetz2025unified}, $\ell(\theta,C)$ is well-approximated by a calculation over a random Wishart-distributed random matrix with $2r$ degrees of freedom, where 

\begin{equation}
    r \coloneqq \frac{\|H_1\|_*^2}{\|H_1\|_F^2}
\end{equation}

Here $H_1 = -\ket{C}\bra{C} - \lambda_1 I = I -\ket{C}\bra{C}$ since $\lambda_1 = -1$ is the minimum eigenvalue of $H$. Note that the loss function used in \cite{anschuetz2025unified} is simply $\mathrm{Tr}[\rho U_\theta^\dagger HU_\theta]$ which matches \eqref{eq:variational_loss} using $\rho = \ket{0}\bra{0}$, $U_\theta = D_d(\theta)$, and $H_1$. 

\begin{proof}
    First we show that $r = 2^n-1$. Since $H = \ket{C}\bra{C}$ is (minus-1) a rank-1 projector onto a normalized quantum state, we have $spec(H) = (-1,0,\dots,0)$ and $spec(H_1) = (0, 1, \cdots, 1)$. Simple calculations show that 

    \begin{align}
        \|H_1\|_* &= \mathrm{Tr}[H_1] = 2^n -1 \qquad \text{ and} \qquad \|H_1\|_F = \sqrt{\mathrm{Tr}[H_1]} = \sqrt{2^n - 1}.
    \end{align}
    
    \noindent Therefore, $r = 2^n-1$. Let $p$ be the number of parameters in $D_d(\theta)$. Then this circuit is overparameterized with respect to loss $\ell(\theta,C)$ when $p \geq 2r \gtrsim  2^{n+1}$.%

    Circuit $D_d(\theta)$ has $n/2$ gates ($b/2$ gates per each of the $n/b$ blocks) per intra-block layers. Since there are $d_{\rm inner}$ intrablock layers, there are $(n/2) \cdot d_{\rm inner}$ random $\text{SU}(4)$ gates per outer layer. There are $d_{\rm outer}$ outer layer and so we have $(n/2) \cdot d_{\rm inner} \cdot d_{\rm outer}$ total parameterized gates. Each gate is a generic $\text{SU}(4)$ which is parameterized by $4^2 - 1 = 15$ angles. Thus, the total number of parameters $ p = 15(n/2) \cdot d_{\rm inner} \cdot d_{\rm outer} $. When $d =  d_{\rm inner} \cdot d_{\rm outer} \geq \frac{2^{n+2}}{15n}$, we have $p \geq 2r$.
\end{proof}

\paragraph{Numerics} We provide numerical evidence of the result established by the lemma above by showing that circuits only become trainable at an depth exponential in $n$. We consider two circuit ensembles: a) all-to-all ensemble with arbitrary $\text{SU}(2)$ gates and b) circuits that make use of Iceberg code discussed in Section~\ref{sec:qed} (or described as ``Circuit Ensemble" above). In the experiment we perform, we use a circuit ensemble that lends itself to the Iceberg code. The results can be seen in Fig.~\ref{fig:variational-numerics}.

For both ensembles, we sample circuits that define the target state an adversary tries to learn. For each target circuit, we train ansatzes of varying depths (and therefore varying number of parameters) with the circuit structure being the same as the target. For the purposes of demonstration, we consider a circuit learned if the infidelity $\epsilon$ falls below $0.1$. If the infidelity doesn't fall below $0.1$ in 500 training cycles, we consider training to have failed for that ansatz. For each depth, we train multiple ansatzes with different initializations. In Fig.~\ref{fig:variational-numerics}\textbf{a}, we plot the minimal depth required for average training infidelity to drop below $\epsilon=0.1$ within 500 epochs. We see that this ``threshold" depth for training is exponential in $n$. Likewise, we also plot the number of oracle queries (training steps) necessary to train circuits from our ensemble to a target fidelity. As expected, the training is difficult at shallow depth, the regime the experiment was carried on. This suggests that, even in this stronger access model, variational training is only feasible in the vastly ``improper" setting where the ansatz is much deeper than the circuit it is trying to mimic. Proper learning in variational training would thus require a large number of queries to the loss function.

\begin{figure*}
    \centering
    \includegraphics[width=\linewidth]{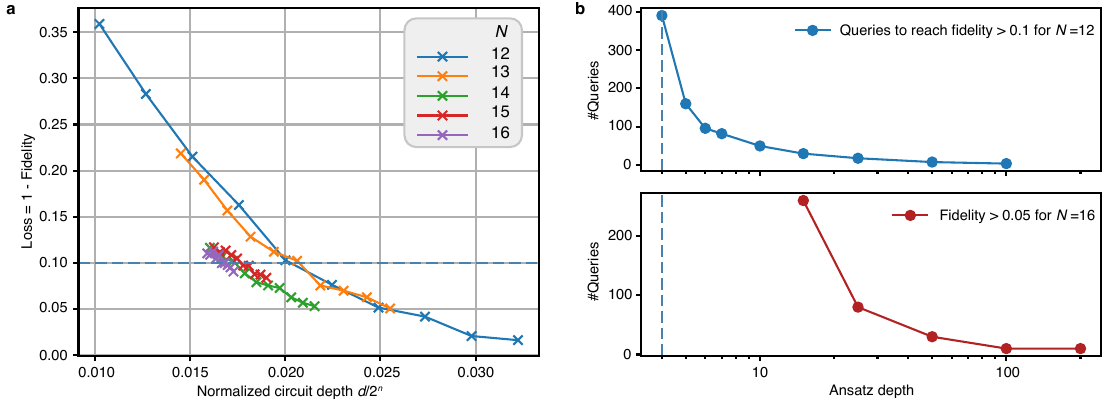}
    \caption{\textbf{a}, Depth required to train parameterized all-to-all circuits to fidelity $ > 0.1$ within 500 training steps and ten random initializations. At constant target fidelity, the required depth scales exponentially in $n$. \textbf{b}, Number of oracle queries required to train circuits from the ensemble used in the experiment to fidelity $>0.1$ and $>0.05$ for $N=12$ and $N=16$ respectively. For both sizes, ten random initializations were trained to 1000 training steps. The dashed vertical line represents the depth used in the experiment with $N=32$.}
    \label{fig:variational-numerics}
\end{figure*}

\newpage
\section{Shadow overlap protocols}\label{sec:shadow}
In Theorem~\ref{thm:noiseless_single_bit_security}, we show that any shadow-overlap-type state certification protocol can be adapted into a digital signature scheme. While known schemes for shadows-based state certification all have sample complexity (size of shadows) that are nominally polynomial, for a practical application, it is desirable to make the exact, numeric size of shadows small, given the attributes of available quantum hardware. The goal of this section is to provide a shadow-based certification protocol which improves on known methods. We begin by briefly summarizing the relevant aspects of the protocol of Ref.~\cite{huang2024certifying}, before introducing the improved protocol. Interested reader is encouraged to see Ref.~\cite{huang2024certifying} for additional details.

At a high level, the procedure for certifying that a ``lab'' state $\rho$ is close to some target state $\ket{\psi}$  of Protocol 2 of \cite{huang2024certifying} is as follows. To generate the shadows, we randomly select a subset of $r$ qubits to measure the lab state in a random single-qubit Pauli basis and measure the rest in the computational basis. Using the shadows, compute the overlap with some operator dependent on the target quantum state to certify. Specifically, the procedure yields the so-called shadow overlap estimator $\omega$, which also satisfies $\mathbb{E}[\omega]=\mathrm{tr}(L\rho)$ where $L$ is some operator and $\rho$ is the measured state. Given $L$, one can identify the certification procedure.

The procedure for certifying quantum state fidelity in \cite{huang2024certifying} uses an $L$ such that $\vert\psi\rangle$ is an eigenvector with eigenvalue $1$ and that the shadow overlap protocol over shadows generated from state $\rho$ yields $\mathrm{tr}(L\rho)$ in expectation (i.e., $\mathrm{tr}(L\rho)=\mathbb{E}[\omega]$, where $\omega$ is the computed shadow overlap for a given trial). A step in establishing $\text{tr}(L\rho)=\mathbb{E}[\omega]$ is that $L$ can be written as the sum
\begin{align}
L=\frac{1}{N}\sum_{r\in[m]}\sum_{k=\{k_1,\dots,k_r\}}\sum_{z_k\in\{0,1\}^{n-r}}\vert z_k\rangle\langle z_k\vert\otimes L_{z_k}\label{eqn:L_form},
\end{align}
where $m$ is a positive integer that is at most $n$, $k$ is a subset of $r$ out of $n$ qubits, $z_k$ corresponds to a particular measurement outcome on the qubits measured in the fixed $Z$ basis, $L_{z_k}$ is efficiently computable, and $N$ is the number of terms due to the summation over $r$ and $k$ which is given by $\sum_{k=1}^m\binom{n}{k}$. 
Specifically, \cite{huang2020predicting} uses the following definition of $L_{z_k}$:
\begin{align}
L_{z_k}=\sum_{\substack{\ell_1,\ell_2\in\{0,1\}^r\\\mathrm{dist}(\ell_1,\ell_2)=r}}\vert\Psi_{z_k}^{\ell_1,\ell_2}\rangle\langle\Psi_{z_k}^{\ell_1,\ell_2}\vert\quad\text{with}\quad\vert\Psi_{z_k}^{\ell_1,\ell_2}\rangle=\frac{\Psi(z_k^{(\ell_1)})\vert\ell_1\rangle+\Psi(z_k^{(\ell_2)})\vert\ell_2\rangle}{\sqrt{\vert\Psi(z_k^{(\ell_1)})\vert^2+\vert\Psi(z_k^{(\ell_2)})\vert^2}},
\end{align}
where $\Psi$ is a model providing query access to amplitudes of $\ket{\psi}$, $z_k^{\ell}$ equals $\ell$ on bits $k_1,\dots,k_r$ and equals $z_k$ on the remaining $n-r$ bits.

To certify a state, $T$ copies of the state $\rho$ are measured according to the shadow overlap protocol and produce an estimator $\hat{\omega}$, which is given by the empirical average $\frac{1}{T}\sum_{t=1}^T \omega_t$. Ref.~\cite{huang2024certifying} states that $T=O(2^{2m})$ samples are required such that $\vert\hat{\omega}-\mathbb{E}[\omega]\vert\leq\varepsilon$ with high probability. The number of samples required (adjusting for an error in the use of Hoeffding's inequality in \cite{huang2020predicting}) is given by the following theorem: 
\begin{theorem}[Theorem 5 of \cite{huang2024certifying}, corrected]\label{thm:fixed_huang_verification}
Using $T\geq\frac{2^{4m-1}}{\varepsilon^2}\ln\frac{2}{\delta}$ samples of the state $\rho$, the empirical average $\frac{1}{T}\sum_{t=1}^T\omega_t$ computed by Protocol 2 of \cite{huang2024certifying} is within $\varepsilon$ additive distance from the shadow overlap $\mathbb{E}[\omega]$ with probability at least $1-\delta$.
\end{theorem}
\begin{proof}
The proof is identical to Theorem 5 of \cite{huang2020predicting} except for the application of Hoeffding's inequality. Instead of Eq. 30 of \cite{huang2024certifying}, the corrected quantity is
\begin{align}
\Pr\left[\left\vert\frac{1}{T}\sum_{t=1}^T\omega_t-\mathbb{E}[\omega]\right\vert>\varepsilon\right]\leq2e^{-\frac{2T^2\varepsilon^2}{T(2^{2m-1}-(-2^{2m-1}))^2}}=2e^{-\frac{T\varepsilon^2}{2^{4m-1}}}\leq\delta.
\end{align}
Solving for $T$ completes the proof.
\end{proof}
While this determines the sample complexity required to estimate the shadow overlap $\omega$ precisely, it still remains to show the connection between $\mathbb{E}[\omega]$ and the fidelity $\langle\psi\vert\rho\vert\psi\rangle$, which is necessary for certification. For any $L$, one can compute the so-called relaxation time $\tau$, which corresponds to the relaxation time of the Markov chain defined by $L$. Then, Theorem 4 of \cite{huang2024certifying} guarantees a equivalence between shadow overlap and fidelity.  Namely, if $\mathbb{E}[\omega]\geq 1-\varepsilon$, then $\langle\psi\vert\rho\vert\psi\rangle\geq1-\tau\varepsilon$ (high shadow overlap implies high fidelity). Conversely, if $\langle\psi\vert\rho\vert\psi\rangle\geq 1-\varepsilon$, then $\mathbb{E}[\omega]\geq 1-\varepsilon$ (high fidelity implies high shadow overlap).

For a concrete experiment, it becomes important to determine $\tau$. The relaxation time $\tau$ for $m=1$ is shown to be upper-bounded by $n^2$ for Haar random states. For noisy quantum devices, however, even if $\rho$ is prepared by executing the quantum circuit for $\vert\psi\rangle$, $\langle\psi\vert\rho\vert\psi\rangle$ cannot be arbitrarily close to 1 due to the fidelity limitations of the device. Having $\langle\psi\vert\rho\vert\psi\rangle\geq1-\varepsilon$, where $\varepsilon$ is the infidelity, implies $\mathbb{E}[\omega]\geq 1-\varepsilon$, which certifies that $\langle\psi\vert\rho\vert\psi\rangle\geq1-\tau\varepsilon$. This means the infidelity of certification grows by a factor of $\tau$. For $\tau=O(n^2)$, certification becomes impossible for low fidelity devices as the requirement $\tau\epsilon\ll 1$ implies $\epsilon\ll\frac{1}{O(n^2)}$.  Although empirical evidence in Fig.~\ref{fig:relaxation_time}\textbf{a} suggests $\tau=O(n)$ for Haar-random states, this is still highly impractical for fidelities below $0.99$.

Specifically, to see this issue, we can state the corrected shadow overlap protocol sample complexity for obtaining completeness and soundness of certification.
\begin{theorem}[Theorem 6 of \cite{huang2024certifying}, corrected]\label{thm:fixed_huang_verification_completeness_and_soundness}
Given any $n$-qubit hypothesis state $\vert\psi\rangle$ with a relaxation time $\tau\geq 1$, error $\varepsilon>0$, failure probability $\delta>0$, and
\begin{align}
T=2^{4m+3}\cdot\frac{\tau^2}{\varepsilon^2}\cdot\ln\frac{2}{\delta}
\label{eq:original-sample-complexity}
\end{align}
samples of an unknown $n$-qubit state $\rho$, with probability at least $1-\delta$, Protocol \ref{prot:random_clifford_improved} with shadow overlap threshold $\omega_{\mathrm{threshold}}\equiv1-\varepsilon'=\frac{3\varepsilon}{4\tau}$ will correctly output \textsc{Failed} if the fidelity is low, i.e. $\langle\psi\vert\rho\vert\psi\rangle<1-\varepsilon$, and will correctly output \textsc{Certified} if the fidelity is high, i.e. $\langle\psi\vert\rho\vert\psi\rangle\geq1-\frac{\varepsilon}{2\tau}$.
\end{theorem}
\begin{proof}
The proof is identical to Theorem 6 of \cite{huang2024certifying} except we use Theorem \ref{thm:fixed_huang_verification}. Simply substituting $\varepsilon\rightarrow\varepsilon/\tau-\varepsilon'=\frac{\varepsilon}{4\tau}$.
\end{proof}

\subsection{Improved random Pauli shadow overlap protocols}\label{sec:supp_improved_shadow}
One way of reducing the sample complexity (Eq.~\ref{eq:original-sample-complexity}) is by reducing the relaxation time $\tau$. In order to do so, we make two observations. First, increasing the level $m$ of the protocol increases the connectivity of the Markov chain. This is apparent since the supported edges are between bitstrings that differ in at most $m$ bits. Second, as long as $L_{z_k}$ is the outer product of normalized quantum states, any $L$ of the form of Eq.~\ref{eqn:L_form} has a correspondence to a shadow overlap procedure, with $L_{z_k}$ as the observable of interest.

As a result, we define the following operator
\begin{align}
L=\binom{n}{m}^{-1}\times\sum_{k=\{k_1,\dots,k_m\}}\sum_{z_k\in\{0,1\}^{n-m}}\vert z_k\rangle\langle z_k\vert\otimes L_{z_k}, \qquad \text{with } \quad L_{z_k}=\vert\Psi_{z_k}\rangle\langle\Psi_{z_k}\vert\quad\text{and}\quad\vert\Psi_{z_k}\rangle=\frac{\sum_{\ell\in\{0,1\}^m}\Psi(z_k^\ell)\vert\ell\rangle}{\sqrt{\sum_{\ell\in\{0,1\}^m}\vert\Psi(z_k^\ell)\vert^2}}.
\end{align}
Now, we redefine the shadow overlap procedure with the new $L_{z_k}$. The same steps of Proposition 3 of \cite{huang2024certifying} yield $\mathrm{tr}[L\rho]=\mathbb{E}[\omega]$.

We will prove later that $L\vert\psi\rangle=\vert\psi\rangle$. It is clear from definition that $L$ is Hermitian. One also can define the Markov chain with transition matrix $P=S^{-1/2}F^\dagger LFS^{1/2}$, which is exactly the same as Proposition 3 of \cite{huang2024certifying}, which allows us to compute the relaxation time and bound the fidelity from the shadow overlap.

Fig.~\ref{fig:relaxation_time} shows the relaxation of Haar random states and states prepared by the ansatz of our experimental quantum circuits using different definitions of the operator $L$. In particular, our experimental quantum circuit ansatz is as described in the main text, and we show the shadow overlap of those with two blocks and $\lceil\log_2n\rceil$ layers. We observe in Fig.~\ref{fig:relaxation_time}\textbf{a} that for Protocol 2 of \cite{huang2020predicting}, increasing $m$ does not reduce the relaxation time meaningfully beyond $m=3$. Additionally, the slopes seem to be about the same for different values of $m$, indicating similar asymptotic scaling of $\tau$. On the other hand, Fig.~\ref{fig:relaxation_time}\textbf{b} shows that our protocol achieves a notably lower relaxation time. For $m=4$, extrapolating to $n=32$ gives $\tau\approx 8$. Note that the horizontal axis scaling is different from Fig.~\ref{fig:relaxation_time}\textbf{a}. The asymptotic scaling also improves with $m$. An overall trend is that increasing $m$ or $n$ decreases the variance in the relaxation time. We also show that the choice of the quantum circuits we make does not impact the average relaxation time very much, as is evidenced by Fig.~\ref{fig:relaxation_time}\textbf{c}, but it does increase the variance in the relaxation time.

\begin{figure}[!t]
    \centering
    \includegraphics[width=\linewidth]{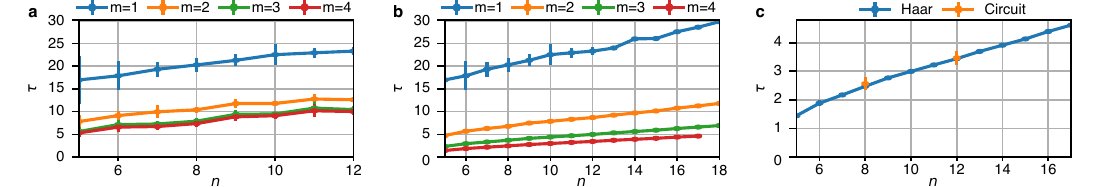}
    \caption{\textbf{Relaxation time of the certification procedures. a}, Relaxation time of Protocol 2 of \cite{huang2024certifying} for Haar random states. \textbf{b}, Relaxation time of our Protocol \ref{prot:random_pauli_improved} for Haar random states. \textbf{c}, Relaxation time of our $m=4$ protocol for Haar random states and the state generated using the ansatz of our experimental quantum circuits with two blocks and $\lceil\log_2n\rceil$ layers. Error bars show the standard deviation. Data points with $n\leq12$ are averaged over 10 random instances, whereas each data point beyond that only consists of a single instance.}
    \label{fig:relaxation_time}
\end{figure}

To understand why this construction achieves a lower relaxation time, consider the extreme case of $L=\vert\psi\rangle\langle\psi\vert$. All eigenvectors other than $\vert\psi\rangle$ have eigenvalue zero and the relaxation time is the smallest it can be, i.e. $\tau=1$. The operator $L$ defined in \cite{huang2024certifying} is not $\vert\psi\rangle\langle\psi\vert$ even for $m=n$. This is due to the normalization factors of $\vert\Psi_{z_k}^{\ell_1,\ell_2}\rangle$ being the root of the sum of two probabilities always. On the other hand, for $m=n$, our $L$ operator is exactly $\vert\psi\rangle\langle\psi\vert$.

Additionally, this construction of $L$ is more natural than the scheme in \cite{huang2024certifying}. Specifically, $L_{z_k}$ is the collapsed state of the hypothesis state $\vert\psi\rangle$ conditioned on the qubits not in $k$ being measured to $z_k$. Although $L$ is a sum of these collapsed states, it is not $\vert\psi\rangle\langle\psi\vert$ even up to normalization. It may be tempting to interpret $L$ as $\vert\psi\rangle$ going through a random measurement channel up to some normalization, but even this is not correct since each term corresponding to $z_k$ have equal weight instead of proportional to the measurement probability. In fact, if we use the density operator of $\vert\psi\rangle$ going through a random measurement channel, it would not have $\vert\psi\rangle$ as an eigenvector in general.

We now formally describe the protocol in Protocol \ref{prot:random_pauli_improved} and prove its correctness.

\parbox{0.96\textwidth}{\begin{protocol_framed}[Improved level-$m$ shadow overlap]\label{prot:random_pauli_improved}\quad 
\normalfont
\newline
\begin{flushleft}\underline{Input:}\end{flushleft}
 \vspace{-1.5em}
 \begin{itemize}
    \item $T$ samples of an unknown state $\rho$.
    \item A model $\Psi$ that gives query access to the amplitudes of $\vert\psi\rangle$.
    \item A level $m\in[n]$.
    \item Error $0\leq\varepsilon'<1$, or equivalently shadow overlap threshold $\omega_{\mathrm{threshold}}\equiv1-\varepsilon'
    $.
\end{itemize}
\begin{flushleft}\underline{Procedure:}\end{flushleft}
 \vspace{-1.5em}
 \begin{enumerate}
    \item Uniformly choose $m$ random qubits denoted by $k=\{k_1,\dots,k_m\}$.
    \item Perform single-qubit $Z$-basis measurements on all but qubits $k_1,\dots,k_m$ of $\rho$. Denote the measurement outcomes collectively by $z_k\in\{0,1\}^{n-m}$.
    \item For each qubit $k_1,\dots,k_m$, choose an $X,Y,Z$-basis measurement uniformly at random and measure it. Denote the post-measurement state of these qubits by $\vert s_1\rangle,\dots,\vert s_m\rangle$. Compute the classical shadow
    \begin{align}
        \sigma=(3\vert s_1\rangle\langle s_1\vert-\mathbb{I})\otimes\cdots\otimes(3\vert s_m\rangle\langle s_m\vert-\mathbb{I}).\label{eqn:pauli_shadow}
    \end{align}
    \item Query the model $\Psi$ for all choices of $m$-bit strings $\ell$ to obtain the normalized state
    \begin{align}
        \vert\Psi_{z_k}\rangle=\frac{\sum_{\ell\in\{0,1\}^m}\Psi(z_k^\ell)\vert\ell\rangle}{\sqrt{\sum_{\ell\in\{0,1\}^m}\vert\Psi(z_k^\ell)\vert^2}},
    \end{align}
    where $z_k^{\ell}$ equals $\ell$ on bits $k_1,\dots,k_m$ and equals $z_k$ on the remaining $n-m$ bits.
    \item Compute the overlap
    \begin{align}
        \omega=\mathrm{tr}(L_{z_k}\sigma)\quad\text{with}\quad L_{z_k}=\vert\Psi_{z_k}\rangle\langle\Psi_{z_k}\vert.
    \end{align}
    \item Repeat steps 1 to 5 for $T$ times to obtain overlaps $\omega_1,\dots,\omega_T$. Report the estimated shadow overlap
    \begin{align}
        \hat{\omega}=\frac{1}{T}\sum_{t=1}^T\omega_t.
    \end{align}
    \item If the estimated shadow overlap $\hat{\omega}\geq 1-\varepsilon'$, output \textsc{Certified}. Otherwise, output \textsc{Failed}.
\end{enumerate}
\end{protocol_framed}}

\begin{lemma}\label{lem:eigenvector}
Consider a quantum state $\psi$
\begin{align}
\vert\psi\rangle=\Psi(z)\vert z\rangle.
\end{align}
Define an operator $L$
\begin{align}
L=\binom{n}{m}^{-1}\times\sum_{k=\{k_1,\dots,k_m\}}\sum_{z_k\in\{0,1\}^{n-m}}\vert z_k\rangle\langle z_k\vert\otimes L_{z_k},
\end{align}
where
\begin{align}
L_{z_k}=\vert\Psi_{z_k}\rangle\langle\Psi_{z_k}\vert\quad\text{with}\quad\vert\Psi_{z_k}\rangle=\frac{\sum_{\ell\in\{0,1\}^m}\Psi(z_k^\ell)\vert\ell\rangle}{\sqrt{\sum_{\ell\in\{0,1\}^m}\vert\Psi(z_k^\ell)\vert^2}},
\end{align}
where $z_k^{\ell}$ equals $\ell$ on bits $k_1,\dots,k_m$ and equals $z_k$ on the remaining $n-m$ bits. We have \begin{align}
L\vert\psi\rangle=\vert\psi\rangle.
\end{align}
\end{lemma}
\begin{proof}
We first observe that given $m$ qubits specified by vector $k$,
\begin{align}
\langle z_k\vert\psi\rangle=\langle z_k\vert\sum_z\Psi(z)\vert z\rangle=\langle z_k\vert\sum_{z_k'}\sum_{\ell\in\{0,1\}^m}\Psi(z_k'^\ell)\vert z_k'\rangle\otimes\vert\ell\rangle_k=\sum_{\ell\in\{0,1\}^m}\Psi(z_k^\ell)\vert\ell\rangle_k,
\end{align}
which means that the post selected state is
\begin{align}
\frac{\langle z_k\vert\psi\rangle}{\sqrt{\langle\psi\vert z_k\rangle\langle z_k\vert\psi\rangle}}=\vert\Psi_{z_k}\rangle\quad
\end{align}
for $\vert\Psi_{z_k}\rangle$ and $L_{z_k}$ defined in the lemma. Additionally,
\begin{align}
\vert\psi\rangle=\mathbb{I}\vert\psi\rangle=\sum_{z_k}\vert z_k\rangle\langle z_k\vert\psi\rangle=\sum_{z_k}\vert z_k\rangle\otimes\sum_{\ell\in\{0,1\}^m}\Psi(z_k^\ell)\vert\ell\rangle_k=\sum_{z_k}\sqrt{p(z_k)}\vert z_k\rangle\otimes\vert\Psi_{z_k}\rangle,
\end{align}
where $p(z_k)=\sum_{\ell\in\{0,1\}^m}\vert\Psi(z_k^\ell)\vert^2$.
\begin{align}
\binom{n}{m}L\vert\psi\rangle
=&\sum_{k}\sum_{z_k}\vert z_k\rangle\langle z_k\vert\otimes L_{z_k}\vert\psi\rangle\\
=&\sum_{k}\sum_{z_k}\left(\vert z_k\rangle\otimes\vert\Psi_{z_k}\rangle_k\right)\left(\langle z_k\vert\otimes\langle\Psi_{z_k}\vert_k\right)\sum_{z_k'}\sqrt{p(z_k')}\vert z_k'\rangle\otimes\vert\Psi_{z_k'}\rangle\\
=&\sum_{k}\sum_{z_k}\sqrt{p(z_k)}\vert z_k\rangle\otimes\vert\Psi_{z_k}\rangle_k =\sum_{k}\vert\psi\rangle=\binom{n}{m}\vert\psi\rangle.
\end{align}
\end{proof}

\begin{lemma}\label{lem:trace_equal_shadow_overlap_pauli}
For $\omega$ of Protocol \ref{prot:random_pauli_improved}, we have
\begin{align}
\mathrm{tr}(L\rho)=\mathbb{E}[\omega].
\end{align}
\end{lemma}
\begin{proof}
We follow the same idea as Proposition 3 of \cite{huang2024certifying}.
\begin{align}
\mathrm{tr}[L\rho]&=\binom{n}{m}^{-1}\times\sum_{k=\{k_1,\dots,k_m\}}\sum_{z_k\in\{0,1\}^{n-m}}\mathrm{tr}\left(\langle z_k\vert\rho\vert z_k\rangle\right)\cdot \mathrm{tr}\left(L_{z_k}\cdot\frac{\langle z_k\vert\rho\vert z_k\rangle}{\mathrm{tr}(\langle z_k\vert\rho\vert z_k\rangle)}\right)=\underset{k,z_k}{\mathbb{E}}\mathrm{tr}\left(L_{z_k}\cdot\frac{\langle z_k\vert\rho\vert z_k\rangle}{\mathrm{tr}(\langle z_k\vert\rho\vert z_k\rangle)}\right).
\end{align}
We note that converting the sum into expectation is valid just like in \cite{huang2020predicting}, although the sampling process is different. In \cite{huang2020predicting}, the procedure samples $r$, followed by sampling $k=\{k_1,\dots,k_r\}$, and finally measuring $z_k$. Here, we directly sample $k=\{k_1,\dots,k_m\}$ and then measure $z_k$. Following the rest of the steps identically as \cite{huang2024certifying} completes the proof.
\end{proof}

\begin{theorem}\label{thm:shadow_implies_fidelity_pauli}
Let $\lambda_1=1-\frac{1}{\tau}$ be the second largest eigenvalue of the transition matrix $P$ defined as $P=S^{-1/2}F^\dagger LFS^{1/2}$ for state $\vert\psi\rangle$, where $L$ is as defined in Lemma \ref{lem:eigenvector}. The shadow overlap satisfies the following. If $\mathbb{E}[\omega]\geq 1-\varepsilon$, then $\langle\psi\vert\rho\vert\psi\rangle\geq 1 - \tau\varepsilon$. If $\langle\psi\vert\rho\vert\psi\rangle\geq 1-\varepsilon$, then $\mathbb{E}[\omega]\geq 1-\varepsilon$ for $\omega$ as defined in Protocol \ref{prot:random_pauli_improved}.
\end{theorem}
\begin{proof}
The proof is identical to Theorem 4 of \cite{huang2024certifying}.
\end{proof}

\begin{theorem}\label{thm:random_pauli_verification}
Using $T=2^{2m+1}\cdot\frac{1}{\varepsilon^2}\cdot\ln(\frac{2}{\delta})$ samples of the state $\rho$, the empirical average $\frac{1}{T}\sum_{t=1}^T\omega_t$ computed by Protocol \ref{prot:random_pauli_improved} is within $\varepsilon$ additive distance from the shadow overlap $\mathbb{E}[\omega]$ with probability at least $1-\delta$.
\end{theorem}
\begin{proof}
Since $L_{z_k}$ is the outer product of one normalized vector, $\|L_{z_k}\|_1=1$. Following the rest of proof identically as Theorem \ref{thm:fixed_huang_verification} while noting the difference in $\|L_{z_k}\|_1=1$ instead of $\|L_{z_k}\|_1\leq2^{m-1}$ completes the proof.
\end{proof}

To reduce the relaxation time $\tau$, we want $m$ to be as high as possible. However, the sample complexity required so that the estimator of the shadow overlap is close to the expectation value is $O(2^{2m})$ as stated in Theorem \ref{thm:random_pauli_verification}. Nevertheless, compared the protocol in \cite{huang2024certifying} where the sample complexity is proportional to $2^{4m-1}$, we need quadratically fewer samples in $m$.

\subsection{Random Clifford shadow overlap protocols}

Nevertheless, for large $m$, even $2^{2m}$ samples can be an obstacle in practice. Fortunately, the shadow overlap variance bound can be notably improved if the shadows are obtained using random Clifford measurements. For any observable $O$ and set $\hat{o}=\mathrm{tr}(O\hat{\rho})$, where $\hat{\rho}$ is a classical shadow, we have
\begin{align}
\mathrm{Var}[\hat{o}]\leq\left\|O-\frac{\mathrm{tr}(O)}{2^n}\mathbb{I}\right\|^2_{\mathrm{shadow}}=\|O_0\|^2_{\mathrm{shadow}}
\end{align}
from Lemma 1 of \cite{huang2020predicting}, where the second equality serves as a definition for $O_0$ which is a traceless operator in $\mathbb{H}_{2^n}$. For classical shadows obtained by random Clifford measurements, we have
\begin{align}
\|O_0\|_{\mathrm{shadow}}^2\leq3\mathrm{tr}(O^2)
\end{align}
by Proposition 1 of \cite{huang2020predicting}. Since our operator $L_{z_k}$ is an outer product of a normalized vector,
\begin{align}
\mathrm{Var}[\omega_{z_k}]\leq\left\|L_{z_k}-\frac{\mathrm{tr}(L_{z_k})}{2^n}\mathbb{I}\right\|_{\mathrm{shadow}}\leq3,
\end{align}
where $\omega_{z_k}$ is the shadow overlap for the measurement outcome $z_k$ on the $n-m$ qubits other than those specified by set $k$.

It should be noted that the above is the variance of $\omega_{z_k}$, not $\omega$. Each $\omega_{z_k}$ could have variation in the expectation value between $0$ and $1$. This is because
\begin{align}
\mathbb{E}[\omega_{z_k}]=\underset{\mathrm{shadows}~\sigma}{\mathbb{E}}[\mathrm{tr}(L_{z_k}\sigma)]=\mathrm{tr}\left(L_{z_k}\cdot\frac{\langle z_k\vert\rho\vert z_k\rangle}{\mathrm{tr}(\langle z_k\vert\rho\vert z_k\rangle)}\right),
\end{align}
where $\mathbb{E}_{\mathrm{shadow}s~\sigma}[\sigma]=\frac{\langle z_k\vert\rho\vert z_k\rangle}{\mathrm{tr}(\langle z_k\vert\rho\vert z_k\rangle)}$ is a normalized state as discussed in \cite{huang2020predicting} and Appendix B of \cite{huang2024certifying}. This is because shadows of a state average to the unknown state by construction.

The overall variance of $\omega$ therefore is upper bounded by $\mathrm{Var}[\omega]\leq 3.25$. To see this, consider the contribution to variance due to $\omega_{z_k}$ where $\mathbb{E}[\omega]-\mathbb{E}[\omega_{z_k}]=\delta_{z_k}$. Then,
\begin{align}
\mathbb{E}\left[(\omega_{z_k}-\mathbb{E}\left[\omega\right])^2\right]&=\mathbb{E}\left[(\omega_{z_k}-\mathbb{E}\left[\omega_{k_z}\right]+\delta_{z_k})^2\right]\\
&=\mathbb{E}\left[(\omega_{z_k}-\mathbb{E}\left[\omega_{z_k}\right])^2+2(\omega_{z_k}-\mathbb{E}\left[\omega_{z_k}\right])+\delta_{z_k}^2\right]\\
&=\mathrm{Var}\left[\omega_{z_k}\right]+\delta_{z_k}^2.
\end{align}
Note that we have the constraint that $\sum_{z_k}p_{z_k}\cdot\delta_{z_k}=0$ since $\sum_{z_k}p_{z_k}\cdot\mathbb{E}[\omega_{z_k}]=\mathbb{E}[\omega]$, where $p_{z_k}$ is the probability of selecting $m$ random qubits denoted by $k$ and measuring the rest of the $n-m$ qubits output bitstring $z_k$. Further, we have $\delta_{z_k}\in[-\mathbb{E}[\omega], 1-\mathbb{E}[\omega]]$. It is easy to see that the maximum possible value of $\sum_{z_k}p_{z_k}\cdot\delta_{z_k}^2\leq-\mathbb{E}[\omega]\cdot(1-\mathbb{E}[\omega])\leq0.25$, where the equality is satisfied by a distribution where all the probability mass concentrates on the boundary. Therefore,
\begin{align}
\mathrm{Var}[\omega]&=\sum_{z_k}p_{z_k}\cdot\mathbb{E}\left[(\omega_{z_k}-\mathbb{E}\left[\omega\right])^2\right]\leq\sum_{z_k}p_{z_k}\cdot(\mathrm{Var}\left[\omega_{z_k}\right]+\delta_{z_k}^2)\leq3 + \sum_{z_k}p_{z_k}\cdot\delta_{z_k}^2 \leq 3.25.
\end{align}

\parbox{0.96\textwidth}{\begin{protocol_framed}[Improved level-$m$ shadow overlap with random Cliffords]\label{prot:random_clifford_improved}\quad 
\normalfont
\newline
\begin{flushleft}\underline{Input:}\end{flushleft}
 Same as Protocol \ref{prot:random_pauli_improved}.
\begin{flushleft}\underline{Procedure:}\end{flushleft}
All steps are the same as Protocol \ref{prot:random_pauli_improved} except step 3, which we describe here.
For qubits $k_1,\dots,k_m$, generate a circuit $U$ from the $m$ qubit Clifford group $\mathrm{Cl}(2^m)$ uniformly at random and apply it to the $m$ qubits. Then, perform computational basis measurement and record the resulting $m$-bit string $b$. Compute the classical shadow
\begin{align}
    \sigma=(2^m + 1)U^\dagger \vert b\rangle\langle b\vert U-\mathbb{I}.\label{eqn:clifford_shadow}
\end{align}
\end{protocol_framed}}

\begin{lemma}\label{lem:trace_equal_shadow_overlap_clifford}
For $\omega$ of Protocol \ref{prot:random_clifford_improved}, we have
\begin{align}
\mathrm{tr}(L\rho)=\mathbb{E}[\omega].
\end{align}
\end{lemma}
\begin{proof}
The proof is identical to Lemma \ref{lem:trace_equal_shadow_overlap_pauli}. Specifically, the difference between the two proofs only occurs during the step of
\begin{align}
\frac{\langle z_k\vert\rho\vert z_k\rangle}{\mathrm{tr}(\langle z_k\vert\rho\vert z_k\rangle)}=\underset{\mathrm{shadow}s~\sigma}{\mathbb{E}}[\sigma].
\end{align}
In Protocol \ref{prot:random_pauli_improved}, $\sigma$ is a shadow computed from random Pauli measurements according to Eq.~\ref{eqn:pauli_shadow}, and the expectation is over random Pauli basis and measurement outcomes. In Protocol \ref{prot:random_clifford_improved}, $\sigma$ is a shadow computed from random Clifford measurements according to Eq.~\ref{eqn:clifford_shadow}, and the expectation is over random Clifford circuits and measurement outcomes.
\end{proof}

\begin{theorem}\label{thm:shadow_implies_fidelity_clifford}
Theorem \ref{thm:shadow_implies_fidelity_pauli} also holds for $\omega$ as defined in Protocol \ref{prot:random_clifford_improved}.
\end{theorem}
\begin{proof}
The proof is identical to Theorem 4 of \cite{huang2024certifying}.
\end{proof}

\begin{theorem}\label{thm:improved_clifford_sample_complexity_estimator}
Using $T=(\frac{6.5}{\varepsilon^2}+\frac{2}{3\varepsilon}\cdot2^m)\cdot\ln\frac{2}{\delta}$ samples of the state $\rho$, the empirical average $\frac{1}{T}\sum_{t=1}^T\omega_t$ computed by Protocol \ref{prot:random_clifford_improved} is within $\varepsilon$ additive distance from the shadow overlap $\mathbb{E}[\omega]$ with probability at least $1-\delta$.
\end{theorem}
\begin{proof}
We first bound the individual values in the sum.
\begin{align}
\mathrm{tr}(L_{z_k}\sigma)&=\mathrm{tr}\left\{\vert\Psi_{z_k}\rangle\langle\Psi_{z_k}\vert\left[(2^m+1)U^\dagger\vert b\rangle\langle b\vert U-\mathbb{I}\right]\right\}\\
&=(2^m+1)\cdot\vert\langle b\vert U\vert\Psi_{z_k}\rangle\vert^2-1.
\end{align}
Therefore, $-1\leq\mathrm{tr}(L_{z_k}\sigma)\leq2^m$. Further, we know $\mathrm{Var}[\omega]\leq3.25$. By Bernstein's inequality,
\begin{align}
\Pr\left[\left\vert\frac{1}{T}\sum_{t=1}^T\omega_t-\mathbb{E}[\omega]\right\vert>\varepsilon\right]\leq2e^{-\frac{T\varepsilon^2}{2\times 3.25+\frac{2}{3}\times2^m\varepsilon}}\leq\delta.
\end{align}
Solving for $T$ completes the proof.
\end{proof}

\begin{theorem}\label{thm:improved_clifford_sample_complexity_certification}
Given any $n$-qubit hypothesis state $\vert\psi\rangle$ with a relaxation time $\tau\geq 1$, error $\varepsilon>0$, failure probability $\delta>0$, and
\begin{align}
T=\left(\frac{104\tau^2}{\varepsilon^2}+\frac{8\tau}{3\varepsilon}\cdot2^m\right)\cdot\ln\frac{2}{\delta}
\end{align}
samples of an unknown $n$-qubit state $\rho$, with probability at least $1-\delta$, Protocol \ref{prot:random_clifford_improved} with shadow overlap threshold $\omega_{\mathrm{threshold}}\equiv1-\varepsilon'=\frac{3\varepsilon}{4\tau}$ will correctly output \textsc{Failed} if the fidelity is low, i.e. $\langle\psi\vert\rho\vert\psi\rangle<1-\varepsilon$, and will correctly output \textsc{Certified} if the fidelity is high, i.e. $\langle\psi\vert\rho\vert\psi\rangle\geq1-\frac{\varepsilon}{2\tau}$.
\end{theorem}
\begin{proof}
The proof is identical to Theorem 6 of \cite{huang2024certifying} except we use Theorem \ref{thm:improved_clifford_sample_complexity_estimator}. Simply substituting $\varepsilon\rightarrow\varepsilon/\tau-\varepsilon'=\frac{\varepsilon}{4\tau}$ of Theorem 5 of \cite{huang2020predicting} or Theorem \ref{thm:improved_clifford_sample_complexity_estimator} of this work gives the sample complexity of Theorem 6 of \cite{huang2024certifying} or this theorem, respectively.
\end{proof}

We note that there is a substantial improvement in the sample complexity compared to Theorem \ref{thm:random_pauli_verification} for the random Pauli shadow overlap protocol. First, there is an apparent quadratic improvement from $2^{2m}$ to $2^m$. Moreover, the exponential factor in $m$ is additionally improved to proportional to $\frac{\tau}{\varepsilon}$ instead of $(\frac{\tau}{\varepsilon})^2$. The term proportional to $(\frac{\tau}{\varepsilon})^2$ only has a constant factor independent of $m$.

There is also a substantial improvement in sample complexity compared to Theorem \ref{thm:fixed_huang_verification_completeness_and_soundness} for the original shadow overlap protocol of \cite{huang2024certifying}, with a quartic improvement from $2^{4m}$ to $2^m$.

The exponential dependence on $m$ can be further removed by using the median of means estimator, which requires $\frac{64\mathrm{Var}\ln(1/\delta)}{\varepsilon^2}$ samples to estimate the mean of a random variance with error at most $\varepsilon$ and failure probability at most $\delta$. This is exactly as discussed in \cite{huang2020predicting}, and the sample complexity only depends on the variance. The disadvantage of the approach is that is has a large constant factor. Specifically, substituting $\varepsilon\rightarrow\frac{\varepsilon}{4\tau}$ and $\mathrm{Var}=3.25$ yields a sample complexity of $3{,}328\cdot\frac{\tau^2}{\varepsilon^2}\cdot\ln\frac{1}{\delta}$. The advantage of the median of means estimator would only be apparent for large $m$. Under the median of means estimator scheme, the special case of $m=n$ becomes exactly the random Clifford classical shadow approach in \cite{huang2020predicting} for estimating the fidelity.

This makes apparent the fact that our protocol is a unification of the direct classical shadow approach of \cite{huang2020predicting} and the shadow overlap approach of \cite{huang2024certifying} for certifying the quantum state fidelity. On the one extreme, the direct classical shadow approach using random Clifford circuits is sample efficient in that the sample complexity does not increase with $n$, and states fidelities can be estimated in an unbiased manner since the relaxation time $\tau$ is $1$. One can replace random Cliffords with random single-qubit Pauli measurements and the estimator is unbiased since the relaxation time is also 1, but the sample complexity is exponential in $n$. On the other extreme, the shadow overlap approach that measures one random qubit in a random Pauli basis achieves low sample complexity that is independent of $n$, but the relaxation time $\tau$ is large and low fidelity states cannot be \textsc{Certified}.

Our work bridges the gap by introducing a family of protocols that connect the two extremes. We can go from the extreme of large relaxation time to the minimum relaxation time by increasing the level $m$ of the protocol. Additionally, the sample complexity is $O(2^{2m})$ for the random Pauli protocol that preserves the desirable property that the measurement is easy to implement. For the random Clifford protocol, the sample complexity is $O(2^m)$ or independent of $m$ for the mean estimator or the median of means estimator, respectively. The complexity of implementing the measurement, which is $O(m^2)$ since implementing a random Clifford circuit requires quadratically many gates, can be controlled by keeping $m$ moderate.

Although \cite{huang2020predicting} also introduces a family of protocols parameterized by the level $m$, the choice of operator $L$ does not lead to a correspondence with the direct random Clifford classical shadow approach of \cite{huang2020predicting} even when $m=n$, and the relaxation time remains large as $m$ increases. Further, the sample complexity is $O(2^{4m})$, which prohibits practical use for large $m$.

\subsection{Certifying the state in our experiment}\label{sec:supp_cert_experiment}

The sample complexity used in Theorem \ref{thm:improved_clifford_sample_complexity_certification} is for both completeness and soundness. However, once the experimental results on the shadow overlap protocol is collected, one can carry out a hypothesis test to see if the lab state has fidelity at least $1-\varepsilon$ and not worry about the completeness side. This motivates the following corollary.

\begin{corollary}\label{cor:sample_complexity_soundness}
If $\varepsilon\geq\tau\varepsilon'$ and the number of shadows collected by Protocol \ref{prot:random_clifford_improved} is $T\geq(\frac{6.5}{(\varepsilon/\tau-\varepsilon')^2}+\frac{2}{3\times(\varepsilon/\tau-\varepsilon')}\cdot 2^m)\cdot\ln\frac{1}{\delta}$, then either the protocol outputs \textsc{Failed} with probability at least $1-\delta$, or $\langle\psi\vert\rho\vert\psi\rangle\geq1-\varepsilon$.
\end{corollary}
\begin{proof}
If $\langle\psi\vert\rho\vert\psi\rangle<1-\varepsilon$, then $\mathbb{E}[\omega]<1-\varepsilon/\tau$. To pass the protocol, we need $\hat{\omega}-\mathbb{E}[\omega]>\varepsilon/\tau-\varepsilon'$. From the proof of Theorem~\ref{thm:improved_clifford_sample_complexity_estimator}, switching from two-sided to one-sided Bernstein's inequality yields
\begin{align}
\Pr\left[\frac{1}{T}\sum_{t=1}^T\omega_t-\mathbb{E}[\omega]>\varepsilon\right]\leq e^{-\frac{T\varepsilon^2}{2\times 3.25+\frac{2}{3}\times2^m\varepsilon}}.\label{eqn:one_sided_bernstein}
\end{align}
Then, substituting $\varepsilon$ in the Eq.~\ref{eqn:one_sided_bernstein} with $\varepsilon/\tau-\varepsilon'$ shows that the protocol must output \textsc{Failed} with probability at least $1-\delta$.
\end{proof}

For our experiment, we are not in the ideal perfect honest fidelity case, so we need Lemma \ref{lem:depolarizing_noise} to bound the infidelity of the adversary state relative to the laboratory state. Assume that the adversary can prepare a state $\lvert\psi_{\rm{adv}}\rangle$ with fidelity $\vert\langle\psi_{\rm{adv}}|\psi_{\rm{hon}}\rangle\vert^2\leq1-\varepsilon_{\rm{CNL}}=0.01$. Further, we have empirically estimated that the honest lab state infidelity is at least $\varepsilon_{\rm{hon}}\leq0.11$, Lemma \ref{lem:depolarizing_noise} implies that the adversary state fidelity to the lab state is
\begin{align}
\vert\langle0\vert C'^{\dagger}\sigma_{\mathrm{lab}}C'\vert0\rangle\vert^2\leq1-(\varepsilon_{\mathrm{CNL}}-\varepsilon_{\mathrm{hon}})+\varepsilon_{\mathrm{hon}}\frac{1}{2^n}=0.12
\end{align}
under depolarizing noise. Therefore, we set $\varepsilon=1-0.12=0.88$. The shadow overlap test threshold is $\omega_{\mathrm{threshold}}=1-\varepsilon'=0.91$, which means $\varepsilon'=0.09$. Finally, the probability (risk) that a given adversary circuit exceeds the shadow overlap threshold is $\delta$. For this set of parameters, one can achieve $\delta=0.36$ at $T=17{,}316$ samples.

The above corollary cannot guarantee that the state is certified with good security in our particular experiment since $\delta$ is very large.
As a result, we use the following result of Zhang \textit{et al.}~\cite{zhang2026efficient} to obtain a tighter bound, which allows us to go to as low as $\delta=0.15$. Ref.~\cite{zhang2026efficient} presents an efficient method to estimate bounds on the sum of conditional means of data random variables $X_1, X_2, \dotsc$ during unseen trials by spot-checking a few trials. This can be mathematically phrased as having a sequence of binary (spot-checking) random variables $Y_1, Y_2, \dotsc$ that determine whether (or not) we have access to the value of $X_i$ at a given trial $i$ with probability $p$ (or $1 - p$). Then, using the observed values of $X_i$, Zhang~\textit{et al.}~\cite{zhang2026efficient} give a protocol to estimate bounds on the sum $S$ of conditional means of the data random variables for the unseen trials, which satisfies the following theoretical guarantee.
\begin{theorem}\label{thm:unpublished}
The procedure described in Ref.~\cite{zhang2026efficient} takes as input the observed values of data random variables $\boldsymbol{X} = \{X_i\}_i$, spot-checking random variables $\boldsymbol{Y} = \{Y_i\}_i$, spot-checking probability $p$, optional calibration trials $\boldsymbol{X}_c$, lower and upper bounds $x_{\min} \leq X_i \leq x_{\max}$ for all $i$, a confidence parameter $\delta$, and outputs a lower bound $b(\delta, x_{\min}, x_{\max}, \boldsymbol{X}_c, \boldsymbol{X}, \boldsymbol{Y})$ on the sum $S$ of conditional means of unobserved trials such that
\begin{align}
\underset{\boldsymbol{X}_c, \boldsymbol{X}, \boldsymbol{Y}}{\Pr}\left[S < b(\delta, p, x_{\min}, x_{\max}, \boldsymbol{X}_c, \boldsymbol{X}, \boldsymbol{Y})\right]<\delta.
\end{align}
\end{theorem}
The notation $b(\delta, p, x_{\min}, x_{\max}, \boldsymbol{X}_c, \boldsymbol{X}, \boldsymbol{Y})$ in Thm.~\ref{thm:unpublished} implicitly assumes that only the observed data is used to calculate the lower bound. Note that Thm.~\ref{thm:unpublished} also holds when the trials are stopped according to some rule determined by data available until that trial.

In this study, the data corresponds to independent and identically distributed samples $\{\omega_i\}_{i \in [T]}$, all of which have the same mean $\mathbb{E}[\omega]$. Since all the random variables have the same mean and we have access to a given number $T$ of samples, we specialize the results of Ref.~\cite{zhang2026efficient} to this case and use the following protocol for estimating a lower bound on $\mathbb{E}[\omega]$.

\parbox{0.96\textwidth}{\begin{protocol_framed}[Lower bound estimate on shadow overlap using Ref.~\cite{zhang2026efficient}]\label{prot:lower_bound_protocol_zhang}\quad 
\normalfont
\newline
\begin{flushleft}\underline{Input:}\end{flushleft}
\vspace{-1.5em}
\begin{itemize}
    \item Number of samples $T$
    \item Number of calibration samples $T_c$
    \item Data $\{\omega_i\}_{i \in [T]}$
    \item Bounds $\omega_{\min} \leq \omega_i \leq \omega_{\max}$ for all $i \in [T]$
    \item Confidence/security parameter $\delta$
    \item Parameter $p \in (0, 1)$
\end{itemize}
\begin{flushleft}\underline{Procedure:}\end{flushleft}
\vspace{-1.5em}
\begin{enumerate}
    \item Sample Bernoulli random variables $\bm{Y} = \{Y_i\}_i$ such that $Y_i = 0$ with probability $p$ for all $i$, until the first trial $N \geq T$ where $\sum_{i = 1}^N (1 - Y_i) = T - T_c$.
    \item Associate $\{\omega_{k + T_c}\}_{k \in [T - T_c]}$ with the observed data random variables, corresponding to trials $i \in [N]$ with $Y_i = 0$, and use $\bm{X}_c = \{\omega_k\}_{k \in [T_c]}$ for calibration. For $i \in [N]$, denote $\bm{X}$ as the list of $N$ samples such that $X_i$ is a data point in $\{\omega_{k + T_c}\}_{k \in [T - T_c]}$ if $Y_i = 0$ and $X_i = \bot$ otherwise.
    \item Use the protocol of Ref.~\cite{zhang2026efficient} to output a lower bound $b(\delta, p, \omega_{\min}, \omega_{\max}, \bm{X}_c, \bm{X}, \bm{Y})$ on $\mathbb{E}[\omega]$.
\end{enumerate}
\end{protocol_framed}}

The notation $\bot$ is used to denote that no data is available when $Y_i = 1$. We use a moderately small number $T_c = 500$ of calibration trials. Then, using the results of Ref.~\cite{zhang2026efficient}, we obtain the following theoretical guarantee for Protocol~\ref{prot:lower_bound_protocol_zhang}.
\begin{theorem}
    \label{thm:lower_bound_protocol_zhang}
    Protocol~\ref{prot:lower_bound_protocol_zhang} takes as input the list of $T$ observed values $\{\omega_i\}_{i\in[T]}$, number of calibration trials $T_c$, lower and upper bounds for $\omega_{\min} \leq \omega_i \leq \omega_{\max}$ for all $i\in[T]$, a confidence/security parameter $\delta$, and a parameter $p \in (0, 1)$, and outputs a lower bound $b$ such that
    \begin{align}
    \Pr[\mathbb{E}[\omega] < b] < \delta, \label{eq:lower_bound_protocol_zhang_guarantee}
    \end{align} 
    where the probability is over $\{\omega_i\}_{i \in [T]}$ and the internal randomness of Protocol~\ref{prot:lower_bound_protocol_zhang}.
\end{theorem}
We choose the parameter $p$ to be a small value, $p = 0.01$, based on the asymptotic performance analysis of Ref.~\cite{zhang2026efficient}.
\begin{proof}
    Consider independent and identically distributed (iid) random variables $\omega'_1, \omega'_2, \dotsc$ with the same distribution as $\omega_1$. In particular, the mean of these random variables is $\mathbb{E}[\omega]$. Protocol~\ref{prot:lower_bound_protocol_zhang} can be thought of as a sequential procedure where we run the experiment by drawing $Y_i$ at trial $i$ and observing $\omega'_i$ when $Y_i = 0$. These $\omega'_i$ that we observe can be associated with the data we have, since they are iid. The experiment is stopped at the first trial $N$ when $\sum_{i = 1}^N (1 - Y_i) = T - T_c$. Then, Eq.~\eqref{eq:lower_bound_protocol_zhang_guarantee} follows from Theorem~\ref{thm:unpublished}, noting in particular that the guarantee holds even when the experiment is stopped according to such a rule.
\end{proof}

Theorem~\ref{thm:lower_bound_protocol_zhang} is essentially the counter part of Theorem~\ref{thm:improved_clifford_sample_complexity_estimator} except no explicit analytical bound is given, and the bound is instead given by a numerical procedure. To see this connection more clearly, one can rephrase Theorem~\ref{thm:improved_clifford_sample_complexity_estimator} as a special case of Theorem~\ref{thm:lower_bound_protocol_zhang} with
\begin{align}
b=\frac{1}{T}\sum_{t=1}^T\omega_t+\varepsilon,
\end{align}
where $\varepsilon$ can be obtaining by solving the following equality using the quadratic formula:
\begin{align}
e^{-\frac{T\varepsilon^2}{2\times 3.25+\frac{2}{3}\times2^m\varepsilon}}=\delta,
\end{align}
which comes from switching from two-sided to one-sided Bernstein's inequality just as in Corollary~\ref{cor:sample_complexity_soundness}. The procedure in Ref.~\cite{zhang2026efficient}, on the other hand, computes $b$ in a much more sophisticated manner exploiting information present in individual data points instead of simply using the aggregate information.

However, phrased in this way, Protocol
~\ref{prot:random_clifford_improved} is no longer well defined, since we are only given access to the lower bound $b$ and not the estimator $\hat{\omega}$ anymore, which is needed in step 7 when determining whether to output \textsc{Certified} or \textsc{Failed}. This requires us to formulate an alternate definition of the protocol.

\parbox{0.96\textwidth}{\begin{protocol_framed}[Improved level-$m$ shadow overlap with random Cliffords]\label{prot:using_lower_bound_procedure}\quad 
\normalfont
\newline
\begin{flushleft}\underline{Input:}\end{flushleft}
 Same as Protocol \ref{prot:random_clifford_improved}.
\begin{flushleft}\underline{Procedure:}\end{flushleft}
All steps are the same as Protocol \ref{prot:random_clifford_improved} except step 6 and 7, which we describe here.

For step 6, compute the lower bound $b$.

For step 7, if $b\geq 1-\varepsilon'$, output \textsc{Certified}. Otherwise, output \textsc{Failed}.

\end{protocol_framed}}

It should be noted that in Protocol~\ref{prot:random_clifford_improved}, $\omega_{\mathrm{threshold}}=1-\varepsilon'$ is set such that it is slightly above the expected shadow overlap of the adversarial case $\mathbb{E}[\omega_{\mathrm{adv}}]$. This is to allow for some fluctuations in the empirical estimator $\hat{\omega}$. On the other hand, for Protocol~\ref{prot:using_lower_bound_procedure}, the lower bound $b$ is already estimated to allow for statistical fluctuations due to the samples, so $\omega_{\mathrm{threshold}}=1-\varepsilon'$ should be set to exactly $\mathbb{E}[\omega_{\mathrm{adv}}]$. Specifically, for our experiment, we require $b>0.89$. Feeding our experimental data into the function $b$ with $\delta=0.15$ exceeds this threshold.

\newpage
\section{Quantum error detection}
\label{sec:qed}

The practical execution of quantum algorithms on current hardware is limited by noise. The conventional approach is quantum error correction (QEC), which encodes logical information in more physical qubits and performs syndrome extraction to detect and correct errors. However, QEC typically incurs substantial overhead in qubit count and circuit depth, making it difficult to implement on near-term devices. As an alternative, quantum error detection (QED) codes have emerged as a promising near-term strategy \cite{he2025performance, jin2025Icebergtipcocompilationquantum,chung2025fault,self2024protecting}. Rather than correcting errors, QED codes identify and discard erroneous outcomes \cite{Steane_iceberg}.

\subsection{Iceberg code}

The Iceberg code is a stabilizer-based quantum error detection code that encodes $k$ logical qubits into $n = k + 2$ physical qubits. 
Its stabilizer group $\mathcal{S}$ is generated by 
the $X$ and $Z$ stabilizers $S_X = \bigotimes_{i=1}^{n} X_i$ and $S_Z = \bigotimes_{i=1}^{n} Z_i$.
The code space is defined as the joint $+1$ eigenspace of these stabilizers:
\begin{equation}
\mathcal{H}_{\mathcal{S}} = \left\{ \ket\psi \in (\mathbb{C}^2)^{\otimes n} : S_X \ket\psi = \ket\psi, \ S_Z \ket\psi = \ket\psi \right\}.
\end{equation}
Logical single and two qubit rotation for the Iceberg code are 
\begin{align}
\overline{X}_i(\theta) &= X_t X_i(\theta);\quad \overline{X}_i\overline{X}_j(\theta) = X_i X_j(\theta)\nonumber \\
\overline{Z}_i(\theta) &= Z_b Z_i(\theta);\quad \overline{Z}_i\overline{Z}_j(\theta) = Z_i Z_j(\theta),
\label{eq:logical_single_qubit_gate}
\end{align}
where for convenience we label $1$-st and $k+2$-nd qubits to be top $t$ and bottom $b$ respectively (see Fig.~\ref{fig:encoder_Iceberg}\textbf{a}).
The logical Hadamard operator on all $k$ logical qubits is
\begin{equation}
\prod_{i=1}^{k} \overline{H_i} = \text{SWAP}_{t,b} \prod_{i=1}^{n} H_i.
\end{equation}

A generic encoding circuit for the Iceberg code encoding an arbitrary $k$-qubit state $\ket{\psi}$ is given in Fig.~\ref{fig:encoder_Iceberg}\textbf{a}.
Fig.~\ref{fig:encoder_Iceberg}\textbf{b}
 is the fault-tolerant encoding circuit for the logical state $\ket{+}^{\otimes k}$ with a single additional ancilla qubit.

\begin{figure}
\includegraphics[width=\textwidth]{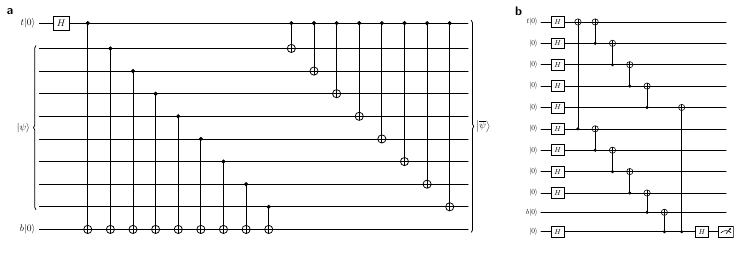}
\caption{\textbf{Encoder circuits for Iceberg code.}
 \textbf{a}, is the circuit that encodes an arbitrary $k$-qubit physical state $\ket{\psi}$ to the logical state $\ket{\overline{\psi}}$ which  requires in total of $2k+1$ physical CNOT gates.
 \textbf{b}, is the fault-tolerant encoding circuit for the state logical state $\ket{+}^{\otimes k}$ using a single ancilla qubit.
 This circuit requires a total of $k+3$ CNOT gates and two-qubit depth of $\lceil k/2\rceil+3$.}
  \label{fig:encoder_Iceberg}
\end{figure}

\subsection{Gauge fixing the Iceberg code}

For the $n$-qubit Pauli group $\mathcal{P}_n$, a stabilizer code is specified by an abelian subgroup $\mathcal{S}\subset\mathcal{P}_n$ generated by independent, commuting Pauli operators with $-I\notin\mathcal{S}$. 
The code space is the joint $+1$ eigenspace
\begin{equation}
\mathcal{H}_{\mathcal{S}} = \left\{\, \ket\psi \in (\mathbb{C}^2)^{\otimes n} : S\ket\psi = \ket\psi \ \forall S\in\mathcal{S}\,\right\}.
\end{equation}
Let $\mathcal{C}(\mathcal{S}) = \left\{ P \in \mathcal{P}_n : [P, S] = 0 \ \forall S \in \mathcal{S} \right\}$ denote the centralizer of $\mathcal{S}$, i.e. the elements of $\mathcal{P}_n$ that commute with all the stabilizers of the code.

\subsubsection{Gauge fixing}

Given a stabilizer code $(\mathcal{S}, \mathcal{H}_{\mathcal{S}})$ and a set of commuting, independent Pauli operators $\mathcal{G} = \langle g_1, \dots, g_r \rangle$ with $
g_i \in \mathcal{C}(\mathcal{S}) / \mathcal{S}  \text{ and } \comm{g_i}{g_j} = 0,
$
the \emph{gauge-fixed} stabilizer group is
\begin{equation}
\mathcal{S}' = \langle \mathcal{S}, \mathcal{G} \rangle,
\end{equation}
where $\langle \cdot \rangle$ denotes the group generated by all finite products of elements from the set. 
This defines a subcode with code space $\mathcal{H}_{\mathcal{S}'} \subseteq \mathcal{H}_{\mathcal{S}}$ and $k-r$ logical qubits. 
Effectively, gauge-fixing projects the code space onto the joint $+1$ eigenspace of $\mathcal{G}$ inside $\mathcal{H}_{\mathcal{S}}$:
\begin{equation}
\Pi_{\mathcal{G}} = \prod_{i=1}^{r} \frac{I + g_i}{2}, \qquad
\mathcal{H}_{\mathcal{S}'} = \Pi_{\mathcal{G}} \mathcal{H}_{\mathcal{S}}.
\end{equation}

\subsubsection{Logical gauge fixing}
\label{subsubsec:logical_gauge_fixing}
Let $\{\overline{X}_j, \overline{Z}_j\}_{j=1}^{k}$ be a generating set of logical Pauli operators for $(\mathcal{S}, \mathcal{H}_{\mathcal{S}})$.
\emph{Fixing a logical} means promoting a subset of these logical operators to stabilizers. For example, fixing $\overline{Z}_a$ yields
\begin{equation}
\mathcal{S}' = \langle \mathcal{S}, \overline{Z}_a \rangle, 
\qquad \mathcal{H}_{\mathcal{S}'} = \left\{\, \ket\psi \in \mathcal{H}_{\mathcal{S}} : \overline{Z}_a \ket\psi = \ket\psi \,\right\},
\end{equation}
which fixes the $a$-th logical qubit to $|0_L\rangle$ and reduces the number of encoded qubits to $k-1$.
Alternatively, fixing $\overline{X}_a$ prepares $|+_L\rangle$ on the $a$-th logical.
For the Iceberg code, we use logical gauge fixing and promote some of the logical operators to be stabilizers. Because the resulting gauge-fixed code $\mathcal{S}'$ results from keeping $r$ of the logical qubits in the parent code $\mathcal{S}$ in $\ket{0_\mathrm{L}}$ or $\ket{+_\mathrm{L}}$, the fault tolerant encoding circuits for useful states such as $\ket{+_\mathrm{L}}^{\otimes k-r}$ need to be modified relative to Fig.~\ref{fig:encoder_Iceberg} (b). An example of this, shown in
Fig.~\ref{fig:initial_state_gauge}, prepares $\ket{+_\mathrm{L}}^{\otimes k-1}$ in the gauge-fixed code which fixes $k$-th logical qubit to $\ket{0_\mathrm{L}}$.

\begin{figure}
    \centering
    \includegraphics[width=0.4\linewidth]{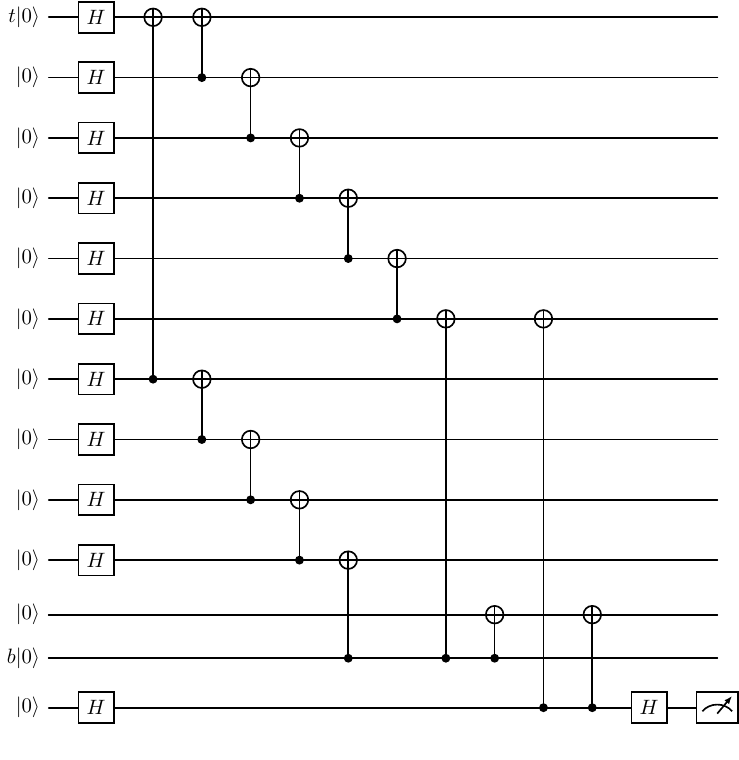}
    \caption{\textbf{Fault tolerant initialization circuit for the gauge-fixed Iceberg code.}
    The initialization  circuit for case when we fix the $(k-1)$-th logical qubit to $|+_\mathrm{L}\rangle$ and $k$-th logical qubit in $\ket{0_\mathrm{L}}$ of the Iceberg code which encode $k$ logical qubits.
    This circuit requires a total of $k+4$ CNOT gates and two-qubit depth of $\lceil k/2\rceil+4$.
    }
    \label{fig:initial_state_gauge}
\end{figure}

\subsubsection{Parallelization of single qubit gates}
Many quantum algorithms contain layers of single-qubit rotations acting on each of the qubits. In the standard Iceberg code, these rotations get compiled down to $X$- and $Z$-rotations; from Eq.~\eqref{eq:logical_single_qubit_gate} a single-qubit $X$ rotation by angle $\theta$ on logical qubit $i$ is physically implemented via $e^{-i\frac{\theta}{2} X_t \otimes X_i}$. 
Therefore, every logical $X$ rotation involve the top qubit. 
If a layer contains many logical $X$ rotations, they must be run sequentially because they share the top qubit .
Likewise, all single-qubit $Z$ rotations involve the bottom qubit $b$ and must be run sequentially.

Suppose that instead of using the standard $\llbracket k+2, k, 2\rrbracket$ Iceberg code, which encodes $k$ logical qubits into $k+2$ physical qubits and has code distance 2, we gauge fix $\ell$ qubits to fixed logical states while keeping the number of computational qubits $k$ constant, which in turn gives us a $\llbracket k+\ell+2, k, 2\rrbracket$ code.
Say we fix the logical qubit $g_z$ to be in the $+1$ eigenstate of $\overline{Z}_{g_z} = Z_b Z_{g_z}$.
This is equivalent to saying that logical qubit $g_z$ is gauge-fixed to the $|0\rangle_L$ state. Thus for any other unfixed logical qubit $q$,
\begin{align}
    \overline{Z}_q = Z_b Z_q = Z_b Z_q \cdot Z_b Z_{g_z} = Z_{g_z} Z_q.
\end{align}
So we can use the new gauge-fixed qubit $g_z$ as a proxy bottom qubit.
Similarly, if we gauge-fix logical qubit $g_x$ to $|+\rangle_L$ by adding $X_t X_{g_x}$ to the stabilizer group, for any other unfixed logical qubit $q$ we can write
\begin{align}
    \overline{X}_q = X_t X_q = X_t X_q \cdot X_t X_{g_x} = X_{g_x} X_q.
\end{align}
And so we can use gauge qubit $g_x$ as a proxy top qubit.

We can use these proxy qubits to parallelize our single qubit gate layers.
For example, to perform $\overline{X}_{q_1}(\theta_1)$ and $\overline{X}_{q_2}(\theta_2)$ in the same layer.
We can now run $\overline{X}_{q_1} \cong X_t X_{q_1}$ and $\overline{X}_{q_2} \cong X_{g_x} X_{q_2}$ in parallel, since they share no qubits. 
Given a layer of single-qubit $X$ rotations of width $w_x$ and $n_x$ gauge qubits, the circuit depth of the single-qubit rotation layer is reduced from $w_x$ to $w_x / n_x$. 
The same is true for the $Z$-type checks.

\subsection{Experimental Results}
\label{sec:qed_results}

Experiments were performed on the Quantinuum  H2-1 and H2-2 quantum hardware devices, both with a total of $56$ physical qubits.
We experimentally demonstrate the effect of QED on the shadow overlap and the fidelity. 
For the shadow overlap, we use the improved protocol with single-qubit random Pauli bases measurements. 
For the estimated fidelity, we apply the transformation $\hat{\phi}=2\hat{\omega}-1$, where $\hat{\phi}$ is the estimated fidelity and $\hat{\omega}$ is the empirical estimate of the average shadow overlap. This is approximately normalized shadow overlap in Ref.~\cite{huang2024certifying} up to a $2^{-n}$ correction. 
For $m$-level shadow overlap protocols, the transformation is approximately $\hat{\phi}=\frac{2^m}{2^m-1}(\hat{\omega}-\frac{1}{2^m})$. 
We refer to Ref.~\cite{huang2024certifying} on why this is a reasonable estimate on experimental fidelity.

Table~\ref{tab:shadow_overlap} summarizes the shadow overlap and estimated fidelity for both unencoded (physical) and Iceberg-encoded circuits.
For both devices, we observe that increasing the number of blocks leads to higher shadow overlap and fidelity.
For example, on the H2-1 device, encoding 24 logical qubits into four blocks yields a shadow overlap of $1.02 \pm 0.06$ and an estimated fidelity of $1.04 \pm 0.10$, compared to $0.90 \pm 0.04$ and $0.80 \pm 0.09$ for the physical circuit.
For more details see \Cref{tab:shadow_overlap}.

\begin{table*}[!ht]
\centering
 \begin{tabular}{|c|c|c|c|c|c|c|c|c|c|c|c|c|}
 \hline
 Device & Type & Shots & \makecell{Logical \\qubits} & Blocks & \makecell{Logical \\2-Q depth} &\makecell{Physical \\qubits} &\makecell{Physical \\2-Q depth}
 & \makecell{Physical \\2-Q gates} 
 &\makecell{Acceptance \\rate (\%)}  
 & \makecell{Shadow\\ overlap} & \makecell{Estimated \\Fidelity} \\
 \hline
 H2-1 & Physical & 500 & 24 & N/A & 5 &24 &5 & 60&100 & $0.90 \pm 0.04$ & $0.80 \pm 0.09$ \\
 \hline
 H2-1 & Encoded & 2000 & 24 & 2 & 5 & 30& 96& 432 &  48 & $1.00 \pm 0.03$ & $0.99 \pm 0.07$ \\
 \hline
 H2-1 & Encoded & 2000 & 24  & 4 & 5 &34&61  &444&17 & $1.02 \pm 0.06$ & $1.04 \pm 0.10$ \\
  \hline
 H2-2 & Physical & 500 & 32 & N/A & 5 &32 & 5&80 & 100 & $0.91 \pm 0.03$ & $0.82 \pm 0.07$ \\
  \hline
 H2-2 & Encoded & 2000 & 32 & 2 & 5 & 38 &128 &572&37  & $0.90 \pm 0.03$ & $0.79 \pm 0.06$ \\
  \hline
 H2-2 & Encoded & 2000 & 32 & 4 & 5 & 42&75 &584&39 &$0.99 \pm 0.02$ & $0.98 \pm 0.05$ \\
  \hline
 H2-2 & Physical & 700 & 40 & N/A & 6 & 40&6 &120&100 & $0.91 \pm 0.03$ & $0.82 \pm 0.06$ \\
  \hline
 H2-2 & Encoded & 3000 & 40 & 4 & 6 & 50& 86&724& 23   &$0.97 \pm 0.03$ & $0.94 \pm 0.06$ \\
 \hline
 \end{tabular}
  \caption{\textbf{Shadow overlap results.} 
  Shadow overlap values and acceptance rates for physical (unencoded) and Iceberg-encoded random quantum circuits on H2-1 and H2-2 Quantinuum hardware.}
   \label{tab:shadow_overlap}
\end{table*}

 Table~\ref{tab:shadow_gauge} compares regular QED encoding and gauge-fixing strategies on the H2-2 device. 
 For circuits with 32 logical qubits and two encoded blocks, gauge fixing improves the shadow overlap from $0.90 \pm 0.03$ to $0.97 \pm 0.03$ and the estimated fidelity from $0.79 \pm 0.06$ to $0.94 \pm 0.05$, while maintaining a similar acceptance rate. 
 Although for $40$ logical qubits and $2$ blocks, the gauge-fixed fidelity is lower than the physical fidelity, gauge fixing still provides significant improvements over regular encoded circuits. 
 If experiments with $40$ logical qubits and $4$ gauge-fixed blocks were feasible, it is plausible that the fidelity would surpass that of regular encoded circuits, which already achieve $0.94$ as shown in Table~\ref{tab:shadow_overlap}.
The gauge-fixed $40$ logical qubit, $4$ block experiment was not performed as it required more physical qubits than available on the device. 
Implementing this configuration would require $58$ physical qubits---$4$ blocks of $10$ logical qubits each, with $2$ gauge-fixed and $2$ additional physical qubits per block, plus two ancilla qubits for syndrome checks---exceeding the available qubit count on the H2 devices.

\begin{table*}[!ht]
 \centering
 \begin{tabular}{|c|c|c|c|c|c|c|c|c|c|c|c|c|}
 \hline
 Device & Type & Shots & \makecell{Logical \\qubits} & Blocks & \makecell{Logical \\2-Q depth} &\makecell{Physical \\qubits} &\makecell{Physical \\2-Q depth}
 & \makecell{Physical \\2-Q gates} 
 &\makecell{Acceptance \\rate (\%)}  
 & \makecell{Shadow\\ overlap} & \makecell{Estimated \\Fidelity} \\
 \hline
 H2-2 & Physical & 500 & 32 & N/A & 5 &32&5&80 & 100 & $0.91 \pm 0.03$ & $0.82 \pm 0.07$ \\
 \hline
 H2-2 & Encoded & 2000 & 32 & 2 & 5 & 38&128&572&37 & $0.90 \pm 0.03$ & $0.79 \pm 0.06$ \\
 \hline
 H2-2 & Gauge-fixed & 2000 & 32 & 2 & 5 &42&88&590& 38 & $0.97 \pm 0.03$ & $0.94 \pm 0.05$ \\
 \hline
 H2-2 & Physical & 700 & 40 & N/A & 6 &40 &6&120&100 & $0.91 \pm 0.03$ & $0.82 \pm 0.07$ \\
 \hline
 H2-2 & Encoded & 3000 & 40 & 2 & 6 &46&155&712& 21 & $0.82 \pm 0.03$ & $0.63 \pm 0.07$ \\
 \hline
 H2-2 & Gauge-fixed & 3000 & 40 & 2 & 6 & 50&108&730&16 & $0.88 \pm 0.04$ & $0.75 \pm 0.07$ \\
 \hline
 \end{tabular}
 \caption{\textbf{Gauge fixing.} 
 Comparison of shadow overlap and acceptance rates for regular QED and gauge-fixing strategies on H2-2 hardware. 
 For circuits with 32 logical qubits and 5 layers, gauge fixing yields higher shadow overlap than regular QED encoding. }
  \label{tab:shadow_gauge}
\end{table*}

\newcommand{\skgen}{\mathsf{SKGen}}
\newcommand{\pkgen}{\mathsf{PKGen}}
\newcommand{\sign}{\mathsf{Sign}}
\newcommand{\ver}{\mathsf{Ver}}
\newcommand{\secparam}{\lambda}
\newcommand{\sk}{\mathsf{sk}}
\newcommand{\pk}{\mathsf{pk}}
\newcommand{\negl}{\mathsf{negl}}
\newcommand{\accept}{\textsc{CERTIFIED}}
\newcommand{\reject}{\textsc{FAILED}}
\newcommand{\alice}{\mathcal{A}}
\newcommand{\puzzlegen}{\mathsf{PuzzleGen}}
\newcommand{\puzzlever}{\mathsf{PuzzleVer}}

\clearpage
\newpage
\section{One-time quantum digital signatures and one-way puzzles}\label{sec:microcrypt}

This section elaborates on the discussion at the end of the main text, and formally proves an equivalence between one-time quantum digital signatures and so-called one-way puzzles. This equivalence broadens the potential applications of quantum digital signatures to other kinds of cryptographic tasks.

\subsection{Definitions}

Below, we give formal definitions for one-time quantum digital signatures and one-way puzzles in the inefficient verifier setting.

\begin{definition}[One-Time Signature Scheme]
    A \emph{one-time (quantum) signature scheme} is a tuple of algorithms\\ $(\skgen, \pkgen, \sign, \ver)$ with the following syntax: \begin{itemize}
        \item $\skgen(1^\secparam)$ is a probabilistic polynomial-time (PPT) algorithm that takes as input a security parameter $\secparam$ in unary and outputs a secret key $\sk$.
        \item $\pkgen(\sk)$ is a quantum polynomial-time (QPT) algorithm that takes as input a secret key $\sk$ and outputs a classical public key $\pk$.
        \item $\sign(\sk, m)$ is a PPT algorithm that takes as input a secret key $\sk$ and a message $m$; it outputs a classical signature $\sigma$.
        \item $\ver(\pk, \sigma, m)$ is a (possibly inefficient) classical algorithm that takes as input a public key $\pk$, a signature $\sigma$, and a message $m$; it outputs either $\reject$ or $\accept$.
    \end{itemize}

    It satisfies the following properties:
    \begin{enumerate}
        \item \textbf{Completeness:} There exists a negligible function $\negl$ such that for any message $m$ we have \begin{align}
            \Pr \left[b = \accept : \substack{ \sk \leftarrow \skgen(1^\secparam) \\ \pk \leftarrow \pkgen(\sk) \\ \sigma \leftarrow \sign(\sk, m) \\ b \leftarrow \ver(\pk, \sigma, m) } \right] \ge 1 - \negl(\secparam).
        \end{align}
        \item \textbf{One-Time Security:} There exists a negligible function $\negl$ such that for any QPT algorithms $\alice_1, \alice_2$ we have \begin{align}
            \Pr \left[m' \ne m \land b = \accept : \substack{ \sk \leftarrow \skgen(1^\secparam) \\ \pk \leftarrow \pkgen(\sk) \\ m \leftarrow \alice_1(\pk) \\ \sigma \leftarrow \sign(\sk, m) \\ (\sigma',m') \leftarrow \alice_2(\pk, \sigma) \\ b \leftarrow \ver(\pk, \sigma', m') } \right] \le \negl(\secparam).
        \end{align}
    \end{enumerate}
\end{definition}
    \noindent Note that the only ``quantum" parts are in the public key generation and the assumptions about the adversary.

\begin{definition}[One-Way Puzzle]
    A one-way puzzle is a tuple of algorithms $(\puzzlegen, \puzzlever)$ with the following syntax: \begin{itemize}
        \item $\puzzlegen(1^\secparam)$ is a QPT algorithm that takes as input a security parameter $\secparam$ in unary; outputs a secret key $s$ and a puzzle $p$.
        \item $\puzzlever(s,p)$ is a possibly inefficient classical algorithm that takes as input a secret key $s$ and a puzzle $p$; it outputs $\accept$ or $\reject$.
    \end{itemize}
    It satisfies the following properties:
    \begin{enumerate}
        \item \textbf{Completeness:} There exists a negligible function $\negl$ such that we have \begin{align}
            \Pr \left[b = \accept : \substack{ (s,p) \leftarrow \puzzlegen(1^\secparam) \\ b \leftarrow \puzzlever(s,p) } \right] \ge 1 - \negl(\secparam).
        \end{align}
        \item \textbf{Soundness:} There exists a negligible function $\negl$ such that for any QPT algorithm $\alice$, we have \begin{align}
            \Pr \left[b = \accept : \substack{ (s,p) \leftarrow \puzzlegen(1^\secparam) \\ s' \leftarrow \alice(p) \\ b \leftarrow \puzzlever(s',p) } \right] \le \negl(\secparam).
        \end{align}
    \end{enumerate} 
\end{definition}

\subsection{Equivalence}

Here, we adapt the arguments of \cite{CGG24}, which show the equivalence between one-time signatures and one-way puzzles in the efficient verification, to our setting in which verification may be inefficient.

\begin{theorem}
    One-time quantum digital signature schemes exist if and only if one-way puzzles exist.
\end{theorem}

\begin{proof}
    We start with the forward direction. Given a one-time signature scheme $(\skgen, \pkgen, \sign, \ver)$, we construct a one-way puzzle $(\puzzlegen, \puzzlever)$ as follows: \begin{itemize}
        \item $\puzzlegen(1^\secparam)$ computes $\sk \leftarrow \skgen(1^\secparam)$ and $\pk \leftarrow \pkgen(\sk)$. It outputs secret key $s = \sk$ and puzzle $p = \pk$.
        \item $\puzzlever(s,p)$ parses $\sk = s$ and $\pk = p$. It computes $\sigma \leftarrow \sign(\sk,m_0)$ and $b \leftarrow \ver(\pk, \sigma, m_0)$ for a fixed message $m_0$. It outputs $b$.
    \end{itemize}

    Completeness follows by the completeness of the signature scheme. To show soundness, suppose there exists $\alice$ that breaks the soundness of the one-way puzzle. We construct $(\alice_1,\alice_2)$ which breaks the one-time security of the signature scheme: \begin{itemize}
        \item $\alice_1(\pk)$ picks any $m \ne m_0$ as its signing query.
        \item $\alice_2(\pk, \sigma)$ ignores $\sigma$. It computes $s \leftarrow \alice(\pk)$ and $\sigma' \leftarrow \sign(s,m_0)$. It outputs $(\sigma', m_0)$. 
    \end{itemize}
    By assumption, $(\alice_1, \alice_2)$ breaks the soundness with the same advantage as $\alice$.

    \par Next, we show the reverse direction. Given a one-way puzzle $(\puzzlegen, \puzzlever)$ we construct a one-time signature scheme $(\skgen, \pkgen, \sign, \ver)$ for one-bit messages below. The argument can be generalized to longer messages as done by Ref. \cite{Lamport2016ConstructingDS}.
    \begin{itemize}
        \item $\skgen(1^\secparam)$ computes $(s_0,p_0) \leftarrow \puzzlegen(1^\secparam)$ and $(s_1,p_1) \leftarrow \puzzlegen(1^\secparam)$ by running $\puzzlegen$ in two independent instances. It outputs $\sk = (s_0,s_1,p_0,p_1)$.
        \item $\pkgen(\sk)$ parses $\sk = (s_0,s_1,p_0,p_1)$ and outputs $\pk = (p_0,p_1)$.
        \item $\sign(\sk, m)$ parses $\sk = (s_0,s_1,p_0,p_1)$ and outputs $\sigma = s_m$.
        \item $\ver(\pk, \sigma, m)$ parses $\pk = (p_0,p_1)$ and computes $b \leftarrow \puzzlever(\sigma, p_m)$. It outputs $b$.
    \end{itemize}

    Correctness follows by the correctness of the one-way puzzle. To show one-time security, suppose there exists $(\alice_1, \alice_2)$ that breaks it. We construct $\alice$ which breaks the soundness of the one-way puzzle: \begin{itemize}
        \item $\alice$ receives a puzzle $p$. It samples a new key-puzzle pair $(s',p')$, then sets $\pk = (p,p')$ or $\pk = (p',p)$ randomly.
        \item Then, $\alice$ computes $m \leftarrow \alice_1(\pk)$. If either $\pk = (p,p')$ and $m=1$, or $\pk=(p',p)$ and $m=0$, $\alice$ continues; otherwise, it aborts.
        \item If $\alice$ does not abort, then $s'$ is a valid signature of $m$. Thus, $\alice$ computes $(\sigma',m') \leftarrow \alice_2(\pk, s')$. Then, $\alice$ outputs $\sigma'$.
    \end{itemize}

    With non-negligible probability, $m' \ne m$ and $\sigma'$ is a valid signature on $m'$. This means that $\puzzlever(\sigma', p)$ accepts with non-negligible probability. Since, the probability that $\alice$ does not abort is $1/2$, this suffices for the proof.

    Note that in the reductions above, the verification algorithms are all inefficient, and they are exclusively run by other verification algorithms, hence not affecting the efficiency of other algorithms.

\end{proof}

\stoptocentries
\putbib[citations]
\end{bibunit}

\end{document}